\documentclass[a4paper, 12pt]{article}

\usepackage{footmisc}
\usepackage{amssymb}
\usepackage{amsmath}
\usepackage{amsthm}
\usepackage{multirow}
\usepackage{enumerate}
\usepackage{graphicx}
\usepackage{mathrsfs}
\usepackage[utf8]{inputenc}
\usepackage{cite}
\usepackage{hyperref}
\usepackage{mathtools}
\usepackage{amsfonts}
\usepackage[T1]{fontenc}
\usepackage{mathpazo}
\usepackage{bm}
\usepackage{isomath}
\usepackage{accents}
\usepackage{changepage}
\usepackage[usenames,dvipsnames,svgnames,table]{xcolor}
\usepackage [english]{babel}
\usepackage{arydshln}
\usepackage{xfrac}

\newtheorem{remark}{Remark}

\definecolor{blueviolet}{RGB}{60,50,200}
\definecolor{oliveg}{RGB}{40,200,30}
\hypersetup{colorlinks=true,       
    linkcolor=blueviolet,
    citecolor=oliveg,
}

\usepackage{tikz}
\usetikzlibrary{calc}

\usepackage{algorithm}
\usepackage{algcompatible}

\usepackage{color}

\usepackage{todonotes}

\theoremstyle{definition}
\newtheorem{definition}{Definition}[section]
\newtheorem{problem}{Problem}
\theoremstyle{plain}
\newtheorem{lemma}{Lemma}[section]
\newtheorem{theorem}{Theorem}

\newtheorem{corollary}{Corollary}[section]
\newtheorem{claim}{Claim}[section]

\newtheorem{proposition}{Proposition}[section]
\newtheorem{conjecture}{Conjecture}[section]
\newtheorem{assumption}{Assumption}[section]

\newcommand{\opPartials}[1]{\bm{\partial}^{={#1}}}

\newcommand{\phiupdown}{\Phi^\land}
\newcommand{\phidownup}{\Phi^\lor}

\newcommand{\Id}{\textnormal{Id}}
\newcommand{\gap}{\textnormal{Gap}}

\newcommand{\diag}{\textnormal{diag}}
\newcommand{\dist}{\textnormal{dist}}

\newcommand{\CAL}[1]{\mathcal{#1}}

\newcommand{\N}{\mathbb{N}}

\newcommand{\Z}{\mathbb{Z}}

\newcommand{\cD}{\mathcal{D}}
\newcommand{\cN}{\mathcal{N}}
\newcommand{\C}{\mathbb{C}}
\newcommand{\R}{\mathbb{R}}
\newcommand{\E}{\displaystyle \mathop{\mathbb{E}}}

\newcommand{\Proj}{\text{Proj}}

\newcommand{\vecOp}{\textnormal{vec}}

\newcommand{\poly}{\textnormal{poly}}

\newcommand{\rank}{\textnormal{rank}}

\newcommand{\adjfull}[3]{\textnormal{Adj}_{#2, #3}(#1)}
\newcommand{\adj}{\textnormal{Adj}}
\newcommand{\tadj}{\widetilde{\textnormal{Adj}}}
\newcommand{\adjmap}{\mathfrak{A}}
\newcommand{\tadjmap}{\tilde{\mathfrak{A}}}
\newcommand{\adjker}{\mathfrak{K}}
\newcommand{\tadjker}{\tilde{\mathfrak{K}}}

\newcommand{\norm}[1]{\left\lVert#1\right\rVert}
\newcommand{\lnorm}[1]{\left\lVert#1\right\rVert_2}
\newcommand{\fnorm}[1]{\left\lVert#1\right\rVert_F}

\newcommand{\veca}{{\mathbf{a}}}
\newcommand{\vecb}{{\mathbf{b}}}
\newcommand{\vecc}{{\mathbf{c}}}
\newcommand{\vecd}{{\mathbf{d}}}
\newcommand{\vece}{{\mathbf{e}}}
\newcommand{\vecf}{{\mathbf{f}}}
\newcommand{\vech}{{\mathbf{h}}}

\newcommand{\vecx}{{\mathbf{x}}}
\newcommand{\vecy}{{\mathbf{y}}}
\newcommand{\vecz}{{\mathbf{z}}}
\newcommand{\vecu}{{\mathbf{u}}}
\newcommand{\vecv}{{\mathbf{v}}}
\newcommand{\vecw}{{\mathbf{w}}}

\newcommand{\vecs}{{\mathbf{s}}}

\newcommand{\vecU}{{\mathbf{U}}}
\newcommand{\vecV}{{\mathbf{V}}}
\newcommand{\vecW}{{\mathbf{W}}}

\newcommand{\veclambda}{\boldsymbol{\lambda}}
\newcommand{\vecalpha}{{\boldsymbol{\alpha}}}
\newcommand{\vecbeta}{{\boldsymbol{\beta}}}
\newcommand{\vecgamma}{{\boldsymbol{\gamma}}}

\newcommand{\Lin}{\mathrm{Lin}}

\newcommand{\eqdef}{\stackrel{\text{def}}{=}}

\newcommand{\tf}{\widetilde{f}}

\newcommand{\tS}{\widetilde{S}}
\newcommand{\tT}{\widetilde{T}}
\newcommand{\tA}{\widetilde{A}}
\newcommand{\tB}{\widetilde{B}}
\newcommand{\tU}{\widetilde{U}}
\newcommand{\tV}{\widetilde{V}}
\newcommand{\tW}{\widetilde{W}}

\newcommand{\tE}{\widetilde{E}}

\newcommand{\tP}{\widetilde{P}}

\newcommand{\tM}{\widetilde{M}}
\newcommand{\tN}{\widetilde{N}}

\newcommand{\tvecU}{\tilde{\mathbf{U}}}

\newcommand{\tvecW}{\tilde{\mathbf{W}}}
\newcommand{\tveca}{\tilde{\veca}}
\newcommand{\tvecv}{\tilde{\vecv}}

\newcommand{\opB}{\mathcal{B}}
\newcommand{\opC}{\mathcal{C}}
\newcommand{\opL}{\mathcal{L}}

\newcommand{\topB}{\widetilde{\mathcal{B}}}
\newcommand{\topC}{\widetilde{\mathcal{C}}}

\newcommand{\IP}[3]{ {\left\langle #1, #2 \right\rangle}_{#3} }

\newcommand{\inparen }[1]{\left(#1\right)}             
\newcommand{\inbrace }[1]{\left\{#1\right\}}           
\newcommand{\insquare}[1]{\left[#1\right]}             
\newcommand{\inangle }[1]{\left\langle#1\right\rangle} 
\newcommand{\inabs}[1]{\left|#1\right|}
\newcommand{\setdef}[2]{\inbrace{{#1}\ : \ {#2}}}      
\mathchardef\mhyphen="2D

\newcommand{\orth}[1]{{#1}^{\perp}}

\newcommand{\tensored}[2]{{#1}^{\otimes #2}}

\makeatletter
\newcommand\footnoteref[1]{\protected@xdef\@thefnmark{\ref{#1}}\@footnotemark}
\makeatother

\definecolor{deepjunglegreen}{rgb}{0.0, 0.29, 0.29}

\newcommand{\eps}{\epsilon}
\newcommand{\lb}{\lambda}
\newcommand{\tlb}{\tilde{\lambda}}

\newcommand\restr[2]{{
  \left.\kern-\nulldelimiterspace 
  #1 
  \littletaller 
  \right|_{#2} 
}}

\newcommand{\littletaller}{\mathchoice{\vphantom{\big|}}{}{}{}}

\newcommand{\offdiag}{{\textnormal{off-diag}}}

\newcommand{\veco}{\boldsymbol{\omega}}
\newcommand{\ubar}[1]{\underaccent{\bar}{#1}}
\newcommand{\upL}{\Bar{L}}
\newcommand{\downL}{\ubar{L}}

\newcommand{\IPP}[2]{ {\left\langle #1, #2 \right\rangle} }
\newcommand{\IPT}[2]{ {\left\langle #1, #2 \right\rangle}_{\tau} }
\newcommand{\Tr}{\mathrm{Tr} \enspace}
\newcommand{\D}{\partial}

\newcommand\Vector[1]{\setstackEOL{,}\bracketVectorstack{#1}}

\newcommand{\Exp}{\mathop{\mathbb{E}}}

\newcommand{\en}{\enspace}

\usepackage[a4paper, margin=1in]{geometry}
\usepackage[english]{babel}
\usepackage{url}            
\usepackage{booktabs}       
\usepackage{nicefrac}       
\usepackage{enumitem}
\usepackage{stackengine}    
\usepackage{comment}        

\begin{document}
\title{Learning Arithmetic Formulas in the Presence of Noise: A General Framework and Applications to Unsupervised Learning}

\author{{Pritam Chandra} \\
\normalsize{Microsoft Research} \\
\normalsize{\tt{t-pchandra@microsoft.com}}
   \and Ankit Garg \\
   \normalsize{Microsoft Research} \\
    \normalsize{\tt{garga@microsoft.com}}
    \and Neeraj Kayal \\
    \normalsize{Microsoft Research} \\
    \normalsize{\tt{neeraka@microsoft.com}}
    \and Kunal Mittal \\
    \normalsize{Princeton University} \\
    \normalsize{\tt{kmittal@cs.princeton.edu}}
    \and Tanmay Sinha \\
    \normalsize{Microsoft Research} \\
    \normalsize{\tt{t-tsinha@microsoft.com}}
}
\maketitle
\begin{abstract}

   We present a general framework for designing efficient algorithms for unsupervised learning problems, such as mixtures of Gaussians and subspace clustering. Our framework is based on a meta algorithm that learns arithmetic formulas in the presence of noise, using lower bounds. This builds upon the recent work of Garg, Kayal and Saha (FOCS '20), who designed such a framework for learning arithmetic formulas without any noise. A key ingredient of our meta algorithm is an efficient algorithm for a novel problem called \emph{Robust Vector Space Decomposition}. We show that our meta algorithm works well when certain matrices have sufficiently large smallest non-zero singular values. We conjecture that this condition holds for smoothed instances of our problems, and thus our framework would yield efficient algorithms for these problems in the smoothed setting.
\end{abstract}

\pagenumbering{gobble}

\newpage


\pagenumbering{roman}
\tableofcontents
\newpage

\pagenumbering{arabic}

\section{Introduction}\label{sec:intro} 
Unsupervised learning involves discovering hidden patterns and structure in data without using any labels or direct human supervision. 
Here we consider data that has a nice mathematical structure or is generated from a mathematically well-defined distribution. An example of the former is when the data points can be grouped into meaningful clusters based on some similarity patterns and the goal is to find the underlying clusters. An example of the latter is mixture modeling, which assumes that the data is generated from a mixture of succinctly described probability distributions, such as Gaussian distributions, and the goal is to learn the parameters of these distributions from samples. A general framework for solving many unsupervised learning problems is the method of moments, which leverages the statistical moments\footnote{
    Recall that moments are measures of the shape and variability of a data set. They are used to describe the the location and dispersion of the data. When the dataset consists of a collection of points 
        $$ A = \{ \veca_i = (a_{i1}, a_{i2}, \ldots, a_{in} ) \in \R^n \quad | \quad i \in [N] \}, $$
    some examples of (low-order) moments are $\E_{\veca_i \in A} [ a_{i1}]$, $\E_{\veca_i \in A} [ a_{i1} \cdot a_{i2}]$, etc.
} of the data to infer the underlying structure or the underlying parameters of the model. For many unsupervised learning problem scenarios wherein the underlying data has some nice mathematical structure, the moments of the data are well-defined functions of the parameters. Heuristic arguments then suggest that the converse should typically hold, i.e. the parameters of the structure/distribution are {\em typically} uniquely determined by a few low order moments of the data. In this broad direction, the main challenge then is to design  algorithms to (approximately) recover the underlying parameters from the (empirical) moments\footnote{
    In scenarios where the data is a finite sample drawn from a distribution $\mathcal{D}$ over $\R^n$, the empirical moments (which can be very easily and efficiently computed) are estimates of, but {\em not} equal to, the true underlying moments. 
}.  
We further want the algorithm to be {\em efficient}, noise-tolerant (i.e. work well even when the moments are known only approximately rather than exactly) and are even {\em outlier-tolerant} (i.e. work well even when a few data points do not conform to the underlying structure/distribution). But even the simplest problems in this area tend to be NP-hard and remain so even when there is no noise and no outliers. So one cannot realistically hope for an algorithm with provable worst-case guarantees. But what one can hope are algorithms that are guaranteed to typically work well, i.e. either for {\em random} problem instances or even more desirably for instances chosen in a smoothed fashion. Accordingly, many different algorithms have been designed for each such problem in unsupervised learning with varying levels of efficiency, noise-tolerance, outlier-tolerance and provable guarantees. 
In this work we give a {\em single meta-algorithm} that applies to many such unsupervised learning problems. The starting point of our work is the observation that many such problems reduce to the task of learning an appropriate subclass of arithmetic formulas. \\

\noindent {\bf Connecting unsupervised learning to arithmetic complexity.}
We now give a few more details of how such a reduction works for the setting in which the data points are drawn from a distribution having a nice mathematical structure. Let $\mathcal{D}$ be a distribution over points in $\R^n$. We introduce $n$ formal variables $(x_1, x_2, \ldots, x_n)$ and denote it as $\vecx$. For a suitably chosen integer $d \geq 2$, form a degree-$d$ polynomial $f(\vecx)$ which encodes the $d$-th order moments\footnote{
    We are making the mild assumption here that the $d$-th order moments of $\mathcal{D}$ are bounded.
} of the distribution in some suitable way. For example, in some applications the coefficient of a monomial of $f(\vecx)$ is simply (a canonically scaled version of) the corresponding moment. At this point, such a formal polynomial is a mere bookkeeping device for the $d$-th order moments of $\mathcal{D}$. For many nice, well-structured distributions such as mixtures of Gaussians, however this polynomial (or variants thereof) turns out to have a remarkable property - it can be computed/represented by a small arithmetic formula! Special cases of this remarkable phenomenon were noted earlier when it was observed that many problems in (unsupervised) learning reduce to the problem of learning set-multilinear depth-three formulas, better known as tensor decomposition. Such connection(s) inspired a whole body of work on tensor decomposition with applications including independent component analysis, learning Hidden Markov Models, learning special cases of mixtures of Gaussians, latent Dirichlet allocation, dictionary learning, etc. (cf. the surveys \cite{vijayraghavan20, koldaB09, aghkt14}). \\

\noindent {\bf Noise-tolerance.}
Notice however that we are given a finite set of points sampled from the distribution $\mathcal{D}$, so we do not have the ($d$-th order) moments of $\mathcal{D}$ exactly but only approximately. Thus for such applications we need the algorithm for learning arithmetic formulas to also be {\em noise-tolerant}, i.e. given a polynomial 
$\tilde{f}(\vecx)$ that is {\em close to}\footnote{
    Under a natural notion of distance between a pair of polynomials akin to Euclidean distance between the coefficient vectors of the pair of polynomials - see section \ref{sec:prelims}.
} a polynomial $f(\vecx)$ that has a small arithmetic formula $\phi$, we want to learn/reconstruct a arithmetic formula $\tilde{\phi}$ from the same subclass as that of $\phi$ whose output polynomial is close to $\tilde{f}(\vecx)$ (and therefore to $f(\vecx)$ as well). Recently, \cite{GKS20} gave a meta-algorithm for learning many different subclasses of formulas including the ones relevant for unsupervised learning (assuming that certain nondegeneracy conditions hold). But it has one important shortcoming that was also pointed out in \cite{bdjkkv22}: the techniques of \cite{GKS20} were algebraic and it was unclear if they could handle noise arising out of the fact that the moments are known only approximately and not exactly. Qualitatively, our main result builds upon and suitably adapts the algorithm \cite{GKS20} to make it noise-tolerant. Quantitatively, in the noisy setting, we provide bounds on the quality of the output of our algorithm that depend on singular values of certain matrices that underlie the algorithm. We expect that for most applications, the relevant singular values would be well-behaved for random instances and maybe even for smoothed/perturbed worst-case instances. If so, our algorithm would work and yield good quality outputs on such instances. Accordingly, we then go on on to analyze the singular values of the relevant matrices pertaining to subspace clustering\footnote{ 
    A recent work \cite{bafna2022polynomial} analyzed the singular values of matrices arising in a related (but also different) algorithm that was tailor-made for the mixtures of (zero-mean) Gaussians and verified that for random instances the singular values are indeed well-behaved.}. 
We also expect (suitable adaptations of) our algorithm to be tolerant to the presence of a few outliers but we do not pursue this direction here and leave it for future work. \\

\noindent {\bf Illustrative example - mixtures of Gaussians.}
Let us make the above discussion concrete via the example of learning mixtures of Gaussians which in itself is a very well-studied problem with history going back to more than a hundred years. 
Suppose we are given a dataset consisting of a finite set of points $A \subset \R^n $
    \begin{equation}\label{eqn:ptSet}
        A = \{ \veca_i = (a_{i1}, a_{i2}, \ldots, a_{in} ) \in \R^n \quad | \quad i \in [N] \}.
    \end{equation}
The points are drawn independently at random from an unknown mixture of $s$ Gaussians 
    $ \mathcal{D} := \sum_{i=1}^s w_i \mathcal{N}(\bm{\mu}_i, \Sigma_i),$ 
which means that the $i$-th component of the mixture has weight\footnote{
    The weights satisfy  $\sum_{i \in [s]} w_i =1 $. 
} $w_i \in [0, 1]$, mean $\bm{\mu}_i \in \R^n$ and covariance matrix $\Sigma_i \in \R^{n \times n}$. Our goal is to estimate the parameters $w_i$ and $\bm{\mu}_i$ and $\Sigma_i$ ($i \in [s]$) from the given samples/data $A$. Let $\vecx = (x_1, x_2, \ldots, x_n)$ be a tuple of formal variables and consider the polynomial $f(\vecx) := \mathbb{E}_{\veca \sim \mathcal{D}} \left[ \langle \vecx, \veca \rangle^d \right]$. It is (a scalar multiple of) a {\em slice of} the formal moment generating function defined as 
    $ \Exp_{\veca \sim \mathcal{D}} \left[ \exp({ \langle \vecx, \veca \rangle }) \right]. $
Notice that the coefficients of a given monomial (over $\vecx$) in $f(\vecx)$ equals the corresponding moment of the distribution (upto some canonical scaling). Then in this case, $f(\vecx)$ has the following small formula\footnote{
    This formula for $f(\vecx)$ can be inferred from the fact that for a single Gaussian distribution $N(\bm{\mu}, \Sigma)$, its moment generating function is in fact equal to $\exp(\vecx^{T} \cdot \bm{\mu} + \frac{1}{2} \vecx^{T} \cdot \Sigma \cdot \vecx ) $.   
}:
    $$ f(\vecx) = \sum_{i \in [s]} w_i G_d(\ell_i(\vecx), Q_i(\vecx)),$$ 
where $\ell_i(\vecx) := \langle \bm{\mu}_i, \vecx \rangle$, $Q_i(\vecx) := \frac{1}{2} \vecx^T \Sigma_i \vecx$ and $G_d$ is a fixed bivariate polynomial depending on $d$. In the zero-mean case (i.e. when $\bm{\mu}_1 = \bm{\mu}_2 = \ldots = \bm{\mu}_s = \bm{0}$), the formula for $f(\vecx)$ is 
    $$ f(\vecx) = \sum_{i \in [s]} \frac{d!}{(d/2)!} w_i Q_i(\vecx)^{d/2} $$
when $d$ is even (and $0$ if $d$ is odd). In this way, if the sample size was infinite (or equivalently that if we knew the true moments of the distribution), learning mixtures of Gaussians would reduce to the problem of learning/reconstructing the subclass of arithmetic formulas indicated by the rhs of the above expression for $f(\vecx)$. But we don't have access to the exact moments. Using the empirical moments, we can get hold of an approximate version of $f$, 
    $$ \tf(\vecx) := \mathbb{E}_{\veca \sim A} \left[ \langle \vecx, \veca \rangle^d \right] = \frac{1}{N} \sum_{i \in [N]} \left[ \langle \vecx, \veca_i \rangle^d \right].  $$
We will have $\tf(\vecx) = f(\vecx) + \eta(\vecx)$ for a noise polynomial $\eta(\vecx)$ whose magnitude will be inversely proportional to square root of the number of samples $N$. In this way, learning mixtures of Gaussians 
reduces to the problem of reconstructing the indicated subclass of arithmetic formulas in the presence of noise. \\

\noindent {\bf Learning arithmetic formulas in the presence of noise - problem formulation.}
The above discussion motivates us to consider the problem of learning (arbitrary subclasses of) arithmetic formulas in the presence of noise. In many practical settings the output gate of the underlying formula is a 
(generalized\footnote{
    A generalized addition gate can compute any fixed linear combination of its inputs.
}) addition gate so that the problem can be formulated as follows. We are given a polynomial $\tf(\vecx)$ of the form
    $\tf(\vecx) := T_1(\vecx) + \cdots + T_s(\vecx) + \eta(\vecx),$ 
for \emph{structured} polynomials $T_i(\vecx)$'s and a noise polynomial $\eta(\vecx)$. Our goal is to approximately recover each summand $T_i(\vecx) $. For example, for the case of mixture of spherical Gaussians we would have $T_i(\vecx) = w_i \cdot \langle \bm{\mu}_i, \vecx \rangle^3$ (see Remark~\ref{rmk:potential}). For the case of mixture of zero-mean Gaussians we would have $T_i(\vecx) = w_i \cdot Q_i(\vecx)^{d/2} $ and so on. In the noiseless setting, i.e. when $\eta(\vecx) = 0$, the paper \cite{GKS20} designed a meta-algorithm applicable to learning many interesting subclasses using a general framework exploiting lower bound techniques in arithmetic complexity theory. The algorithm worked under certain relatively mild non-degeneracy assumptions. However, their algorithm had some algebraic components and it was not clear how to design an algorithm in the noisy case when the noise polynomial $\eta(\vecx)$ is non-zero.\footnote{In most settings, one would like the running time of the algorithm to be inverse polynomial in the magnitude of the noise, to have a polynomial dependence on the number of samples in the final learning problem.}  Our main contribution is that we show how to modify the general framework in \cite{GKS20} to the noisy setting. We also show how to use this framework to design efficient algorithms for two well studied problems in unsupervised learning: mixtures of (zero mean) Gaussians and subspace clustering.

\begin{remark}\label{rmk:potential}
    \begin{enumerate}
        \item [(a).] {\bf Simpler reductions.} The ability to handle arbitrary subclasses of arithmetic formulas not only yields a common (meta) algorithm that applies to a wide variety of problems in unsupervised learning but it also often makes the reductions simpler. For example, in the discussion above the reduction of learning mixtures of arbitrary Gaussians to learning the appropriate subclass of arithmetic formulas is perhaps simpler than the reduction of learning mixtures of spherical Gaussians\footnote{
            A spherical Gaussian is one where the covariance matrix $\Sigma_i$ is the identity matrix.
        } to tensor decomposition. We sketch this reduction now. Consider
        $f(\vecx) := \mathbb{E}_{\veca \sim \mathcal{D}} \left[ \langle \vecx, \veca \rangle^3 - 3 \left( \sum_{i \in [n]} x_i^2 \right) \cdot \langle \vecx, \veca \rangle \right]$. When $\mathcal{D}$ is a mixture of spherical Gaussians, expanding and simplifying this expression, we can get that $f(\vecx) = \sum_{i=1}^s w_i \langle \bm{\mu}_i, \vecx\rangle^3$. 

        \item [(b).] {\bf Mixtures of general Gaussians.} We expect that our algorithm can be extended to general mixtures of Gaussians (different means and/or covariance matrices) but its analysis will likely get much more cumbersome, so we avoid this more general case for the sake of simplicity.
        
        \item [(c).] {\bf Handling outliers.} The ability to handle arbitrary subclasses of arithmetic formulas can also allow the algorithm to be tolerant to the presence of outliers. To see this, consider the case of zero-mean Gaussians and suppose that the given set of data points $A$ contains a subset $\hat{A} \subset A$ of outliers of size $\hat{N} \ll N$. In that case the empirical moment polynomial $\tf(\vecx)$ would have the following structure:
            $$ \tf(\vecx) = \frac{N - \hat{N}}{N} \cdot \frac{d!}{(d/2)!}  \cdot \left( \sum_{i \in [s]} w_i Q_i(\vecx)^{d/2} \right) + \frac{1}{N} \cdot \left( \sum_{\veca_j \in \hat{A}} (\vecx \cdot \veca_j)^{d} \right) + \eta(\vecx). $$
        We expect that our algorithm can be adapted to learn the class of formulas corresponding to the right side of the above expression however the analysis of such an algorithm can get cumbersome. For the sake of keeping the length of this paper to within reasonable bounds, we do not do the analysis of the outlier tolerance of our algorithm. 

        \item [(d).] {\bf Other mixtures models.} The connection between learning mixtures of Gaussians and learning an appropriate subclass of arithmetic formulas arose out of the fact that (any slice of) the moment generating function of a multivariate Gaussian has a simple algebraic expression. For some other 
        distributions also the (slices of) moment generating function or some other related function like the cumulant generating function or the characteristic function have a nice algebraic expression and we can expect our approach to be applicable for such mixtures also. 

        \item [(e).] {\bf Mixtures of structured point sets and those sampled from probability distributions.}
        Consider a set of points $A \subset \R^n$ that can be partitioned into two subsets 
        $A = A_1 \uplus A_2 $ such that $A_1$ is some structured set of points (such as being contained in the union of a small number of low-dimensional subspaces for example) and $A_2$ is chosen from some mixture model (such as being chosen from a mixture of Gaussians for example). When say 
        the moment polynomials of both the structured set $A_1$ and the sampled set $A_2$ admit small 
        formulas from a tractable subclass of formulas, we can expect our methods to apply. In particular, 
        we expect (a suitable adaptation of) our algorithm to be to handle the case where points in $A_1$ are chosen from a union of low-dimensional subspaces in an non-degenerate way without conforming to any nice distribution and points in $A_2$ conform to a (mixture of) Gaussians. We leave the task of handling such mixed datasets and analyzing the relevant algorithms as a possible direction for future work.

        \item [(f).] {\bf Potential application - Topic Modeling.}  It turns out that there are some other problems in unsupervised learning which reduce to robustly learning an appropriate subclass of arithmetic formulas. We expect that a suitable instantiation/adaptation of our algorithm should apply for these applications but we do not pursue these applications here and leave it as a direction for future work.    
        One such problem is called topic modelling. It is known that learning some simple topic models reduce to tensor decomposition. It turns out that learning some general topic models as proposed in \cite{wallach06} reduce to the problem of learning set-multilinear formulas of larger depth. 

        \item [(g).] {\bf Potential application - Learning (Mixtures of) Polynomial Transformations.}
        Another such application is the problem of learning polynomial transformations as studied in \cite{cllz23} which also reduces to learning a certain subclass of arithmetic formulas\footnote{
            The work of \cite{cllz23} does not state it this way but this can be inferred from the observations 
            underlying their work. 
        }. The generality of our approach makes us expect that it should apply to this task also as well as to its generalizations like learning {\bf mixtures of} polynomial transformations. We do not pursue this potential application here but leave it as a direction for future work.
    
    \end{enumerate}
  
\end{remark}

\subsection{Overview - Arithmetic Formula Learning algorithm.}\label{sec:ov_ckts}
{\bf Background.} 
Arithmetic formulas are a natural model of computing polynomials using the basic operations of addition ($+$) and multiplication ($\times$). A natural problem about arithmetic formulas is that of learning: given a polynomial $f(\vecx)$\footnote{There are various input models all of which lead to interesting questions. Some of the common ones are as a black box or described explicitly as a list of coefficients.}, find the smallest (or somewhat small) arithmetic formula computing $f(\vecx)$. We consider formulas in their alternating normal form: i.e. the formula consists of alternating layers of addition and multiplication gates. The learning problem boils down to recovering the polynomials computed at each child of a node $v$ given the polynomial computed at $v$. When $v$ is a multiplication node then generically, the polynomials computed at its children are   irreducible\footnote{
    Random multivariate polynomials are almost surely irreducible and with that as intuition, one expects the output of a formula with output being an addition gate to almost surely be an irreducible polynomial when the underlying field constants are chosen randomly. However proving this can be technically involved for any given subclass of formulas. 
} in which case the efficient multivariate polynomial factorization algorithm of Kaltofen and Trager \cite{kT90} recovers the children's outputs. Even when there is noise, the robust factorization algorithm of \cite{kaltofenMYZ08} can recover the factors approximately\footnote{
    The work of \cite{kT90} aims to devise a factorization algorithm that is empirically as robust as possible and does not contain theoretical bounds on how much the output factors get perturbed as a function of the noise added to a true factorization. Nevertheless such a bound can be inferred from their work. The bound would depend on the appropriate singular values of an instance-dependent matrix called the Ruppert matrix that comes up in their algorithm.       
}. Thus the main challenge is to recover the children of addition gates. This connects us to the problem discussed in the previous section with the \emph{structured} polynomials being the polynomials computed at the children gates. In the noiseless setting, a meta algorithm for this problem was given in \cite{GKS20}. We provide a meta algorithm in the noisy case and show worst-case bounds on the quality of the output in terms of singular values of certain matrices \footnote{The matrices whose singular values are used to bound the quality of the output depend on the input instance as well as on the choice of linear operators used to instantiate our framework}. The abstract problem is as follows. Given a polynomial $\tf(\vecx)$ that can be expressed as 
    \begin{equation}\label{eqn:sosp1}
       \tf(\vecx) = T_1(\vecx) + T_2(\vecx) + \ldots + T_s(\vecx) + \eta(\vecx), 
    \end{equation}
where $T_i$'s are {\em structured} polynomials and the noise/perturbation polynomial $\eta(\vecx)$ has {\em small norm}, can we approximately recover the $T_i$'s via an efficient algorithm? \\

\noindent{\bf Learning from lower bounds.} \cite{GKS20} showed how the linear maps used in the known arithmetic formula lower bound proofs could be used to recover the $T_i$'s in the noiseless ($\eta = 0$) setting, {\em assuming that appropriate non-degeneracy conditions hold}. \cite{GKS20} observed that the assumption that the $T_i$'s are structured can effectively be operationalized via the existence of a known set of linear maps $\opL$ from the vector space of polynomials to some appropriate vector space $W_1$ such that $\dim (\langle \opL \cdot T_i \rangle)$ is\footnote{
    Here, $\langle S \rangle$ denotes the $\R$-linear span of a set $S$ that consists of vectors or linear maps. Also $\opL \cdot T_i$ denotes the set of vectors obtained by applying each linear map in $\opL$ to $T_i$.
} {\em small} for every simple polynomial $T_i$. When we apply such a set of linear operator $\opL$ to (\ref{eqn:sosp1}) with $\eta = 0$, we get:
    \begin{equation}\label{eqn:sosp2}
       \inangle{\opL \cdot \tf(\vecx)} \subseteq  \inangle{\opL \cdot T_1(\vecx)} + \inangle{\opL \cdot T_2(\vecx)} + \ldots + \inangle{\opL \cdot T_s(\vecx)}.  
    \end{equation}
\cite{GKS20} observe that generically two things tend to happen. 
\begin{assumption}
    {\bf First blessing of dimensionality\footnote{
        The intuition is that (pseudo)-randomly chosen small-dimensional subspaces of a large-dimensional ambient space should form a direct sum.}.}
    If $\sum_{i \in [s]} \dim(\inangle{\opL \cdot T_i(\vecx)}) \ll \dim(W_1)$ then almost surely (over the independent random choice of the $T_i$'s), it holds that the subspaces $\inangle{\opL \cdot T_i(\vecx)}$ form a direct sum, i.e.
        $$ \dim(\inangle{\opL \cdot T_1(\vecx)} + \inangle{\opL \cdot T_2(\vecx)} + \ldots + \inangle{\opL \cdot T_s(\vecx)}) = \sum_{i \in [s]} \dim(\inangle{\opL \cdot T_i(\vecx)}). $$
\end{assumption}

\begin{assumption}
    {\bf Second blessing of dimensionality\footnote{
        The intuition is that in most applications when the underlying dimension $n = |\vecx|$ is large enough then the dimension of the set of operators $\opL$ is large relative to $\dim(\inangle{\opL \cdot T_i(\vecx)})$ for any $i$. In such a situation if the $T_i$'s are chosen generically then $\opL$ tends to contain many operators that kill all the other $T_j$'s (for $j \neq i$) so that $\inangle{\opL \cdot f(\vecx)}$ tends to contain each of the subspaces $\inangle{\opL \cdot T_i(\vecx)}$.
    }.}
    If $\sum_{i \in [s]} \dim(\inangle{\opL \cdot T_i(\vecx)}) \ll \dim(\inangle{\opL})$ then almost surely (over the independent random choice of the $T_i$'s), it holds that for all $i \in [s]$:
        $$ \inangle{\opL \cdot (T_1(\vecx) + \ldots + T_s(\vecx))} \supseteq \inangle{\opL \cdot T_i(\vecx)}. $$
\end{assumption}

\noindent Under these nondegeneracy assumptions we then have (for $\eta = 0$): 
    \begin{equation}\label{eqn:direct_sum}
       U \eqdef \inangle{\opL \cdot \tf(\vecx)}  =  \inangle{\opL \cdot T_1(\vecx)} \oplus \inangle{\opL \cdot T_2(\vecx)} \oplus \ldots \oplus \inangle{\opL \cdot T_s(\vecx)}.  
    \end{equation}
We observe that in the noisy case, on input $\tf$, finding the best $(\dim(U))$-rank subspace through the set of points $(\opL \circ \tf )$ yields a subspace $\tU$ that is pretty close to $U$ (lemma \ref{lemma_ckt_recon_tUdistU} gives quantitative bounds). Coming back to the noiseless case, \cite{GKS20} then observe that linear maps constructed for the purpose of proving lower bounds also yield a set of linear maps $\opB$ such that 
    \begin{align}
        V \eqdef \langle \opB \cdot U \rangle = \langle \opB \cdot U_1 \rangle \oplus \cdots \oplus \langle \opB \cdot U_s \rangle, \label{eqn:direct_sum2}
    \end{align}
where $U_i \eqdef  \inangle{\opL \cdot T_i(\vecx)}$. This motivated the following problem which they call \emph{Vector Space Decomposition}. Given a set of linear maps $\opB$ between two vector spaces $U$ and $V$, find a (maximal) decomposition $U = U_1 \oplus \cdots \oplus U_s$, $V = V_1 \oplus \cdots \oplus V_s$ s.t. $\opB \cdot U_i \subseteq V_i$ for all $i \in [s]$. In most applications, such a decomposition turns out to be unique (up to some obvious symmetries like permuting the subspaces) and hence an algorithm for vector space decomposition finds the intended decomposition. \\

\begin{sloppy}
\noindent {\bf Reduction to Vector Space Decomposition in the noisy setting.} We then formulate and give an algorithm for a robust/noise-tolerant version of vector space decomposition. But there is an important difficulty that crops up in trying to use the problem of robust vector space decomposition as formulated below to the setting of learning arithmetic formulas in the presence of noise. $\opB$ is a collection of maps from $W_1$ to $W_2$ where $U \subseteq W_1, V \subseteq W_2$ and $\opB \cdot U$ equals $V$. However, $\opB \cdot \tU$ will typically {\em not} be contained in $\tV$. In fact the dimension of the image of $\tU$ under the action of $\opB$ (denoted  $\dim\left(\inangle{\opB \cdot \tU}\right)$) will typically be much larger than the dimension of $\tV$  (denoted $\dim(\tV)$). To overcome this difficulty, our idea is to compose maps in $\opB$ with the projection\footnote{
    Projection to $\tV$ here implicitly uses a decomposition of the ambient space $W_2$ into $\tV$ and its orthogonal complement (defined via some canonical inner product on $W_2$ that is clear from context). It is the unique map in $\Lin(W_2, W_2)$ which is identity on $\tV$ and whose kernel is the the orthogonal complement of $\tV$.} 
map to $\tV$ to obtain a tuple of maps $\topB$ from $\tU$ to $\tV$. In general such a composition can completely spoil the structure of the set of maps $\opB$ but our conceptual insight here is that in this situation, one can set up a natural correspondence between $\Lin(U, V)$ and $\Lin(\tU, \tV)$ that can be used to infer that the projection-composed maps $\topB$ {\em are slight perturbations} of the corresponding maps in $\opB$ (lemma \ref{lemma:restr_op_close} gives quantitative bounds). This insight gives us the reduction. Then, the robust vector space decomposition algorithm yields a decomposition 
    \begin{equation}\label{eqn:tudecomp}
        \tU = \tU_1 \oplus \tU_2 \oplus \ldots \oplus \tU_s        
    \end{equation}
where $\tU_1, \tU_2, \ldots, \tU_s$ are slightly perturbed versions of $\inangle{\opL \cdot T_1(\vecx)}, \inangle{\opL \cdot T_2(\vecx)}, \ldots, \inangle{\opL \cdot T_s(\vecx)}$ respectively (corollary \ref{corr:rvsd_common_op} gives quantitative bounds). In particular this implies that for each $L \in \opL$ we can obtain a vector close to $L \cdot T_1(\vecx) $ by {\em projecting}\footnote{
    Projection to $\tU_1$ here refers to using the decomposition given by (\ref{eqn:tudecomp}). It is applying the unique map in $\Lin(\tU, \tU)$ which is identity on $\tU_1$ and whose kernel is $(\tU_2 \oplus \tU_3 \oplus \ldots \oplus \tU_s)$. 
} $L \cdot \tf(\vecx)$ to $\tU_1$. This implies that we can approximately recover $T_1(\vecx)$ itself via an appropriate pseudo-inverse computation. Similarly, we can recover all the $T_i(\vecx)$'s up to some error (Theorem \ref{thm:robustCircuitReconstruction} gives quantitative bounds). Before stating the quantitative bound on this error (Theorem~\ref{thm:IntroductionCircuitReconstruction}) let us discuss the subroutine of robust vector space decomposition which is perhaps of interest in itself and might have wider applicability.    
    
\end{sloppy}

\subsection{Overview - Vector Space Decomposition algorithm}
We refer to the noise-tolerant version of vector space decomposition as {\em Robust Vector Space Decomposition (RVSD)}. The setting is the following: let $W_1$ and $W_2$ be vector spaces, and let $U = U_1\oplus\dots\oplus U_s \subseteq W_1$ and $V = V_1\oplus\dots\oplus V_s \subseteq W_2$ be subspaces.
Let $\opB = (B_1, B_2, \ldots, B_m)$ be an $m$-tuple of linear operators, with each $B_j:U\to V$ being a linear map from $U$ to $V$.
Suppose that, under the action of $\opB$, each $U_i$ is mapped inside $V_i$; that is, for each $i\in [s]$, it holds that $\inangle{\opB \cdot U_i} \subseteq V_i$. We consider the problem of recovering the $U_i$'s approximately given noisy access to $U, V$ and $\opB$. Specifically\footnote{As discussed above, it is often the case that a set of operators $(B_1,\dots,B_m)$, with each $B_i:W_1\to W_2$, satisfying the above property are exactly known. 
In this case, we can instantiate the Robust Vector Space Decomposition problem with suitable projections of these operators on the set of linear maps from $U\to V$, and $\tU \to\tV$ respectively.
For more details, the reader is referred to Section~\ref{subsec:rvsd_common_op}.}

~\\\noindent{\bf Robust Vector Space Decomposition (RVSD). }\ 
We are given as input the integer $s$, two vector spaces $\tU \subseteq W_1$ and $\tV \subseteq W_2$, and a $m$-tuple of operators $\topB = (\tB_1, \tB_2, \dots, \tB_m)$ from $\tU$ to $\tV$, such that $\dist(\tU,  U)$, $\dist(\tV, V)$ and $\dist(\topB,\opB)$\footnote{In the formulation here, the distance $\dist(\topB,\opB)$ is defined by extending all operators to map $W_1$ into $W_2$.} are "small."
Our goal is to \emph{efficiently} find an $s$-tuple $\widetilde \vecU = (\tU_1, \tU_2, \ldots, \tU_s)$ of subspaces in $\tU\subseteq W_1$, such that (upto a reordering of the components) for each $i\in [s]$, $\dist(\tU_i, U_i)$ is "small"\footnote{As we note in Remark~\ref{remark:vecV_find}, our algorithms can be used to find $(V_1,\dots,V_s)$ approximately as well, but we omit that here since our applications do not need it.}.\\

Now we give some rough ideas that go behind our Robust Vector Space Decomposition algorithm.
For more details, the reader is referred to Section~\ref{sec:rvsdAlgo}. Let us first consider the noiseless setting, in which we are given an integer $s$, the vector spaces $U\subseteq W_1, V\subseteq W_2$, and a $m$-tuple of operators $\opB=(B_1,\dots,B_m)$ from $U\to V$; the goal is to find a decomposition $U = U_1\oplus\dots\oplus U_s$ and $V = V_1\oplus\dots\oplus V_s$, such that each $U_i$ is mapped into $V_i$ under the action of $\opB$, i.e.
    \begin{equation}\label{eqn:vsdecomp}
        U = U_1 \oplus U_2 \oplus \ldots\oplus U_s \quad \text{and~} V = V_1 \oplus V_2 \oplus \ldots \oplus V_s, \quad \inangle{\opB \cdot U_i} \subseteq V_i \quad \forall i \in [s].
    \end{equation}

\noindent{\bf The adjoint algebra and its properties.}
Based on \cite{qiao, ChistovIK97}, \cite{GKS20} defined a notion called the {\em adjoint algebra}\footnote{The adjoint algebra is a generalization of the notion of the centralizer algebra in matrix/group theory to the case when the image space of the set of linear maps is different from the domain space.} whose structure can be used to understand (the potentially many) decompositions. Let us recall this notion. 
\begin{definition}\label{defn:adjointAlgebra}{\bf Adjoint algebra}
    The adjoint algebra, corresponding to the vector spaces $U, V$, and the tuple of operators $\opB$, denoted $\adjfull{\opB}{U}{V}$ is defined to be the set of all tuples of linear maps $(D,E)$, with $D:U\to U$, $E:V\to V$, such that  $B_j \cdot D = E \cdot B_j \ \text{for all~} j\in[m].$
\end{definition}

\noindent Observe that the adjoint algebra always contains the space of scaling maps\footnote{
    This observation is due to \cite{ChistovIK97} and forms the starting point of the~\cite{GKS20} algorithm for vector space decomposition.}: that is, the set of maps $D:U\to U, E:V\to V$ such that $D$ (resp. $E$) simply scales each $U_i$ (resp. $V_i$) by some scalar $\lb_i$, for each $i\in [s]$. We observe that in most applications these maps are all that the adjoint algebra contains, and in this case, there is a simple algorithm to solve the vector space decomposition, and the obtained decomposition is unique:

\begin{proposition} [Proposition A.3\footnote{
    This proposition is a special case of the more general proposition A.3 in \cite{GKS20} wherein 
    the blocks of $\adjfull{\opB}{U}{V}$ consist of scalar matrices only.} in \cite{GKS20}]\label{prop:adjointScalingUniqueness}
    Suppose that $U, V$ admit a decomposition into direct sum of $s$ spaces under the action of $\opB$ as in (\ref{eqn:vsdecomp}). 
    If $\dim(\adjfull{\opB}{U}{V})=s$,  then it holds that:
        \begin{enumerate}
            \item $\adjfull{\opB}{U}{V}$ equals the set of scaling maps (as defined above) and,

            \item The decomposition given by (\ref{eqn:vsdecomp}) is the unique irreducible decomposition, i.e. if 
            $$ U = \hat{U}_1 \oplus \hat{U}_2 \oplus \ldots\oplus \hat{U}_{\hat{s}} \enspace \text{and~} V = \hat{V}_1 \oplus \hat{V}_2 \oplus \ldots \oplus \hat{V}_{\hat{s}}, \quad \hat{s}\geq s,$$
            and $$\inangle{\opB \cdot \hat{U}_i} \subseteq \hat{V}_i, \quad \forall i \in [\hat{s}],$$
            then $\hat{s}= s$ and upto reordering if necessary, $\hat{U}_i=U_i$ and $\hat{V}_i=V_i$ for all $i \in [s].$
        \end{enumerate}
\end{proposition}

\noindent {\bf Noiseless algorithm.} Note that given $\opB$ (and $U, V$) computing $\adjfull{\opB}{U}{V}$ is easy and simply involves solving for $D$ and $E$ that satisfy the linear constraints specified in definition \ref{defn:adjointAlgebra}. Further under the assumption that $\adjfull{\opB}{U}{V}$ equals the set of scaling maps (this we refer to as strong uniqueness), the required subspaces $U_1, U_2, \ldots U_s$ can be obtained as the eigenspaces corresponding to distinct eigenvalues of the linear map $D : U \mapsto U$ which is the component of a random element $(D, E)$ of $\adjfull{\opB}{U}{V}$. \\

        
            
        


\noindent {\bf Making the algorithm robust.} There is a relatively straightforward way to make this algorithm robust: we use the maps in $\topB$ to compute a vector space\footnote{
    This space is typically not closed under multiplication and so does not form an algebra.
} that is in some sense an approximation to the original adjoint algebra. 
Finally, we recover the $U_i$'s approximately as (the sum of a few) eigenspaces of suitably chosen elements of this approximate adjoint algebra. In the noiseless setting it suffices to chose random elements of the adjoint algebra but in the noisy setting this does not work very well. This is because the error incurred in the recovery of an eigenvector/eigenspace of an operator is inversely related to the corresponding eigengap(s) (see lemma \ref{lemma:eigenvec_pert}). Simply picking a random element of the adjoint algebra $\adjfull{\opB}{U}{V}$ leads to a rather small eigengap and we therefore incur a rather large error both theoretically and practically (i.e. in both the worst case noise scenario and the random noise scenarios). Our insight here is that the multiplicative structure of the adjoint algebra can be exploited to find operators in it with (some) large eigengaps and this yields an algorithm that is more robust. Indeed, our initial experiments suggest that the resulting algorithm when applied to tensor decomposition empirically performs better (in terms of error in the output) than any of the known algorithms for tensor decomposition. The details and quantitative bounds are provided in section \ref{sec:rvsdAlgo}. \\


\noindent {\bf Our Results.} The noise-tolerance and performance of our meta-algorithm is captured by the following theorem which bounds the error incurred in terms of various parameters involved. 

\begin{theorem} [\textbf{Learning Noisy Arithmetic Circuits}, Informal version of Theorem \ref{thm:robustCircuitReconstruction}]\label{thm:IntroductionCircuitReconstruction}
    Let $f(\vecx) = T_1(\vecx) + \cdots + T_s(\vecx)$ be a polynomial such that each $T_i \in \R[\vecx]^{=d}$ belongs to a circuit class $\mathcal{C}$ that admits operators $\opL$ and $\opB$ satisfying the following properties: 
    \begin{itemize}
        \item $\opL$ consists of linear maps $L:\R[\vecx]^{=d} \to W_1$ such that $U \eqdef \inangle{\opL \cdot f} = U_1 \oplus \cdots \oplus U_s$, $\dim(U) = d_U$, where $U_i \eqdef \inangle{\opL \cdot T_i}$.
        \item $\opB$ consists of linear maps $B:W_1 \to W_2$ satisfying $V \eqdef \inangle{\opB \cdot \opL \cdot f} = V_1 \oplus \cdots \oplus V_s$, $\dim(V) = d_V$, where $V_i \eqdef \inangle{\opB \cdot \opL \cdot T_i}$.
        \item The decomposition of $(U,V)$ under $\opB$ is strongly unique, i.e. $\dim(\adjfull{\opB}{U}{V}) = s$.
    \end{itemize}
    We also need the robust versions of the above assumptions and that $\opL$ and $\opB$ are appropriately normalized. Let $M, N$ be matrices with columns $L \cdot f, L \in \opL$ and $B \cdot L \cdot f, B \in \opB, L \in \opL$ respectively. Suppose that the $d_U^{\text{th}}$ and the $d_V^{\text{th}}$ largest singular values of $M$ and $N$, respectively, are bounded from below by some $\sigma > 0$. Similarly, for an appropriate operator corresponding to the adjoint algebra, we need an appropriate singular value lower bounded by $\sigma$.
    
    Let $\tf(\vecx) = f(\vecx) + \eta(\vecx)$ be a polynomial such that $\norm{\eta} \le \epsilon$\footnote{Under an appropriate norm called the Bombieri norm as defined in Section \ref{sec:prelims}. The Bombieri norm is a suitably scaled version of the $\ell_2$ norm that has many desirable properties including being invariant under a unitary transformation of the underlying variables.}.
    Then, there is an efficient algorithm, which on input $\tf$, recovers $\tT_1, \tT_2, \ldots, \tT_s$, such that for any $\delta > 0$, with probability at least $1 - \delta$, (upto reordering) for each $i\in [s]$ it holds that
    \begin{equation*}
        \norm{T_i - \tT_i} \le \textnormal{poly} \left( s,d, d_U, d_V, 1/\delta, 1/\sigma \right) \cdot \epsilon.
    \end{equation*}
\end{theorem}

\begin{remark}
    \begin{enumerate}[label=\arabic*)]
        \item {\bf Error for random noise.} The above bound on the output error is for the case when the noise $\eta(\vecx)$ is chosen in an adversarial (i.e. worst-case) fashion, subject of course to the indicated upper bound on its norm. In practice $\eta(\vecx)$ often behaves like a random vector so that the output error is in practice significantly less\footnote{
            This situation is reminiscent of the well-studied spiked tensor problem in machine learning which can be thought of as a very special case of our problem. 
        } than the worst-case bound in the above theorem. Our intuition is that when $\eta(\vecx)$ is random the output error should be less by a factor of $\poly(dim(\inangle{\opL}))$ compared to when $\eta(\vecx)$ is 
        adversarially chosen. We leave it as a potential direction for future investigation. 

        \item {\bf Noise-tolerance.} As noted earlier, our initial experiments indicate that for the well-studied special case of tensor decomposition our algorithm seems to be more noise-tolerant than existing algorithms. We remark here that for subspace clustering, one can have a somewhat different reduction to vector space decomposition which also incorporates the affinity-based information to obtain a more noise-tolerant clustering algorithm.  It might be interesting to do an empirical comparison of noise-tolerance of (such adaptations of) our algorithm to existing algorithms for various applications of interest.

        \item {\bf Running Time.} The algorithm boils down to computing 
        singular value decompositions and/or pseudoinverses of certain matrices and thus its running time\footnote{
            This is in the model where operations over real numbers are of unit cost. A more precise bound on the running time in terms of the dimensions of the various relevant vector spaces can be 
        } is upper bounded by the cube of the dimension of the largest vector space involved. 
        
        \item We suggest a potential way to speed up the above algorithm  in section \ref{sec:conclusion}.

    \end{enumerate}
\end{remark}

\subsection{Application 1: Subspace Clustering.}\label{sec:scOverview}

Subspace clustering is the following problem - we are given a set of $N$ points 
$A =\{ \veca_1,\veca_2,\ldots,\veca_N \} \subseteq \R^n$ that admit a partition 
    $$ A = A_1 \uplus A_2 \uplus \ldots \uplus A_s, $$
such that the points in each $A_j$ ($j \in [s]$) span a low-dimensional (relative to the number of points in $A_j$) space $\inangle{A_j}$. 
The goal is to find such a partition.

Even for $n=3$, this problem is NP-hard in the worst case \cite{MegiddoTamir82}. Despite this, it has been intensely studied and we refer the reader to the surveys \cite{subspaceClusteringSurvey}, \cite{sc_survey} and the references therein. Most state of the art techniques rely on constructing an affinity matrix, which measures how likely two points are to be in the same subspace, followed by spectral clustering using the affinity matrix. Most such algorithms have little theoretical analysis about the robustness and recovery guarantees. \\

\noindent {\bf A non-degeneracy condition and a reduction.} Suppose now that the span of the $A_j$'s satisfy the following non-degeneracy condition: they form a direct sum, i.e.
    \begin{equation}\label{eqn:nondegAi}
        \inangle{A} = \inangle{A_1} \oplus \inangle{A_2} \oplus \ldots \oplus \inangle{A_s}.
    \end{equation}
We will see that in this case subspace clustering reduces to vector space decomposition in the following way. For a point $\veca =\inparen{a_1,a_2,\ldots,a_n} \in \R^n$, let $\veca \cdot \vecx \in \R[\vecx]$ denote the linear form $a_1 x_1 + a_2 x_2 + \cdots + a_n x_n$, in the formal variables $\vecx = \inparen{x_1,x_2,\ldots,x_n}$. For $d \geq 1$ we denote by $\tensored{A}{d}$ the set $\setdef{(\veca \cdot \vecx)^d}{\veca \in A} \subseteq \R[\vecx]^{=d}$. Consider the space of first-order partial differential operators $\opB = \opPartials{1}$ acting on the subspace of polynomials  
    $$ \inangle{\tensored{A}{2}} = \inangle{(\veca_1 \cdot \vecx)^2, (\veca_2 \cdot \vecx)^2, \ldots, (\veca_N \cdot \vecx)^2} \subseteq \R[\vecx]^{=2}. $$ 
The image space is then 
    $$ \inangle{\tensored{A}{1}} = \inangle{(\veca_1 \cdot \vecx), (\veca_2 \cdot \vecx), \ldots, (\veca_N \cdot \vecx)} \subseteq \R[\vecx]^{=1} . $$
Note that our non-degeneracy condition can be restated as saying that 
    $$ \inangle{\tensored{A}{1}} =  \inangle{\tensored{A_1}{1}} \oplus  \inangle{\tensored{A_2}{1}} \oplus \ldots \oplus  \inangle{\tensored{A_s}{1}}. $$
This implies that the subspaces $\inangle{\tensored{A_j}{2}}$ also form a direct sum. Its also easily seen that  the image of each $\inangle{\tensored{A_j}{2}}$ under $\opB = \opPartials{1}$ is precisely $\inangle{\tensored{A_j}{1}}$. Thus the vector space $\inangle{\tensored{A}{2}}$ admits a decomposition under the action of $\opB$. Furthermore, under the additional mild assumption that each $\inangle{\tensored{A_j}{2}}$ is {\em indecomposable} under the action of $\opB$ it turns out (using Corollary~\ref{corr:uniqueness_sc}) that the decomposition is unique and thus the subspace clustering problem reduces to the problem of vector space decomposition. \\

\noindent {\bf A weaker non-degeneracy condition.} Note that the non-degeneracy condition given by (\ref{eqn:nondegAi}) is rather restrictive - it implies in particular that the number of subspaces $s$ cannot exceed $n$, the dimension of the ambient space. We can get a weaker non-degeneracy condition by considering the action of first-order partial differential operators $\opB = \opPartials{1}$ on the space $\inangle{\tensored{A}{d}}$ instead (for some suitable choice of $d \geq 2$). The image space is then $\inangle{\tensored{A}{(d-1)}}$. As before, under the (now weaker) non-degeneracy condition that 
    $$ \inangle{\tensored{A}{(d-1)}} = \inangle{\tensored{A_1}{(d-1)}} \oplus \inangle{\tensored{A_2}{(d-1)}} \oplus \ldots \oplus \inangle{\tensored{A_s}{(d-1)}}, $$
the vector space $\inangle{\tensored{A}{d}}$ admits a decomposition under the action of $\opB$. Furthermore, as before, under the additional mild assumption that each $\inangle{\tensored{A_j}{d}}$ is indecomposable under the action of $\opB$ it turns out (Corollary~\ref{corr:uniqueness_sc})
that the decomposition is unique and thus the subspace clustering problem reduces to the problem of vector space decomposition (see Theorem~\ref{thm:sc_noiseless}). \\

\noindent {\bf Robust subspace clustering.}
The robust or noisy version of the subspace clustering problem is the following. Given a set of points $\tA = \{ \tveca_1, \tveca_2, \ldots, \tveca_N \} \subseteq \R^n$ suppose that each point $\tveca_i$ is {\em close to} an (unknown point) $\veca_i \in \R^n$ such that the resulting set of points 
$A =\{ \veca_1,\veca_2,\ldots,\veca_N \} \subseteq \R^n$ can be clustered using $s$ subspaces, i.e.  
    $$ A = A_1 \uplus A_2 \uplus \ldots \uplus A_s, $$
where each $A_j$ spans a {\em low-dimensional} subspace $\inangle{A_j}$. The computational task is to {\em approximately} recover each subspace  $ \inangle{A_j} $, that is, output $\tilde{\vecW} = (\tW_1,\tW_2,\ldots,\tW_s)$ such that (upto reordering) each $\tW_j$ is close to $\inangle{A_j}$ for each $j \in [s]$. We can reduce this problem to the robust vector space decomposition as follows. Let 
$m_d \eqdef \dim(\inangle{\tensored{A}{d}} ) $ and $m_{d-1} \eqdef \dim (\inangle{\tensored{A}{(d-1)}}) $. 
Given $\tA$ we algorithmically compute the best fitting subspace $\tU$ (resp. $\tV$) of dimension $m_d$ (resp. $m_{d-1}$) to $\tensored{\tA}{d} $ (resp. to $\tensored{\tA}{(d-1)} $). It turns out then that $\tU$ (resp. $\tV$) is {\em close to} $\inangle{\tensored{A}{d}}$ (resp. to $\inangle{\tensored{A}{(d-1)}}$) (Lemmas~\ref{lemma:sc_distanceToTensored},~\ref{lemma:sc_distanceToTensored2} give the quantitative bounds). Applying the robust version of vector space decomposition on $(\tU, \tV, \opB)$, the subspaces that we obtain are close to $\inangle{\tensored{A_j}{d}}$ ($ j \in [s]$) and from these we can, in turn, also approximately recover $\inangle{A_j}$ (Proposition~\ref{prop:rrstp}), as required. This yields the following theorem.

\begin{theorem}[\textbf{Robust Subspace Clustering}, Informal version of Theorem \ref{thm:sc}]
	Let $A = \inbrace{\veca_1,\dots,\veca_N} \subseteq \R^n$ be a finite set of $N$ points of unit norm, which can partitioned as $A = A_1 \uplus \cdots \uplus A_s$, where each $\inangle{A_i}$ is subspace of dimension at most $t$.
     
    Let $d\geq 2$ be an integer, let $\vecU = (U_1,\dots,U_s)$ (resp. $\vecV = (V_1,\dots,V_s)$) be an $s$-tuple of subspaces with $U_j =  \inangle{\tensored{A_j}{d}}$ (resp. $V_j =  \inangle{\tensored{A_j}{d-1}}$) for each $j\in [s]$.
    Let $U = \inangle{\vecU}$ (resp. $V = \inangle{\vecV}$) have dimension $m_d$ (resp. $m_{d-1}$).
    
    Suppose that:
    \begin{itemize}
    	\item $U = U_1\oplus\dots\oplus U_s$, $V = V_1\oplus\dots\oplus V_s$, and for each $j\in [s]$, it holds that $\dim(U_j) = \binom{\dim(\inangle{A_j})+d-1}{d}$, $\dim(V_j) = \binom{\dim(\inangle{A_j})+d-2}{d-1}$.
    	\item $\sigma_A$ is the minimum of $\sigma_{m_d}(M_{A,\ d})$ and $\sigma_{m_{d-1}}(M_{A,\ d-1})$, where $M_{A,d}$ (resp. $M_{A,\ d-1}$) is the matrix whose columns are the polynomials $(\veca_i\cdot\vecx)^d$ (resp. $(\veca_i\cdot\vecx)^{d-1}$) (see Definition~\ref{defn:sc_m_ad}). 
    	\item $\kappa(\vecU)$ denotes the condition number of the tuple of subspaces $\vecU$ (see Section~\ref{sec:prelims}).
    	\item $\sigma_{-(s+1)}(\adjmap)$ is the $(s+1)\textsuperscript{th}$ smallest singular value of the adjoint algebra map (see Definition~\ref{defn:adj_alg_map}), corresponding to the action of $\opB=(B_1,\dots,B_n)$ on $\vecU, \vecV$, where $B_i$ corresponds to the operator $\partial_{x_i}$. 
    \end{itemize}
    
   	Let $\tA = \inbrace{\tveca_1, \tveca_2, \ldots, \tveca_N}  \subseteq \R^n$ be a set of unit norm vectors such that $\lnorm{\veca_i - \tveca_i} \leq \epsilon$ for each $i \in [N]$.
   	Then, there is an algorithm, which on input $\tA$, runs in time $\poly(N, n^d)$, and recovers subspaces $(\tW_1, \tW_2, \ldots, \tW_s)$, such that with probability at least $1 - \delta$, (upto reordering) for each $j\in [s]$ it holds that
   	\[ \dist(\tW_j, \inangle{A_j}) \leq \poly\inparen{t, N, d, s, 1/\delta,\ \kappa(\vecU),\ 1/\sigma_A,\ 1/\sigma_{-(s+1)}(\adjmap)}\cdot \epsilon.\]
\end{theorem}


We show how to lower bound $\sigma_{-(s+1)}(\adjmap)$ (Theorem \ref{thm:sing_val_of_adjoint_algebra_operator}). The main technical component is an inductive argument to analyze singular values of basic adjoint operators, which is inspired by the inductive argument in recent works on analyzing eigenvalues for random walks on simplicial complexes (e.g. \cite{anari2019log}). Next we see what our algorithms would yield in the smoothed case and state some explicit conjectures about singular values of relevant smoothed matrices.
\\

\noindent {\bf Smoothed analysis of subspace clustering.} We first describe the input model. For simplicity, we assume that each of the subspaces have the same dimension (equal to $t$). 
\begin{enumerate}
    \item {\bf Perturbation model for subspaces.} We have a tuple of $s$ hidden subspaces of $\R^n$, $\vecW = (W_1, W_2, \ldots, W_s)$, each of dimension $t$. Let $P_1, P_2, \ldots, P_s \in \R^{n \times t}$ be matrices with orthonormal columns, such that the column span of $P_i$ is $W_i$. Each subspace $W_i$ is perturbed by perturbing $P_i$ by a random (and independent) Gaussian matrix $G_i \sim \mathcal{N}(0,\rho^2/n)^{n \times t}$. Let $\hat{P}_i = P_i + G_i$, and $\hat{W}_1, \hat{W}_2, \ldots, \hat{W}_s$ be the column spans of $\hat{P}_1, \hat{P}_2, \ldots, \hat{P}_s$ respectively. 
    
    \item {\bf Perturbation models for points from each subspace.} Sample (possibly adversarially) sets of points $A_1, A_2, \ldots, A_s$ from $\hat{W}_1, \hat{W}_2, \ldots, \hat{W}_s$ respectively, of unit norm. For each $i \in [s]$, perturb each point in $A_i$ with respect to $\hat{W}_i$ to get the set of points $\hat{A}_i$. Formally, this means perturbing points in $A_i$ by $\hat{B}_i \cdot v$, where $\hat{B}_i$ is an $n \times t$ matrix describing an orthonormal basis for $\hat{W}_i$ and $v \sim \mathcal{N}(0,\rho^2/t)^t$ (independently generated for each point), and normalizing. Let $\hat{A} = \hat{A}_1 \cup \hat{A}_2 \cup \cdots \cup \hat{A}_s$.
    
    \item {\bf Adding noise.} For each $\veca \in \hat{A}$, add noise (possibly adversarially) and normalize to get a unit norm point $\veca^\prime$ such that $\norm{\veca - \veca^\prime}_2 \le \epsilon$. We are given as input $\hat{A}^\prime$, the set of noise-added points.
\end{enumerate}

Given the set of points $\hat{A}^\prime$, the goal is to recover subspaces $\tvecW = (\tW_1, \tW_2, \ldots, \tW_s)$ such that $\dist(\hat{\vecW}, \tvecW)$ is small. Next we state a couple of conjectures about minimum singular values of smoothed random matrices that we encounter:

\begin{conjecture}\label{conj:smoothed1}
    Let $\vecv_{i1},\ldots, \vecv_{it}$ be an orthonormal basis for $\hat{W}_i$ generated as above. Define the linear forms $\ell_{ij}(\vecx) = \langle \vecv_{ij}, \vecx \rangle$. Consider the $\binom{n+d-1}{d} \times s \binom{t+d-1}{d}$ matrix $M$ where the columns are divided into $s$ chunks and in the $i^{\text{th}}$ chunk, the columns are all the monomials of degree $d$ in the polynomials $\ell_{i1},\ldots, \ell_{it}$. Also suppose $s \binom{t+d-1}{d} \le (1-\delta) \binom{n+d-1}{d}$ for a constant $\delta > 0$. Then for constant $d$, with high probability, $\sigma_{s \binom{t+d-1}{d}}(M) \ge \poly\left( \rho, 1/n\right)$.
\end{conjecture}

\begin{conjecture}\label{conj:smoothed2}
    Consider arbitrary vectors $v_1,\ldots, v_s \in \R^t$ of unit norm and their smoothed versions $\hat{v}_1,\ldots, \hat{v}_s$, where $\hat{v}_i = v_i + g_i$, $g_i \sim \mathcal{N}(0,\rho^2/t)^t$ (and then further normalized to unit norm). Consider the $s \times \binom{t+d-1}{d}$ matrix $M$ where the $i^\text{th}$ row contains the polynomial $\langle \hat{v}_i, \vecx \rangle^d$. Suppose $s \ge (1 + \delta) \binom{t+d-1}{d}$ for a constant $\delta > 0$. Then for constant $d$, with high probability, $\sigma_{\binom{t+d-1}{d}}(M) \ge \poly\left( \rho, 1/t\right)$.
\end{conjecture}

\begin{theorem}[\textbf{Smoothed analysis of subspace clustering}, Theorem \ref{thm:smoothed_sc} restated]\label{thm:smoothed_sc_intro}
    Suppose Conjectures \ref{conj:smoothed1} and \ref{conj:smoothed2} are true. Then for constant $d$, Algorithm \ref{alg:sc} on input $(\hat{A}^\prime, d, s, m_d, m_{d-1})$ outputs $\tvecW = (\tW_1, \ldots, \tW_s)$ such that with high probability,
    $$
    \dist(\tW_j, \hat{W}_j) \leq \poly\left( n, t, 1/\rho\right) \cdot \epsilon
    $$
\end{theorem}

Regarding the two conjectures, Conjecture \ref{conj:smoothed2} is closely linked to the paper \cite{bhaskara2019smoothed}. There they considered the setting where $s \le (1-\delta) \binom{t+d-1}{d}$ and proved a similar lower bound for $\sigma_s(M)$. In both the settings there is slack, so it is plausible that the techniques of \cite{bhaskara2019smoothed} can be adapted to prove Conjecture \ref{conj:smoothed2}. But we don't know how to do that. In Conjecture \ref{conj:smoothed1}, the matrix $M$ is such that both the rows and columns share random variables. Most of the smoothed analysis till now focuses on matrices where either rows or columns have different sets of variables involved, and this makes it amenable to the leave-one-out distance method. Still, in Conjecture \ref{conj:smoothed1}, the sharing of variables is not completely arbitrary. One can divide rows into chunks so that different chunks have different sets of variables. However, even this setting seems to require new techniques to analyze.

\subsection{Application 2: Learning Mixtures of Gaussians}\label{sec:gmOverview}
In this section we will see how the problem of computing the parameters of a mixture of Gaussians reduces to (several instances of) vector space decomposition. \\

\noindent {\bf Reduction to a special case of formula learning.} It is implicit in \cite{GHK15} that learning a mixture of $s$ zero-mean Gaussians reduces to robustly expressing a given homogeneous polynomial $p(\vecx)$ as a sum of $s$ powers of quadratics, i.e.
    \begin{equation}\label{eqn:power-sum}
       p(\vecx) = p_1(\vecx)^d + p_2(\vecx)^d + \ldots + p_s(\vecx)^d, 
    \end{equation}
where the $p_i$'s are homogeneous quadratic polynomials. Following the ideas in \cite{GKS20}, we give a direct reduction\footnote{In \cite{GKS20}, there is an additional "multi-gcd" step which we avoid here.} to vector space decomposition as follows. \\

\noindent {\bf Obtaining a vector space that is the direct sum of unknown spaces.}
    Following \cite{GKS20}, we apply partial derivatives followed by a {\em random projection} to obtain a vector space that is a direct sum of $s$ unknown subspaces, one corresponding to each $p_i(\vecx)$. Specifically, let $\opL$ be the set of operators corresponding to taking $k$-th order partial derivatives followed by {\em a random restriction}\footnote{W can think of a random projection as keeping a subset $\vecy\subseteq \vecx$ of the variables alive and setting the rest to zero.}. Applying $\opL$ to both sides of equation (\ref{eqn:power-sum}), we get 
        $$ \inangle{\opL \cdot p(\vecx)} \subseteq \inangle{\opL \cdot p_1(\vecx)^d} + \inangle{\opL \cdot p_2(\vecx)^d} + \ldots + \inangle{\opL \cdot p_s(\vecx)^d}, $$
    It turns out that (Lemma \ref{lem:equalsDirectSum}) under relatively mild nondegeneracy conditions on the choice of the $p_i$'s, the vector space sum on the right hand side of the above equation is actually a direct sum and the containment is actually an equality, i.e. 
        $$ \inangle{\opL \cdot p(\vecx)} = \inangle{\opL \cdot p_1(\vecx)^d} \oplus \inangle{\opL \cdot p_2(\vecx)^d} \oplus \ldots \oplus \inangle{\opL \cdot p_s(\vecx)^d}. $$
    We now carefully choose another set of operators $\opB$ such that the subspace $U \eqdef \inangle{\opL \cdot p(\vecx)} $ admits a {\em unique} decomposition under the action of $\opB$. \\
    
\begin{sloppy}
\noindent{\bf Choice of $\opB$.}    
    The set of operators $\opL$ maps polynomials in $\vecx$ to polynomials in a subset of variables $\vecy \subseteq \vecx$. Under the above mentioned nondegeneracy conditions, it also turns out that for each $i \in [s]$, $\inangle{\opL \cdot p_i(\vecx)^d}$ is of the 
    form $U_i \eqdef \inangle{\vecy^{=k} \cdot q_i(\vecy)^{d-k}} \subseteq \R[\vecy]^{=(2d-k)}$. With this 
    in mind, we choose $\opB$ as the following set of operators: first order partial derivatives followed by multiplication\footnote{
        The relevant literature on arithmetic formula lower bounds would refer to the set of operators $\opB$ as {\em shifted partials} and denote it by $\vecy^{=1} \cdot \opPartials{1}_{\vecy} $. 
    } by polynomials of degree $1$. In detail: $\opB$ consists of $\inabs{\vecy}^2$ operators with the 
    $(i, j)$-th operator ($i, j \in [\inabs{\vecy}]$) being 
        $$ B_{ij} : \R[\vecy]^{=(2d-k)} \mapsto \R[\vecy]^{=(2d-k)}, \quad B_{ij} \cdot q(\vecy) = y_j \cdot (\partial_{y_i} q(\vecy)) \ \text{for any~} q(\vecy) \in \R[\vecy]^{=(2d-k)}. $$
    It turns out that for any $i \in [s]$, under the action of $\opB$, the image of 
    $$ U_i \eqdef \inangle{\vecy^{=k} \cdot q_i(\vecy)^{d-k}} \ \text{is the subspace~} V_i \eqdef \inangle{\vecy^{=(k+2)} \cdot q_i(\vecy))^{d-k-1}}$$ 
    and that the $U_i$'s and $V_j$'s form direct sums (Lemma \ref{lem:direct_sum}). Furthermore, under mild non-degeneracy conditions such a decomposition is unique (Corollary \ref{cor:unique}) implying that our vector space $U$ has a unique decomposition into $s$ subspaces under the action of $\opB$. Lastly, from each $U_i$ we can recover the corresponding $q_i(\vecy)$ which is a restriction of $p_i(\vecx)$ to a chosen subspace. Any polynomial 
    can be recovered from its restriction to a small number of chosen subspaces and we use this to recover 
    each $p_i(\vecx)$ ($ i \in [s]$), as required. In this way, the problem of learning mixtures of Gaussians 
    reduces to robust vector space decomposition. \\
\end{sloppy}
    

\noindent {\bf Robust version.} 
Our general algorithm for learning arithmetic circuits with noise (Theorem \ref{thm:IntroductionCircuitReconstruction}) can be used to make the above algorithm robust.
We will also need to use the algorithm of \cite{bafna2022polynomial} in this case to combine the various projections of $p_i$'s. Our algorithm will depend on condition numbers of certain matrices which can be deduced from the operators used in the above algorithm. Lemmas \ref{lem:equalsDirectSum}, \ref{lem:direct_sum} and \ref{lem:rednMain} show that at least the ranks of these matrices are as expected. Lemmas \ref{lem:equalsDirectSum} and \ref{lem:direct_sum} are from \cite{GKS20}. Lemma \ref{lem:rednMain} is new and is the main technical contribution for this section, and shows that the relevant adjoint algebra is of the correct dimension. Also \cite{bafna2022polynomial} analyze similar matrices corresponding to Lemmas \ref{lem:equalsDirectSum} and \ref{lem:direct_sum} and prove the required condition number bounds in the fully random case. For the singular values of the adjoint operator (robustification of Lemma \ref{lem:rednMain}), we believe similar techniques as Theorem \ref{thm:sing_val_of_adjoint_algebra_operator} should work to give us a bound but the setting is more challenging and we don't know how to prove a bound here yet.

~\\\noindent {\bf Comparison to~\cite{GKS20} and~\cite{ bafna2022polynomial}. } The algorithms of~\cite{GKS20, bafna2022polynomial} for learning mixtures of Gaussians roughly proceed as follows (for simplicity, we only consider the the noiseless case here).

Given a polynomial $p(x) = \sum_{i=1}^s p_i(\vecx)^d$, where each $p_i$ is a quadratic polynomial:
\begin{enumerate}
    \item Apply a set of operators $\mathcal{L}$ to $p(\vecx)$, where $\mathcal{L}$ corresponds to taking some $k$-th order partial derivatives followed by a random restriction:  as described before, each $\inangle{\mathcal{L}\cdot p_i(\vecx)}$ is of the form $\inangle{\vecy^{=k} \cdot q_i(\vecy)^{d-k}} \subseteq \R[\vecy]^{=(2d-k)}$.
    This step is essentially the same in both~\cite{GKS20, bafna2022polynomial}.
    
    We note however that~\cite{bafna2022polynomial} actually do not work under the non-degeneracy condition of the spaces $\inangle{\mathcal{L}\cdot p_i(\vecx)}$'s forming a direct sum, and instead explicitly characterize the structure of the intersections $\inangle{\mathcal{L}\cdot p_i(\vecx)}\cap \inangle{\mathcal{L}\cdot p_j(\vecx)}$. This allows them to deal with a broader range of parameters compared to~\cite{GKS20}.
    
    \item The next step is a "multi-gcd" step, which is used to find the vector space $\inangle{q_i(\vecy)^{d-k}} + \dots + \inangle{q_s(\vecy)^{d-k}}$.
    This step is already present in the algorithm of~\cite{GKS20}, however~\cite{bafna2022polynomial} give a significantly simpler algorithm for this step, along with an analysis for the robust version of this step. 

    \item The next step, which is in some sense the "main part" of the algorithm, is where the two algorithms~\cite{GKS20} and~\cite{bafna2022polynomial} differ:
    \begin{enumerate}
        \item The algorithm of~\cite{GKS20} considers another application of $k$-th order partial derivatives + random restriction on this vector space, and uses vector space decomposition with respect to this set of operators.
        This allows them to recover the component polynomials.
        
        \item The algorithm of~\cite{bafna2022polynomial} follows the approach in~\cite{GHK15}, and does  a "desymmetrization + tensor-decomposition" step.
        This roughly enables them to convert the sum of polynomials to a sum of tensors, and then apply standard tensor decomposition methods to obtain the required components.
    \end{enumerate}

    \item The final step is to repeat the above procedure multiple times, using a different random restriction each time, and then aggregating the obtained $q_i(\vecy)$'s into $p_i(\vecx)$, as described before.
\end{enumerate}

Our algorithm essentially follows the same first and final step as both these algorithms.
It significantly deviates from the two algorithms in Steps 2 and 3:
\begin{enumerate}
    \item While we follow the same vector space decomposition paradigm as~\cite{GKS20}, our algorithm completely eliminates the use of the multi-gcd step. Instead, we use a very simple set of operators, namely order one partial derivatives + order one shifts, directly on the vector space $\inangle{\mathcal{L}\cdot p(\vecx)}$.
    Hence, our approach provides a much more direct reduction to vector space decomposition.
    
    \item In comparison to~\cite{bafna2022polynomial}, we first eliminate the use of the multi-gcd step, and further we do not go through the desymmetrization step at all. Instead, our framework of vector space decomposition allows us to deal with symmetric polynomials throughout the algorithm; this inherently seems much more natural since the inputs and outputs all deal only with polynomials (symmetric tensors).
\end{enumerate}

Finally, we note the the above described simplification allows us to obtain a much better range of parameters compared to~\cite{GKS20}, whereas we still expect them to be slightly worse than~\cite{bafna2022polynomial}.

\subsection{Conclusion and Future Directions}\label{sec:conclusion}
In this work we showed how to adapt the algorithm of \cite{GKS20} for learning subclasses of arithmetic formulas to make it noise-tolerant. This turns out to have a number of applications arising out of the remarkable fact that in these applications, a suitably defined polynomial formed out of the statistics of the data has a small arithmetic formula. We feel that our approach has the potential to give algorithms which are fast, noise-tolerant, outlier-tolerant and come with provable guarantees\footnote{For most such applications the worst-case instances  are intractable so the best we can hope for are algorithms whose performance can be bounded using singular values of certain instance-dependent matrices.} for many such applications and is therefore worthy of further investigation. We now pose some problems that might encourage or guide such further study.  \\

\noindent {\bf Making the vector space decomposition algorithm faster.}
Consider a set of operators $\opB$ mapping a real vector space $U$ to another real vector space $V$.  
Our algorithm for decomposition of $U$ (and $V$) under the action of $\opB$ involved computations with the  adjoint algebra which entailed working in the vector spaces of linear maps $\Lin(U, U)$ and $\Lin(V, V)$. These spaces of linear maps have larger dimension than that of $U$ and $V$ themselves and consequently, our approach for decomposing $U$ has running time pertaining to the cost of doing linear algebra over spaces of dimension $(\dim(U)^2 + \dim(V)^2)$. Let us first make an observation. Suppose that the decomposition induced by $\opB$, namely: 
    $$ U = U_1 \oplus U_2 \oplus \ldots \oplus U_s, \quad V = V_1 \oplus V_2 \oplus \ldots \oplus V_s $$
had the property that the $U_i$'s (respectively also the $V_i$'s) were orthogonal complements of each other (under some canonical inner product on the spaces $U$ and $V$). 
Consider the collection of linear maps $\opL \subseteq \Lin(U, U)$ defined as 
    $ \opL := \setdef{ B_j^{T} \cdot B_i}{B_i, B_j \in \opB} $. 
Then each $U_i$ is an invariant subspace (i.e. an eigenspace) of every operator in $\opL$. In such a situation we 
typically expect the following simple algorithm to work: simply pick three random maps $B_1, B_2, B_3 \in \inangle{\opB}$ and compute\footnote{
    The linear maps $B_2^T, B_3^T$ from $V$ to $U$ are defined using the canonical inner products on these two spaces. 
} $L := B_2^{T} \cdot B_1 $ and $M := B_3^{T} \cdot B_1$. Then for each eigenvector $u$ of $L$, compute the span of the orbit of $u$ under the action of $M$.
The distinct subspaces so obtained should typically give us the required subspaces $U_1, U_2, \ldots, U_s$. 
Clearly such an algorithm, {\em when it works}, would be much faster. 
We expect that for most applications, the above algorithm should work but we don't know. 

\begin{problem}
    For problems such as subspace clustering and learning mixtures of Gaussians, if the relevant $U_i$'s (respectively also the $V_i$'s) are orthogonal to each other, does the above algorithm correctly 
    recover the $U_i$'s? 
\end{problem}

\begin{problem}
    Whats the best way to make this algorithm noise-tolerant? 
\end{problem}

Finally, in situations where the $U_i$'s (resp the $V_i$'s) are not orthogonal to each other we can clearly make them so by using appropriate inner products on $U$ and $V$. But how do we find such an inner product? 
We expect the operator scaling algorithm of \cite{ggow20} to yield such an inner product(!) 

\begin{problem}
    For (noiseless) subspace clustering, does the operator scaling algorithm of \cite{ggow20} applied on the relevant $\opB$ yield inner products under which the relevant subspaces are orthogonal?  
\end{problem}

\noindent {\bf Mixture of Gaussians.}
As mentioned in remark \ref{rmk:potential}(b) earlier, we expect that our algorithm can be extended to handle general mixtures of Gaussians with differing means and covariance matrices. Let us formally state this as an open problem. 

\begin{problem}{\bf Random instances of general mixtures of Gaussians.} 
    Let $n, s \geq n$ be integers. For $i \in [s]$ suppose that we pick $\bm{\mu}_i \in \R^n$ and covariance matrices $\Sigma_i \in \R^{n \times n}$ independently at random\footnote{Any reasonable distribution would do 
    but for concreteness say we pick $\bm{\mu}_i \sim \mathcal{N}(0, I_n)$ and we pick $\Sigma_i = B^{T} \cdot B$, where 
    $B \sim \mathcal{N}(0, 1)^{n \times n}$}. Let $ \mathcal{D} := \sum_{i=1}^s \frac{1}{s} \cdot \mathcal{N}(\bm{\mu}_i, \Sigma_i)$ be the equi-weighted mixture of Gaussians with the above randomly chosen parameters. Design an efficient algorithm that given samples from $D$ recovers the $\bm{\mu}_i$'s and $\Sigma_i$'s approximately.
\end{problem}

Our work as well as that of \cite{bafna2022polynomial} leave open the problem of doing a smoothed analysis of the corresponding algorithm for mixtures of zero-mean Gaussians. To encourage this direction of research, let us state this explicitly in the form of a conjecture.

\begin{conjecture}{\bf Smoothed analysis of our algorithm for mixture of zero-mean Gaussians.} 
    Our algorithm efficiently recovers the unknown parameters for smoothed instances of mixtures of zero-mean Gaussians. 
\end{conjecture}
    
\noindent {\bf Handling outliers and other applications.} In Remark \ref{rmk:potential}, we conjectured that our approach/framework should enable the design of efficient algorithms that can handle outliers and also be useful for many more applications in unsupervised learning. It would be nice to have concrete results in such directions.

\section{Preliminaries}\label{sec:prelims}
We shall use $[n]$ to denote the set $\{1,2,\dots,n\}$.


~\\{\bf Matrices, Norms, Pseudo-inverse.} 
Let $\R^{m\times n}$ denote the space of $m\times n$ matrices over $\R$.
Given an $m\times n$ matrix $M$, we denote by $\lnorm{M}$ its operator norm, and by $\fnorm{M}$ its Frobenius norm.
We shall use $M^\dag$ to denote its Moore-Penrose pseudo-inverse, and we define its condition number  as $\kappa(M)\eqdef \lnorm{M}\lnorm{M^\dag}$.
We mention some relevant properties of matrix norms and the pseudo-inverse in Section~\ref{subsec:app_matrix}.


~\\{\bf Vector Spaces, Linear Operators, Projection Maps.} 
    Every vector space $V$ that we will deal with in this work will be a finite dimensional real vector space
    that comes equipped with an inner product\footnote{
        Very often, $V$ is a space of homogeneous multivariate polynomials in which case the 
        inner product is the Bombieri inner product.
    } denoted $\langle \cdot , \cdot \rangle_{V}$, or simply  $\langle \cdot , \cdot \rangle$ when the underlying vector space is clear from context.
    For two vector spaces $U$ and $V$, $\Lin(U, V)$ shall denote the set of linear maps from $U$ to $V$.
    
    Inner products on $U$ and $V$ can be used to generate an inner product on $\Lin(U, V)$, called the Hilbert–Schmidt inner product:
    For $A, B \in \Lin(U,V)$, we define $\langle A, B \rangle_{\Lin(U,V)}$ $= \sum_{i\in[\dim(U)]} \langle A\cdot \vece_i, B\cdot \vece_i\rangle_V$ where $\vece_1, \dots, \vece_{\dim(U)}$ is an orthonormal basis of $U$ with respect to $\langle\cdot,\cdot \rangle_U$. It is shown easily that this inner product is independent of the choice of the orthonormal basis, and it matches the usual Frobenius inner product on the space of $\dim (U) \times \dim(V)$ matrices, when $A$ and $B$ are represented as matrices under a choice of orthonormal basis for $U$ and $V$.

    For any linear map $A\in \Lin(U,V)$, we shall use $\lnorm{A}$ to denote its operator norm, and $\fnorm{A}=\langle A,A\rangle_{\Lin(U,V)}^{1/2}$ to denote its norm under the Hilbert-Schmidt inner product.

    For any linear map $A \in \Lin(U,V)$, we use $\sigma_n(A)$ to denote the $n$\textsuperscript{th} largest singular value of $A$. We also use $\sigma_{-n}(A)$ to denote the $n$\textsuperscript{th} smallest singular value.

If $U \subseteq V$ is a subspace, then $\orth{U} \subseteq V$ shall denote the subspace that is an the orthogonal complement of $U$, and $\Proj_{U} \in \Lin(V, V)$ shall denote the projection onto $U$, i.e. 
        \begin{align*}
           \orth{U}     &\eqdef \{\vecw \in V : \inangle{\vecu, \vecw} = 0  \} \subseteq  V \quad \mathrm{and} \\
           \Proj_{U} \cdot \vecv &= \vecu \quad \text{where $\vecu \in U,\ \vecw \in \orth{U}$ are the unique vectors such that $\vecv = \vecu + \vecw. $}
        \end{align*}

~\\{\bf Tuples of subspaces.}
    Let $\vecU = (U_1, U_2, \ldots, U_s)$, $U_i \subseteq V$, be an $s$-tuple of subspaces.
    $\inangle{\vecU}$ shall denote the span of the constituent subspaces, i.e.
        $$ \inangle{\vecU} \eqdef U_1 + U_2 + \ldots + U_s. $$
    We will be interested in recovering the constituents of a subspace tuple $\vecU$ using operators acting on $\inangle{\vecU}$. Towards this end, 
    we fix some relevant terminology.  

    {\it Associated matrices, independent tuples of subspaces, condition numbers.}
    Let $\vecU = (U_1, U_2, \ldots, U_s)$ be an $s$-tuple of subspaces of an $n$-dimensional vector space $V$, and $d_i = \dim(U_i)$ and $d = \sum_{i \in [s]} d_i$. We will say that an $n\times d$ matrix $M$ is a 
    $\vecU$-{\em associated matrix} if and only if the first set of $d_1$ columns of $M$ forms an orthonormal basis for $U_1$, 
    the next set of $d_2$ columns
    forms an orthonormal basis for $U_2$ and so on and the last set of $d_s$ columns forms an orthonormal basis of $U_s$. 
    Any two $\vecU$-associated matrices are equal up to right multiplication by an orthogonal matrix and so the rank and condition number of all matrices associated to $\vecU$ are the same.
    
    We will say that $\vecU = (U_1, U_2, \ldots, U_s)$ is an {\em independent} tuple of subspaces if 
        \[ \inangle{\vecU} = U_1\oplus U_2\oplus\dots U_s.\]
    In particular, $\vecU$ is an independent tuple of subspaces if and only if any $\vecU$-associated matrix has full rank. Motivated by this, we define a measure of the "robustness" of independence of the $U_i$'s, called the condition number $\kappa(\vecU)$ of the tuple $\vecU$, as the condition number $\kappa(M)$ of a $\vecU$-associated matrix $M$.

~\\ {\bf Distances between Subpaces.} 
The distance $\dist(U,V)$ between subspaces $U,V \subseteq W$ will be defined by \[\dist(U,V)= \lnorm{ \Proj_{U} - \Proj_{V}},\] and correspondingly we will say that these two subspaces are $\epsilon$-close if $\dist(U,V) \leq \epsilon$. 
In particular, observe that for any two subspaces $U,V\subseteq V$, it holds that $\dist(U,V)\in [0,1]$, and that $\dist(U,V)=1$ if $\dim(U)\not=\dim(V)$.

The distance between two $s$-tuples $\vecU = (U_1, U_2, \ldots, U_s)$ and $\vecV = (V_1, V_2, \ldots, V_s)$ of subspaces of a vector space $W$, is defined as $\dist(\vecU, \vecV) \eqdef \max_{i \in [s]} \dist(U_i,V_i).$

~\\{\bf Tuples of operators.}
Let $\opB = (B_1, B_2, \ldots, B_m) \in (\Lin(U, V))^m$ be an $m$-tuple of linear operators. 
For $\vecu \in U$, $\inangle{\opB \cdot \vecu}$ shall denote the space spanned by $\setdef{B_i \cdot \vecu}{i \in [m]}$.
Similarly, for a subspace $U' \subseteq U$, $\inangle{\opB \cdot U'}$ shall denote the space spanned by $\{B_i \cdot \vecu \ : \  \vecu \in U', i \in [m]\}$.

\begin{definition}\label{defn:joint_op}
    Corresponding to any such $m$-tuple $\opB = (B_1,\dots,B_m)$ of operators, we shall associate a linear map $\hat{B} \in \Lin(U,V^m)$, given by \[\hat{B}\cdot \vecu = (B_1\cdot \vecu, \dots, B_m\cdot\vecu).\]
    
    Further, we shall define the norm $\lnorm{\opB}$ of the tuple of operators, to be the norm $\lnorm{\hat{B}}$. More generally, we define the $i$-th singular value of $\opB$ to be the $i$-th singular value of $\hat{B}$, and $\kappa(\opB) = \norm{\hat{B}}_2 \cdot \norm{\hat{B}^\dagger}_2$.

    Note that it holds trivially that $\lnorm{\opB}\leq \sqrt{m}\cdot \max_{i\in [m]}\lnorm{B_i}$.
\end{definition}

    


~\\ {\bf Vector Spaces of Homogenous Polynomials.} 
We denote by $\R[\vecx]^{=d}$ the space of degree $d$ homogenous polynomials in $n$ variables $\vecx = (x_1,\ldots, x_n)$. Let $\N_d^n$ denote the set of multi-indices i.e. the set of $n$-tuples of non-negative integers $\vecalpha = (\alpha_1,\alpha_2, \ldots,\alpha_n)$ such that $|\vecalpha| = \alpha_1 + \cdots + \alpha_n = d$. Then, for $\vecalpha\in \N_d^n$, we use $\vecx^\vecalpha$ to denote the monomial $x_1^{\alpha_1} x_2^{\alpha_2} \ldots x_n^{\alpha_n} \in \R[\vecx]^{=d}$. These monomials form a basis of $\R[\vecx]^{=d}$, and we have $\dim \R[\vecx]^{=d} = |\N_d^n| = \binom{n+d-1}{d}$. 

We endow $\R[\vecx]^{=d}$ with the Bombieri inner product, defined on the monomials as follows:
\begin{equation*}
    \IP{\vecx^{\vecalpha}}{\vecx^\vecbeta}{B} =     \begin{cases}
        \frac{\vecalpha!}{d!} & \text{ if } \vecalpha = \vecbeta, \\
        \multicolumn{1}{@{}c@{\quad}}{0}  & \text{ otherwise.} 
    \end{cases}
\end{equation*} where $\vecalpha! = \alpha_1 ! \alpha_2 ! \cdots \alpha_n !$. The Bombieri basis i.e. the orthonormal basis with respect to this inner product is the basis of scaled monomials: $p_\vecalpha(\vecx) = \sqrt{\frac{d!}{\vecalpha!}} \vecx^\vecalpha$, $\vecalpha \in \N_d^n$. 

We use the following properties of the Bombieri inner product.
    For any homogenous polynomials $p,q$, and any $\veca, \vecb \in \R^n$, it holds: 
    \begin{itemize}
        \item $\norm{p \cdot q}_B \le \norm{p}_B \norm{q}_B$.
        \item $\IP{(\veca \cdot \vecx)^d}{(\vecb \cdot \vecx)^d}{B} = \IP{\veca}{\vecb}{}^d$, where $\veca \cdot \vecx = \sum_{i\in [n]}a_ix_i \in \R[\vecx]^{=1}$.
    \end{itemize}
\begin{lemma}\label{lemma:prelims_bomb_der}
	Let $p\in \R[\vecx]^{=d}$ be any polynomial, where $\vecx = (x_1,\dots,x_n)$. Then, we have that 
	\[ \sum_{i=1}^n \norm{\partial_{x_i}p}_B^2 = d^2 \norm{p}_B^2 . \]
\end{lemma}
\begin{proof}
	Let $p(\vecx) = \sum_{\vecalpha \in \N_d^n} c_{\vecalpha}\vecx^{\vecalpha}$.
	Then, $\norm{p}_B^2 = \sum_{\vecalpha} c_{\vecalpha}^2\cdot \frac{\vecalpha!}{d!}$, and
	\[\sum_{i\in [n]}\norm{\partial_{x_i}p}_B^2 = \sum_{i\in [n]}\sum_{\vecalpha: \alpha_i>0} (\alpha_i c_{\vecalpha})^2 \cdot \frac{\vecalpha!/\alpha_i}{(d-1)!} =\sum_{i\in [n]} \sum_{\vecalpha} c_{\vecalpha}^2\cdot \frac{\vecalpha!}{d!}\cdot (d\alpha_i) = d^2 \cdot \sum_{\vecalpha} c_{\vecalpha}^2\cdot \frac{\vecalpha!}{d!} . \qedhere\]
\end{proof}

\section{Robust Recovery from Scaling Maps (RRSM)}\label{sec:rrsm}

In this section, we look at a special case of the robust vector space decomposition problem, involving linear maps that correspond to scaling the component subspaces.
This shall later be used in our algorithm for the general robust vector decomposition problem as a sub-routine.

Let $\vecU = (U_1, U_2, \ldots, U_s)$ be an \emph{independent} $s$-tuple of subspaces in $W$ and let $U = \inangle{\vecU} \subseteq W$.

\begin{definition}\label{defn:scaling_maps}
    {\bf Space of Scaling Maps.}
    The space of scaling maps $S(\vecU) \subseteq \Lin(U, U)$ is defined as: 
        \begin{align*}
         S(\vecU) \eqdef \{A \in \Lin(U, U) \  : \ & \exists\ \lambda_1, \lambda_2, \ldots, \lambda_s \in \R \ \\ &\text{such that~} \forall i \in [s], \vecu_i \in U_i \ \text{we have~} A \cdot \vecu_i = \lambda_i \vecu_i \}.
        \end{align*}
    In other words, in a basis of $U$ obtained by concatenating the bases of $U_i$'s, the space $S(\vecU)$ consists of block diagonal matrices wherein the $i$-th diagonal block is a scalar multiple of the identity matrix of size $\dim(U_i)$. 
    For ease of notation, we will use $S$ to denote $S(\vecU)$.
\end{definition}

We are interested in the following problem:
\begin{problem} \label{prob:rrsm} {\bf Robust Recovery from Scaling Maps (RRSM).} 
    We are given as input the integer $s$, a vector space $\tilde{U}\subseteq W$, and a space $\tilde{S}\subseteq \Lin(\tilde{U}, \tilde{U})$ of linear operators on $\tilde{U}$.
    It is known that $\dist(\tilde{U}, U)$ and $\dist(\tilde{S}, S(\vecU))$ are "small," and our goal is to \emph{efficiently} find an $s$-tuple $\tilde \vecU = (\tilde U_1, \tilde U_2, \ldots, \tilde U_s)$ of subspaces in $\tU\subseteq W$, such that (upto reordering) $\dist(\tilde \vecU, \vecU)$ is "small."
\end{problem}

In the above formulation,  $\dist(\tilde{S}, S))$ is defined as follows: We extend the space $S$ (resp. $\tilde S$) to be a subspace of $\Lin(W,W)$, by extending each $A\in S$ (resp. $A\in \tilde S$) to be zero on $\orth{U}$ (resp. $\orth{\tilde U}$). Then, the distance between $\tilde {S}$ and $S$ is defined using the Hilbert-Schmidt inner product on $\Lin(W, W)$.


\subsection{RRSM: Algorithm for the Noiseless Case}\label{ref:subsec_rrsm_exact}

First, we consider the noiseless case of Problem~\ref{prob:rrsm}, in which $\tilde U = U$ and $\tilde {S} = S$ are exactly known, and we wish to recover $U_1,\dots,U_s$ exactly.

In this case it is trivially easy to recover the constituent $U_i$'s (up to permutation):
\begin{description}
    \item {\bf Exact Algorithm 1:} pick a random $A \in S$, diagonalize it and output the eigenspaces corresponding to distinct eigenvalues of $A$.
    For a random $A$, eigenvalues corresponding to distinct $U_i$'s will be distinct with high probability, and the algorithm answers correctly.
\end{description}

Next, we will give another (slightly more complicated) algorithm for this exact case.
We expect the robust version of this algorithm to have better tolerance to noise (when the noise is random) compared to the algorithm described above.

\begin{definition}\label{defn:rrsm_proj_maps} {\bf Projection Maps.}
    For each $i\in [s]$, the projection map $P_i\in \Lin(U,U)$ is defined as \[P_i \cdot (\vecu_1+\dots+\vecu_s) = \vecu_i,\] where $\vecu_j\in U_j$ for each $j\in[s]$.
    Further, we define the map $\hat{M}:\R^s\to\Lin(U,U)$ by \[\hat{M}(\lambda_1,\dots,\lambda_s) = \sum_{i=1}^s\lambda_iP_i.\]
    Note that the maps $P_i$'s do not correspond to orthogonal projections if the spaces $U_i$'s are not orthogonal to each other.
\end{definition}

Observe that:
\begin{enumerate}
    \item The space of scaling maps $S=\text{span}(P_1, \dots, P_s)$ in $\Lin(U, U)$.
    \item Recovering $\vecU = (U_1,\dots, U_s)$ is equivalent to recovering $(P_1, \dots, P_s)$.
\end{enumerate}

Next, given any $A\in S$, we consider the action (by left multiplication) of this map $A$ on the space $S$ itself:

\begin{definition}\label{defn:rrsm_map_on_S}
    Given any map $A\in S \subseteq \Lin(U,U)$, we define the map $\hat{A} \in \Lin(S, S)$ by $\hat{A}\cdot B = A\cdot B$ for all $B\in S$.

    This map $\hat{A}$ is well-defined since the space $S$ is closed under composition of maps, and, if $A = \sum_{i=1}^s \lambda_i P_i$, then $\hat{A}$ has eigenvalues $\lambda_1,\dots,\lambda_s$ with eigenvectors $P_1,\dots,P_s$ respectively.
\end{definition}

\begin{description}
    \item {\bf Exact Algorithm 2:} pick a random $A\in S$, compute the map $\hat{A}\in \Lin(S,S)$, and diagonalize it.
    With high probability, its eigenvectors are $P_1, \dots, P_s$ (appropriately scaled), and the spaces $U_1,\dots,U_s$ are the images of these maps.
\end{description}


\subsection{RRSM: Algorithm for the Robust Case}\label{ref:subsec_rrsm_robust}

In the robust case, we are given a subspace $\tilde U$ which is "close" to $U$, and a space $\tilde S\subseteq \Lin(\tilde U, \tilde U)$ which is "close" to $S$, and we wish to \emph{efficiently} recover a tuple of subspaces $\tilde \vecU$ close to $\vecU$.
Formally, for the time complexity analysis, we shall assume that the input to the algorithm is given as follows:
Let $\dim(W)=n, \dim(U)=d,\ \dim(S)=s\leq d$.
The vector space $\tU\subseteq W$ is given as $nd$ field elements, consisting of an orthonormal basis of $\tU$ with respect to some fixed orthonormal basis of $W$.
The vector space $\tS\subseteq \Lin(\tU,\tU)$ is given as $sd^2$ field elements, consisting of an orthonormal basis of $\tS$ with respect to the above orthonormal basis of $\tU$.
The total input size is $N=nd+sd^2$.

Throughout this section, we shall let $d^* = \max_{i\in [s]} \dim(U_i)$ and $ d_* = \min_{i\in [s]} \dim(U_i)$.
Before giving an algorithm, we note that the performance of our algorithm will depend on how "well-separated" the component  $U_i's$ are.
For this purpose, we shall be interested in two condition numbers, namely $\kappa(\vecU)$ and $\kappa(\hat{M})$ (see Definition~\ref{defn:rrsm_proj_maps}).
These satisfy the following relations:

\begin{lemma}\label{lemma:rrsm_simp_bound}
   \[\lnorm{\hat{M}}\leq \kappa(\vecU)\cdot\sqrt{d^*}, \quad \kappa(\hat{M})\leq \kappa(\vecU)\cdot\sqrt{\frac{d^*}{d_*}}.\]
\end{lemma}
\begin{proof}
    We defer the proof of this Lemma to Section~\ref{subsec:rrsm_analysis_simplified_bounds}.
\end{proof}

We note that while the inequalities in the above lemma may be tight in the worst-case, we expect $\kappa(\hat{M}) \ll \kappa(\vecU)$ in practice, since it sort of measures the "average separation" between the component subspaces $U_i's$.


Next, we shall give robust versions of both the exact case-algorithms.
We note that the performance of the robust version of the first exact-case algorithm depends on $\kappa(\vecU)$, whereas that of the second case exact-case algorithm depends on $\kappa(\hat{M})$.
As stated above, we expect $\kappa(\hat{M}) \ll \kappa(\vecU)$, and so the first algorithm is expected to be worse than the second algorithm.
For this reason, we will only analyze the second algorithm formally in this work.
It also turns out that the second algorithm is technically a bit easier to analyze: as we will see, it only requires perturbation bounds on eigenvectors corresponding to simple eigenvalues (see Lemma~\ref{lemma:eigenvec_pert}).

\begin{description}
    \item {\bf Robust Algorithm 1:}\label{alg1} pick a "random" (suitably defined) map $\tA \in \tS$ , "cluster" the eigenvalues which are close together into $s$ clusters, and output the eigenspaces corresponding to each cluster. We expect that if $S$ and $\tS$ are "sufficiently close" then $\vecU$ and $\tvecU$ are "fairly close".
\end{description}

Next, we will describe the robust version of the second algorithm for the exact case.
First, we give an (approximate) analogue of Definition~\ref{defn:rrsm_map_on_S}.
\begin{definition}\label{defn:rrsm_map_on_tS}
    Given any map $\tA\in \tS \subseteq \Lin(\tU,\tU)$, we define the map $\hat{\tilde{A}} \in \Lin(\tS, \tS)$ by $\hat{\tilde{A}} \cdot \tB = \Proj_{\tS}\inparen{\tA\cdot \tB}$, where $\Proj_{\tS}: \Lin(\tU, \tU)\to \tS$ is the orthogonal projection onto $\tS$.
\end{definition}

\begin{description}
    \item {\bf Robust Algorithm 2:} pick a random $\tA\in \tS$, compute the map $\hat{\tilde{A}} \in \Lin(\tS,\tS)$, and diagonalize it.
    Let its eigenvectors be $\tilde P_1, \dots, \tilde P_s \in \tS$\footnote{\label{footnote:real_eigenvec_eigenval} We will show in the analysis that with high probability the eigenvalues of $\hat{\tilde{A}}$ are real and distinct, and hence the eigenvectors lie in the real vector space $\tS$.}.
    For each $i\in[s]$, let $\tU_i \subseteq \tU$ be the span of the left singular vectors of the map $\tilde{P}_i$, with singular values "not too small."
    Output $\tilde\vecU = (\tU_1, \dots, \tU_s)$.
\end{description}

The above algorithm is formally described as Algorithm~\ref{alg:rrsm} and it gets the following guarantees:

\begin{theorem}\label{thm:rrsm}{\bf Robust Recovery from Scaling Maps.}\quad
	Let $\vecU = (U_1,\dots,U_s)$ be an independent $s$-tuple of subspaces in a vector space $W$, and let $U = \inangle{\vecU}\subseteq W$.
	Let $S = S(\vecU) \subseteq \Lin(U,U)$ be the space of scaling maps as defined in Definition~\ref{defn:scaling_maps}, and let the map $\hat{M}$ be as defined in Definition~\ref{defn:rrsm_proj_maps}.
	
    Let $\tU\subseteq W$ and $\tS\subseteq \Lin(\tU, \tU)$ be vector spaces, and let $\tau\in(0,1)$ be such that:
    \begin{enumerate}
        \item $\dist(\tU, U)<1$.
        \item $\dist(\tS, S) \leq \epsilon < 1$, where the distance is measured after extending both $S, \tS$ to subspaces of $\Lin(W,W)$.
        \item The parameter $\tau$ satisfies $\frac{1}{3}\cdot\frac{1}{\lnorm{\hat{M}}}  < \tau \leq \frac{2}{3}\cdot\frac{1}{\lnorm{\hat{M}}}$.
    \end{enumerate}
    
    Then, for any $\delta>0$, Algorithm~\ref{alg:rrsm}, on input $(W, \tU, \tS, \tau)$, runs in time $O(s^3+s^2d^\omega+sd^3+sd^2n) = O(N^{5/3})$, and outputs an $s$-tuple $\tvecU=(\tU_1,\dots,\tU_s)$ of subspaces in $\tU$, such that with probability at least $1-\delta$, it holds (upto reordering) that for each $i\in [s]$,
    \begin{align*}
        \dist(U_i, \tU_i) &\leq 300\cdot\kappa(\hat{M})\cdot  \lnorm{\hat M}^2\cdot  s^{2} \sqrt{s+\ln\frac{s^2}{\delta}}\cdot \frac{\eps}{\delta}\\&\leq  300 \cdot\sqrt{\frac{{d^*}^3}{d_*}}\cdot \kappa(\vecU)^3\cdot   s^{2} \sqrt{s+\ln\frac{s^2}{\delta}}\cdot \frac{\eps}{\delta}.
    \end{align*} 
\end{theorem}
\begin{proof}
    We defer the proof of the Theorem to Section~\ref{sec:rrsm_analysis}. 
\end{proof}

\begin{remark}\label{remark:param_tau}
    We notice that our algorithm uses an auxiliary parameter $\tau\in (0,1)$ and the correctness of the algorithm depends on $\tau$ lying in a correct range.
    This is fine for our purposes, since in applications one can usually check the correctness of the final solution obtained;  so it is possible to simply iterate over $\tau$, halving it in each iteration, and checking the solution obtained for correctness.
   The number of iterations is at most logarithmic in the condition number: Theorem~\ref{thm:rrsm} shows that a valid $\tau$ is encountered in at most $O\inparen{\log_2{\lnorm{\hat{M}}}} = O\inparen{\log_2\inparen{\kappa(\vecU)\cdot d^*}}$ iterations.
    As the error bounds in Theorem~\ref{thm:rrsm} depend on $\kappa(U)$, we will be mostly interested in the case where this condition number is not too large, and hence the runtime blow up is small.
    More formally, since we require $\kappa(\vecU)^3\cdot \epsilon \ll 1$, we will have a blow up of at most $O(\log_2{1/\epsilon})$.

    
        
    Note that in Algorithm~\ref{alg:rrsm}, there is another natural way to get $\tU_i$ once the map $\tP_i$ is known:  we can simply let $\tU_i$ be the span of the left singular vectors of the map $\tilde{P}_i$, corresponding to the $\dim(U_i)$ largest singular vectors.    
        However, it turns out that in applications, each of the $s$ dimensions $\dim(U_i)$'s may not be known.
        Hence, we use a single threshold parameter $\tau\in(0,1)$, and just consider all the singular vectors corresponding to singular value at least $\tau$.
\end{remark}

\begin{algorithm}[H]
    \caption{RRSM: Robust Recovery From Scaling Maps.}\label{alg:rrsm}
    \begin{algorithmic}
        \STATE \textbf{Input}: $(W, \tilde{U}, \tilde{S}, \tau)$, $\tilde{U}\subseteq W$ is a subspace of vector space $W$, and $\tilde{S} \subseteq \Lin(\tilde{U}, \tilde{U})$ is a subspace of $\dim(\tS)=s$, and $\tau\in(0,1)$.
        \STATE \textbf{Assumptions}: $\vecU = (U_1, U_2, \ldots, U_s)$ is an independent $s$-tuple of subspaces in $W$, and $U=\inangle{\vecU}, S=S(\vecU)$ are such that:
            \begin{enumerate}
                \item $\dist(\tU, U)<1$.
                \item  $\dist(\tS, S) \leq \epsilon < 1$, where the distance is measured after extending both to subspaces of $\Lin(W,W)$.
            \end{enumerate}
            
        \STATE \textbf{Output}: $s$-tuple $\tilde{\vecU} = (\tilde{U_1}, \dots, \tilde{U_s})$ of subspaces in $\tU\subseteq W$ such that $\dist(\tilde \vecU, \vecU)$ is small.
    \end{algorithmic}
        
        

    
    \begin{algorithmic}[1]
        \STATEx
        \STATE \label{algstep:rrsm_rand_mat} Pick a random element $\tilde{A} \in \tilde{S}$ with $\fnorm{\tilde{A}}=1$ as follows:
        \begin{enumerate}[label=(\alph*)]
            \item Let $\tilde \vecs_1,\dots, \tilde \vecs_s$ be any orthonormal basis of $\tilde{S}$.
            \item Pick $\vecalpha = (\alpha_1,\dots, \alpha_s)\in \R^s$, with each $\alpha_i\sim\mathcal{N}(0,1)$ chosen independently.         
            \item Let $\tilde A = \sum_{i=1}^s \frac{\alpha_i}{\lnorm{\vecalpha}} \cdot \tilde{\vecs}_i$.
        \end{enumerate}
        
        \STATE Compute the map $\hat{\tilde A}\in \Lin(\tS,\tS)$, as in Definition~\ref{defn:rrsm_map_on_tS}.
        
        \STATE \label{algstep:rrsm_eigenvec}Compute the eigen-decomposition of $\hat{\tilde{A}}$: 
        Suppose that it has eigenvectors $\tP_1, \dots, \tP_s \in \tS$, with $\fnorm{\tP_i}=1$ for each $i\in [s]$\textsuperscript{\ref{footnote:real_eigenvec_eigenval}}.


        \STATE \label{algstep:rrsm_final_subspace}For each $i\in [s]$, let $\tU_i \subseteq \tU$ be the span of all left singular vectors of $\tilde{P}_i$, with singular value at least $\tau$.

        \STATE Output $\tilde \vecU = (\tU_1, \dots, \tU_s)$.
    \end{algorithmic}
\end{algorithm}

Finally, we also show that if the error in Theorem~\ref{thm:rrsm} is small, then the recovered subspaces form a direct sum.

\begin{proposition}\label{prop:rrsm_direct_sum}
    Let $U = U_1\oplus \dots \oplus U_s \subseteq W$ be the direct sum of $s$-subspaces of a vector space $W$.
    Let $\tU \subseteq W$ be a subspace such that $\dist(U,\tU)<1$, and let $\tU_1,\dots,\tU_s\subseteq \tU$ be such that for each $i\in [s]$, $\dist(\tU_i, U_i)\leq \gamma<1$.
    If $2\gamma \sqrt{s}\cdot \kappa(\vecU) < 1$, then $\tU = \tU_1\oplus\dots\oplus\tU_s$.
\end{proposition}
\begin{proof}
    We defer the proof to Section~\ref{sec:rrsm_direct_sum}.
\end{proof}
\section{Robust Vector Space Decomposition (RVSD)}\label{sec:rvsdAlgo}

Let $W_1$ and $W_2$ be real vector spaces, and let $\vecU = (U_1, U_2, \dots, U_s)$ be an \emph{independent} $s$-tuple of subspaces in $W_1$, and let $\vecV = (V_1, V_2, \dots, V_s)$ be an \emph{independent} $s$-tuple of subspaces in $W_2$.
Let $U \eqdef \inangle{\vecU} \subseteq W_1$ and $ V \eqdef \inangle{\vecV} \subseteq W_2$, and let $\opB = (B_1, B_2, \ldots, B_m) \in (\Lin(U, V))^{m}$ be an $m$-tuple of linear operators from $U$ to $V$.
Suppose that each $U_i$ is mapped inside $V_i$ under the action of $\opB$, that is, for each $i\in [s]$, it holds that $\inangle{\opB \cdot U_i} \subseteq V_i$.

We are interested in the following problem:

\begin{problem}{\bf Robust Vector Space Decomposition (RVSD).}\label{prob:rsvd}
We are given as input the integer $s$, two vector spaces $\tilde U \subseteq W_1$ and $\tilde V \subseteq W_2$ such that $\dist(\tU,  U)$ and $\dist(\tV, V)$ are "small," and a tuple of operators $\topB = (\tB_1, \tB_2, \ldots, \tB_m) \in \Lin(\tU,\tV)^m$, such that $\topB$ is close to $\opB$.
Our goal is to \emph{efficiently} find an $s$-tuple $\tilde \vecU = (\tilde U_1, \tilde U_2, \ldots, \tilde U_s)$ of subspaces in $\tU\subseteq W_1$, such that (upto a common reordering of the components)  $\dist(\tilde\vecU, \vecU)$ is "small."

\end{problem}

\begin{remark}
    Note that in applications, it is usually sufficient to only find the tuple $\vecU$ approximately and so we frame our problem in this form.
    If one wishes to find $\vecV$ approximately, our algorithm can easily be extended to do that (see Remark~\ref{remark:vecV_find}).
\end{remark}

In the above formulation, we formally define  closeness between $\topB$ and $\opB$ as follows:

\begin{definition}\label{defn:op_closeness}
    Let $\hat{B}\in\Lin(U,V^m)$ (resp. $\hat{\tilde B} \in\Lin(\tU,\tV^m)$) be the map corresponding to $\opB$ (resp. $\topB$) defined as in Definition~\ref{defn:joint_op}.
    We say that $\opB$ and $\topB$ are $\eps$-close if $\lnorm{\hat{B}-\tilde{\hat{B}}}\leq \eps\cdot \lnorm{\hat{B}} = \eps \lnorm{\opB}$, where the difference is taken by viewing $\hat{B}, \tilde{\hat{B}}$ as elements of $\Lin(W_1,W_2^m)$ (by defining them to be zero on $\orth{U},\orth{\tU}$ respectively).
\end{definition}


The exact version of the problem (where $\tilde{U}=U$, $\tilde{V}=V$, $\topB=\opB$, and we wish to find $\vecU, \vecV$ exactly) first appeared in in~\cite{GKS20}, where they give an algorithm to solve it efficiently (in some special cases), and use it to recover individual components in a "sum of powers of low-degree polynomials."
Building on the ideas in~\cite{GKS20}, we give a very general algorithm to solve this robust version of the vector space decomposition problem.

\subsection{RVSD: Algorithm for the Noiseless Case}\label{sec:rvsd_exact}
We look at the exact version of RVSD where the spaces $\tU=U,\tV=V,\topB=\opB$ are exactly known.
We follow the sketch defined in Section~\ref{sec:ov_ckts}.

\begin{definition} {\bf Map Corresponding to the Adjoint Algebra.}\label{defn:adj_alg_map}
The \emph{adjoint algebra map} corresponding to the subspaces $U$ and $V$ and the operator tuple $\opB$, denoted by $\adjmap_{U,V}(\opB): \Lin(U, U) \times \Lin(V, V) \to \Lin(U,  V)^m$, is defined as
\[ \adjmap_{U,V}(\opB)\cdot (D, E) = ( B_1 D - E  B_1, \dots,  B_m D - E B_m). \]

For ease of notation, we shall simply use $\adjmap$ to denote $\adjmap_{U,V}(\opB)$.
\end{definition}

\begin{definition}[\textbf{Adjoint algebra},~\cite{ChistovIK97, GKS20, qiao}] \label{defn:adj_alg}
    The adjoint algebra, corresponding to the vector spaces $U, V$, and the tuple of operators $\opB$, denoted $\adjfull{\opB}{U}{V} \subseteq \Lin(U, U) \times \Lin(V, V)$, is defined to be the null space of the map $\adjmap$, given by 
        $$ \adjfull{\opB}{U}{V} \eqdef \setdef{(D, E)}{ B_j \cdot D = E \cdot B_j \ \text{for all~} j\in[m]}.  $$ 

    For ease of notation, we shall simply use $\adj$ to denote $\adjfull{\opB}{U}{V}$.
    Also, we shall use $\adj_1 \subseteq \Lin(U, U)$ to denote the projection of $\adj$ onto the "$U$-part," formally defined as:
    \[ \adj_1 \eqdef \setdef{D}{\exists E \text{ such that } (D,E)\in \adj}.\]
\end{definition}

\begin{definition}
    We define $\vecU\times \vecV$ to be the $s$-tuple of subspaces given by $\vecU\times \vecV = (U_1\times V_1,\dots,U_s\times V_s)$, which satisfies $\inangle{\vecU\times \vecV} = U\times V$. 
\end{definition}

Observe that that the space of scaling maps $S(\vecU\times \vecV)$ (see Definition~\ref{defn:scaling_maps}) is always contained in the adjoint algebra $\adj$, that is, $S(\vecU\times \vecV)\subseteq \adj$.
The next proposition shows that if the two are equal, then the decomposition of $U$ and $V$ into the $s$ components is unique, and further indecomposable.

\begin{proposition}[Proposition A.3 in~\cite{GKS20}]\label{prop:adj_diag_unique_decomp}
    If $\dim(\adj) = s$, or equivalently, if $\adj=S(\vecU\times \vecV)$,  then:
    \begin{description}
        \item The decomposition $U = U_1\oplus\dots\oplus U_s, V = V_1\oplus\dots\oplus V_s$ is the unique irreducible decomposition satisfying $\inangle{\opB \cdot U_i} \subseteq V_i$ for each $i\in [s]$. That is, if \[ U = \hat{U}_1 \oplus \hat{U}_2 \oplus \ldots\oplus \hat{U}_{\hat{s}} \quad \text{and~} V = \hat{V}_1 \oplus \hat{V}_2 \oplus \ldots \oplus \hat{V}_{\hat{s}}, \quad \quad \hat{s}\geq s,\]
        and
        \[\forall i \in [\hat{s}]\ \inangle{\opB \cdot \hat{U}_i} \subseteq \hat{V}_i,\]
            then $\hat{s}= s$ and upto reordering if necessary, $\hat{U}_i=U_i$ and $\hat{V}_i=V_i$ for all $i \in [s]$.
    \end{description}
    Also, the above implies that $\adj_1 = S(\vecU)$, and that $V_i = \inangle{\opB \cdot U_i}$ for each $i\in [s]$.
\end{proposition}

Based on the above proposition, the following simple algorithm recovers the $U_i$'s and the $V_i$'s under the assumption that $\adj = S(\vecU\times \vecV)$:
\begin{description}
    \item {\bf Exact Algorithm 1:} Compute the adjoint algebra $\adj$ (by solving a system of linear equations). Run the exact algorithm for RRSM (see Section~\ref{ref:subsec_rrsm_exact}) on $\adj$ with respect to the space $U\times V$; suppose that the output is $T_1,\dots, T_s\subseteq U\times V$. For each $i\in [s]$, let $U_i,V_i$ be the projection of $T_i$ on the "$U$, $V$-parts" respectively.
\end{description}

We also give a slightly different variation of the above algorithm, which turns out to be easier to analyze in the robust case.

\begin{description}
    \item {\bf Exact Algorithm 2:} Compute the adjoint algebra $\adj$, and then compute $\adj_1$. Run the exact algorithm for RRSM on $\adj_1$ with respect to the space $U$; the output is $(U_1,\dots, U_s)$. For each $i\in [s]$, compute $V_i = \inangle{\opB \cdot U_i}$.
\end{description}

\subsection{RVSD: Algorithm for the Robust Case}\label{subsec:rvsd_noisy}

Next, we give a robust version of the (second) exact-case algorithm.
We are given a subspaces $\tilde U\approx U,\tilde V\approx V$, and an $m$-tuple of operators $\tilde \opB\approx\opB$, and we wish to recover a tuple of subspaces $\tilde \vecU\approx\vecU$.
Formally, for the time complexity analysis, we shall assume that the input to the algorithm is given as follows:
Let $\dim(W_1)=n_1,\ \dim(W_2)=n_2,\ \dim(U)=d_1,\ \dim(V)=d_2$.
The vector space $\tU\subseteq W_1$ (resp. $\tV\subseteq W_2$) is given as $n_1d_1$ (resp. $n_2d_2$) field elements, consisting of an orthonormal basis of $\tU$ (resp. $\tV)$, with respect to some fixed orthonormal basis of $W_1$ (resp. $W_2$).
The $m$-tuple of operators $\opB=(\tB_1,\dots,\tB_m)\in \Lin(\tU,\tV)^m$ is given as $md_1d_2$ field elements, with each $\tB_i$ given as a matrix with respect to the above orthonormal basis of $\tU$ and $\tV$.
The total input size is $N=n_1d_1+n_2d_2+md_1d_2$.

\begin{definition} {\bf (Approximate) Map Corresponding to the Adjoint Algebra.}\label{defn:tadj_alg_map}
The (approximate) \emph{adjoint algebra map} corresponding to the subspaces $\tilde{U}$ and $\tilde{V}$ and the operator tuple $\tilde \opB$, denoted by $\tadjmap_{\tU,\tV}(\topB): \Lin(\tilde U, \tilde U) \times \Lin(\tilde V, \tilde V) \to \Lin(\tilde U, \tilde V)^m$, is defined as
\[ \tadjmap_{\tU,\tV}(\topB)\cdot (D, E) = (\tB_1 D - E \tB_1, \dots, \tB_m  D - E \tB_m). \]

For ease of notation, we shall use $\tadjmap$ to denote $\tadjmap_{\tU,\tV}(\topB)$.
\end{definition}

\begin{definition} {\bf (Approximate) Adjoint Algebra.}\label{defn:approx_adj}
The (approximate) \emph{adjoint algebra} corresponding to the subspaces $\tilde{U}$ and $\tilde{V}$ and the operator tuple $\tilde \opB$, denoted by $\widetilde{\adjfull{\opB}{U}{V}} \subseteq \Lin(\tilde U, \tilde U) \times \Lin(\tilde V, \tilde V)$, is defined to be the vector space spanned by the right singular vectors of the map $\tadjmap$, corresponding to the $s$ smallest singular values.
Note that the relevant inner product on $\Lin(\tilde U, \tilde U) \times \Lin(\tilde V, \tilde V)$ is the direct sum of the two inner products in the natural way. 

For ease of notation, we shall simply use $\tadj$ to denote $\widetilde{\adjfull{\opB}{U}{V}}$.
Note that by definition, the dimension of $\tadj$ is equal to $s$.

Further, we define $\tadj_1 \eqdef \setdef{D}{\exists E \text{ such that } (D,E)\in \tadj}$.


\end{definition}


Based on the above definitions, we next give an algorithm for the RVSD problem.
It relies on the problem of Robust Recovery from Scaling Maps (RRSM), which we discussed in Section~\ref{sec:rrsm}.

\begin{description}
    \item {\bf Robust Algorithm:} Compute the map $\tadjmap$ and the adjoint algebras $\tadj$ and $\tadj_1$. Run RRSM algorithm on $\tadj_1$ with respect to the space $\tU \subseteq W_1$; let the output be $(\tU_1,\dots,\tU_s)$.
    Output $\tilde{\vecU} = (\tilde{U}_1, \ldots \tilde{U}_s)$.
\end{description}

The above algorithm is formally described as Algorithm~\ref{alg:rvsd} and it gets the following guarantees:

\begin{theorem}\label{thm:rvsd}{\bf Robust Vector Space Decomposition.}\quad
Let $\vecU=(U_1,\dots,U_s)$ and $\vecV=(V_1,\dots,V_s)$ be independent $s$-tuple of subspaces in vector spaces $W_1$ and $W_2$ respectively, and let $U=\inangle{\vecU}, V=\inangle{\vecV}$.
Let $\opB=(B_1,\dots,B_m)\in\Lin(U,V)^m$ be an $m$-tuple of operators such that for each $i\in[s], \inangle{\opB\cdot U_i}\subseteq V_i$.

Let the map $\hat{M}:\R^s\to \Lin(U,U)$ corresponding to the space $U$ be defined as in Definition~\ref{defn:rrsm_proj_maps}.
Let $\sigma_{-(s+1)}(\adjmap)$ denote the $(s+1)\textsuperscript{th}$-smallest singular value of the adjoint-algebra map $\adjmap$ (see Definition~\ref{defn:adj_alg_map}), and let $\lnorm{\opB}$ be defined as in Definition~\ref{defn:joint_op}.

Suppose that $\tU\subseteq W_1, \tV\subseteq W_2, \topB\in\Lin(\tU, \tV)^m$, and $\tau\in (0,1)$ are such that:
\begin{enumerate}
    \item $\dist(U,\tU)\leq \eps_1<1$, $\dist(V,\tV)\leq \eps_2<1$, and that the operator tuples $\opB, \topB$ are $\eps$-close (see Definition~\ref{defn:op_closeness}).
    \item $\dim(\adj) = s$ (see Definition~\ref{defn:adj_alg}).
    \item The parameter $\tau$ satisfies $\frac{1}{3}\cdot\frac{1}{\lnorm{\hat{M}}}  < \tau \leq \frac{2}{3}\cdot\frac{1}{\lnorm{\hat{M}}}$.
\end{enumerate}

Then, for any $\delta>0$ we have that Algorithm~\ref{alg:rvsd}, on input $(W_1,W_2,s,\tU,\tV,\topB,\tau)$, runs in time $O(\inparen{md_1d_2\cdot (d_1^2+d_2^2)}^3 + sd_1^2n_1) = O(N^6)$, and outputs an $s$-tuple $\tvecU=(\tU_1,\dots,\tU_s)$ of subspaces in $\tU$, such that with probability at least $1-\delta$, it holds (upto reordering) that for each $i\in [s]$,
    \begin{align*}
        \dist(U_i, \tU_i) &\leq 1800\cdot\kappa(\hat{M})\cdot  \lnorm{\hat M}^2\cdot  s^{2} \sqrt{s+\ln\frac{s^2}{\delta}}\cdot \frac{\eps+\eps_1+\eps_2}{\delta}\cdot\frac{\lnorm{\opB}}{\sigma_{-(s+1)}(\adjmap)}
        \\&\leq 1800\cdot\sqrt{\frac{{d^*}^3}{d_*}}\cdot \kappa(\vecU)^3\cdot  s^{2} \sqrt{s+\ln\frac{s^2}{\delta}}\cdot \frac{\eps+\eps_1+\eps_2}{\delta}\cdot\frac{\lnorm{\opB}}{\sigma_{-(s+1)}(\adjmap)}.
    \end{align*}

    Furthermore, if $2\gamma \sqrt{s}\cdot \kappa(\vecU) < 1$, then $\tU = \tU_1\oplus\dots\oplus\tU_s$, where $\gamma$ denotes the above error bound.
\end{theorem}
\begin{proof}
    We defer the proof of the above theorem to Section~\ref{sec:rvsd_analysis}.
\end{proof}

\begin{remark}\label{remark:rvsd_rem}
    Note that the assumption that the integer $s$ is given to the algorithm as input is merely for simplicity.
    If $s$ is not known, its value can be determined by looking at the singular values of the map $\tilde{\adjmap}$ corresponding to the adjoint algebra: the smallest $s$ singular values are usually very close to zero, and are "much smaller" that the $(s+1)$\textsuperscript{th} smallest singular value.
    Alternatively, we can just iterate over $s$, since in applications we can usually check the correctness of the final solution obtained.
   
    Also, the algorithm having access to the correct value of $\tau$ can be handled using the iteration strategy mentioned in Remark~\ref{remark:param_tau}.
\end{remark}

\begin{remark}\label{remark:vecV_find}
    Note that if we wish to find $\vecV = (V_1,\dots,V_s)$ approximately as well, one can run the RRSM algorithm on $(W_1\times W_2, \tU\times \tV, \tadj, \tau)$, to recover approximate versions of $U_1\times V_1,\dots,U_s\times V_s$.
    In this case, the error bounds will be the same as in Theorem~\ref{thm:rvsd}, with the definition of the map $\hat{M}$ changed appropriately to $\hat{M}:\R^s\to\Lin(U\times V, U\times V) $.
\end{remark}

\begin{algorithm}[H]
    \caption{RVSD: Robust Vector Space Decomposition.} \label{alg:rvsd}
    \begin{algorithmic}
        \STATE \textbf{Input}: $(W_1, W_2, s, \tilde{U}, \tilde{V}, \tilde{\opB}, \tau)$, where $s$ is a positive integer, $\tilde{U}\subseteq W_1, \tilde{V}\subseteq W_2$ are subspaces of vector spaces $W_1, W_2$ respectively, $\topB = (\tB_1, \dots, \tB_m) \in \Lin(\tU, \tV)^m$ is an $m$-tuple of linear operators, and $\tau\in(0,1)$.
        \STATE \textbf{Assumptions}: $\vecU = (U_1, \ldots, U_s)$ and $\vecV = (V_1, \ldots, V_s)$ are independent $s$-tuples of subspaces in $W_1,W_2$ respectively, and $U=\inangle{\vecU}, V=\inangle{\vecV}$, and $\opB=(B_1,\dots,B_m)\in\Lin(U, V)^m$ are such that:
            \begin{enumerate}
                \item For each $i\in [s]$, it holds that $\inangle{\opB \cdot U_i} \subseteq V_i$.
                \item $\dist(U, \tilde{U})\leq \epsilon_1<1$ and $\dist(V, \tilde{V})\leq \epsilon_2<1$.
                \item $\opB$ and $\topB$ are $\eps$-close, according to Definition~\ref{defn:op_closeness}.
            \end{enumerate}
            
        \STATE \textbf{Output}: $s$-tuple of subspaces $\tilde{\vecU} = (\tilde{U}_1, \ldots, \tilde{U}_s)$ in $\tU\subseteq W_1$ such that $\dist(\tilde{\vecU}, \vecU)$ is small.
    \end{algorithmic}
    \begin{algorithmic}[1]
        \STATEx

        \STATE Compute a singular value decomposition of the adjoint algebra map $\tadjmap$ as defined in Definition~\ref{defn:tadj_alg_map}.

        \STATE Compute the approximate adjoint algebra $\tadj \subseteq \Lin(\tilde U, \tilde U) \times \Lin(\tilde V, \tilde V)$ and $\tadj_1 \subseteq \Lin(\tilde U, \tilde U)$, as defined in Definition~\ref{defn:approx_adj}.

        \STATE Run Robust Recovery from Scaling Maps (RRSM, Algorithm~\ref{alg:rrsm})  on $(W_1, \tilde{U}, \tadj_1, \tau)$, and let $(\tU_1, \dots, \tU_s)$ be the $s$-tuple of subspaces in $\tU\subseteq W_1$ it outputs.


        \STATE Output $\tilde{\vecU} = (\tilde{U}_1, \ldots \tilde{U}_s)$.
    \end{algorithmic}
\end{algorithm}

\subsection{RVSD: Using a Common Tuple of Operators on a Larger Space}\label{subsec:rvsd_common_op}

It is often the case in applications that a tuple of operators $\opB = (B_1,\dots,B_m) \in \Lin(W_1,W_2)^m$ are known exactly, and these satisfy $\inangle{\opB \cdot U_i}\subseteq V_i$ for each $i\in [s]$.
In this case, we work with the relevant projections of these operators to $\Lin(U,V)$ and $\Lin(\tU,\tV)$ respectively.

\begin{definition} {\bf Projected Tuple of Operators.}\label{defn:restr_ops}
	Let $\opB = (B_1,\dots,B_m) \in \Lin(W_1,W_2)^m$ be an $m$-tuple of operators from $W_1$ to $W_2$.
	For subspaces $U\subseteq W_1, V\subseteq W_2$, we define the operator tuple $\opC = (C_1,\dots,C_m)\in\Lin(W_1,W_2)^m$ as follows: For each $j\in[m]$,
	\[C_j = \Proj_V \cdot B_j \cdot \Proj_U,\]
where $\Proj_{U}:W_1 \to W_1,\ \Proj_{V}:W_2 \to W_2$ are the orthogonal projection maps onto $U,V$ respectively.
Observe that each $C_j$ maps the space $U$ into $V$, and is the zero map on $\orth{U}$.
\end{definition}

\begin{lemma}\label{lemma:restr_op_close}
Let $\opB = (B_1,\dots,B_m) \in \Lin(W_1,W_2)^m$ be an $m$-tuple of operators from $W_1$ to $W_2$, and let $U,\tU\subseteq W_1$ and $V,\tV\subseteq W_2$ be subspaces satisfying $\dist(U,\tU)\leq\eps_1$ and $\dist(V,\tV)\leq\eps_2$.
Let $\opC$ (resp. $\topC$) be the projected tuple of operators with respect to $U,V$ (resp. $\tU,\tV$) according to Definition~\ref{defn:restr_ops}.

Then, $\lnorm{\opC} \leq \lnorm{\opB}$, and $\opC$ and $\topC$ are $\eps$-close (see Definition~\ref{defn:op_closeness}), for $\eps = (\eps_1+\eps_2)\cdot \frac{\lnorm{\opB}}{\lnorm{\opC}}$.
\end{lemma}

\begin{proof}
    By Definition~\ref{defn:joint_op} and Definition~\ref{defn:restr_ops}, we have for any $\vecw \in W_1$,
    \[\lnorm{(\hat{C}-\hat{\tilde C})\cdot \vecw} = \lnorm{\inparen{\inparen{\Proj_V \cdot B_i \cdot \Proj_U-\Proj_{\tV} \cdot B_i \cdot \Proj_{\tU}}\cdot \vecw}_{i = 1}^m}\]
    Then, using $\lnorm{\Proj_U-\Proj_{\tU}}\leq \eps_1$, $\lnorm{\Proj_V-\Proj_{\tV}}\leq \eps_2$, and the triangle inequality, we get that 
    \[ \lnorm{\hat{C}-\hat{\tilde C}} \leq (\eps_1+\eps_2)\cdot \lnorm{\opB}.\]
    Similarly, the fact $\lnorm{\opC} \leq \lnorm{\opB}$ follows easily.
\end{proof}

Now, Lemma~\ref{lemma:restr_op_close}, combined with Theorem~\ref{thm:rvsd}, immediately gives us the following corollary:

\begin{corollary}\label{corr:rvsd_common_op}
Suppose that a tuple of operators $\opB = (B_1,\dots,B_m) \in \Lin(W_1,W_2)^m$ is known exactly, where for each $i\in [s]$, $\inangle{\opB\cdot U_i} \subseteq V_i$.
Then, under the assumptions of Theorem~\ref{thm:rvsd}, Algorithm~\ref{alg:rvsd} can recover the tuple $\vecU$ of subspaces, upto error 
\begin{align*}
\gamma &= 3600 \cdot \kappa(\hat{M})\cdot  \lnorm{\hat M}^2\cdot s^{2} \sqrt{s+\ln\frac{s^2}{\delta}}\cdot \frac{\eps_1+\eps_2}{\delta}\cdot\frac{\lnorm{\opB}}{\sigma_{-(s+1)}(\adjmap)}
 \\&\leq 3600 \cdot \sqrt{\frac{{d^*}^3}{d_*}}\cdot \kappa(\vecU)^3\cdot  s^{2} \sqrt{s+\ln\frac{s^2}{\delta}}\cdot \frac{\eps_1+\eps_2}{\delta}\cdot\frac{\lnorm{\opB}}{\sigma_{-(s+1)}(\adjmap)},
\end{align*}
where $\lnorm{\opB}$ is as in Definition~\ref{defn:joint_op} (with respect to $W_1$ and $W_2$), and the adjoint algebra map $\adjmap$ is defined by the projection of operators in $\opB$ onto $\Lin(U,V)$ (see Definition~\ref{defn:restr_ops}). 

Furthermore, if $2\gamma \sqrt{s}\cdot \kappa(\vecU) < 1$, then $\tU = \tU_1\oplus\dots\oplus\tU_s$.
\end{corollary}



\bibliographystyle{alpha}
\bibliography{references}

\newpage
\appendix

\section{Linear Algebra and Probability}

\subsection{Matrices and Subspaces}\label{subsec:app_matrix}

\subsubsection{Matrix Norms}

The following are some easy to verify properties of the Frobenius norm:
\begin{proposition}\label{prop:matrix_norms}

    ~\begin{enumerate}
        \item (Sub-multiplicativity) For any $m\times n$ matrix $M$, and $n\times p$ matrix $N$, it holds $\fnorm{MN}\leq \fnorm{M}\cdot \lnorm{N}\leq \fnorm{M}\cdot \fnorm{N}.$
        \item Let the $n\times n$ matrix $M$ have eigenvalues $\lambda_1,\dots, \lambda_n$. Then, 
        $\fnorm{M}^2 \geq \sum_{i\in [n]} \inabs{\lambda_i}^2.$
    \end{enumerate}
\end{proposition}

\subsubsection{Pseudo-Inverse}\label{subsubsec:pseudo_inv}

Recall that for a matrix $M$, we use $M^\dag$ to denotes its Moore-Penrose pseudo-inverse.
This satisfies the following easy to check properties:
\begin{proposition}
    For any $m\times n$ matrix $M$ with $\rank(M)=n\leq m$, it holds that:
    \begin{enumerate}
        \item  $M^\dag = (M^\top M)^{-1}M^\top$, where $M^\top$ is the transpose of $M$.
        \item $M^\dag M$ equals the identity matrix of size $n$.
        \item $M^\dag (M^\dag)^\top = (M^\top M)^{-1}$.
        \item $\lnorm{M^\dag} = 1/\sigma_{n}(M)$, where $\sigma_n(M)$ denotes the $n\textsuperscript{th}$ largest (or the smallest non-zero) singular value of $M$.
    \end{enumerate}
\end{proposition}

We will need also the following lemma about pseudo-inverses.

\begin{lemma}\label{lemma:pseudo_inv_same_col_space}
    Let $A,B$ be $n\times r$ matrices be matrices with $\rank(A)=\rank(B)=r\leq n$, and with the same column space.
    Then, $A^\dag B  \in \R^{r\times r}$ is invertible, and $(A^\dag B)^{-1} = B^\dag A$.
\end{lemma}
\begin{proof}
    Fix any $\vecx\in \R^r$. Then, there is a unique $\vecy\in \R^r$ such that $A\vecx=B\vecy$, and it holds that 
    \[ (A^\dag B)\cdot (B^\dag A)\cdot \vecx = (A^\dag B)(B^\dag B)\vecy = (A^\dag B)\cdot \vecy = A^\dag A \vecx = \vecx. \]
    Hence, $(A^\dag B)\cdot (B^\dag A)$ must be the identity matrix.
\end{proof}

\subsubsection{Distances Between Subspaces}

For any pair of subspaces in $\R^n$, we can find a nice basis that relates how they are situated with respect to each other in $\R^n$, as follows:


\begin{theorem}\label{thm:cs_decomp}{\bf Canonical Decomposition and Angles Between Subspaces.} (\cite{Jor75}; see Theorem I.5.2 in~\cite{SS90})
    
    Let $U, \tU \subseteq \R^n$ be two subspaces, each of dimension $r$, and let $k = r - \dim(U\cap \tU)$.
    Then, there exists an orthonormal basis $\vece_1,\dots, \vece_k, \vece_{k+1}, \dots, \vece_{r}, \vecf_1,\dots,\vecf_k, \vech_1,\dots, \vech_{n-(r+k)}$ of $\R^n$, and angles $\frac{\pi}{2} \geq \theta_1\geq \dots \geq \theta_k > 0$, such that:
    \begin{enumerate}
        \item $\vece_1,\dots, \vece_r$ form an orthonormal basis of $U$.
        \item $\tilde{\vece}_i,\dots, \tilde{\vece}_r$ form an orthonormal basis of $\tU$, where for each $i\in [r]$, $$\tilde{\vece}_i = \begin{cases} \cos(\theta_i)\cdot \vece_i + \sin(\theta_i)\cdot \vecf_i,&  i\leq k \\ \vece_i ,& o/w \end{cases}.$$
    \end{enumerate}
    Furthermore, it holds that:
    \begin{enumerate}
        \item $\dist(U,\tU) = \lnorm{\Proj_{U}-\Proj_{\tU}} = \sin(\theta_1)$.
        \item Let $\alpha_1,\dots, \alpha_r\in \R$ be such that $\sum_{i=1}^k \inabs{\alpha_i}^2 = 1$, and let $u = \sum_{i=1}^r \alpha_i \cdot \vece_i  \in U,\ \tilde u = \sum_{i=1}^r \alpha_i \cdot \tilde{\vece}_i \in \tU$.
        Then, it holds that $\lnorm{u - \tilde{u}} \leq 2\sin\inparen{\frac{\theta_1}{2}} \leq 2\sin(\theta_1) = 2\cdot \dist(\tU, U)$.
    \end{enumerate}
    Note that for consistency, we define $\theta_1=0$ if $k=0$.
\end{theorem}

\begin{remark}\label{remark:cs-decomposition}
    Theorem \ref{thm:cs_decomp} implies that, given subspaces $U, V$ of $\R^n$ of dimensions $r$, there are orthogonal matrices $P, Q$ such that $U = \inangle{P}, V = \inangle{Q}$ and $P^T Q = C$, where $C$ is a diagonal matrix with entries $C_{ii} = \cos(\theta_i)$. In the case when the dimensions of the subspaces are different, $C$ is a rectangular matrix of shape $\dim(U) \times \dim(V)$, the principal diagonal having elements $\cos(\theta_1), \ldots, \cos(\theta_r)$ where $r = \min(\dim(U),\dim(V))$.
\end{remark}


We will need also the following lemma:


\begin{lemma}\label{lemma:subspace_subspace_dist}
    Let $U,\tU\subseteq \R^n$ be subspaces of dimension $d$ such that for each $\tilde{\vecu}\in \tU$, $\norm{\vecu}=1$, $\dist(\tilde{\vecu},U) \leq \eta$.
    Then, $\dist(U,\tU)\leq \eta$.
\end{lemma}
\begin{proof}
    Consider the Canonical Decomposition as in Theorem~\ref{thm:cs_decomp}.
    We have that \[\dist(U,\tU) = \sin(\theta_1) = \lnorm{\tilde{\vece}_1-\cos(\theta_1)\cdot \vece_1 } = \dist(\tilde{\vece}_1, U) \leq \eta. \qedhere\]
\end{proof}


\subsection{Matrix Perturbation Bounds}

\subsubsection{Perturbation bounds for Singular values and Singular vectors}\label{sec:singularPerturb}

If we {\em slightly} perturb an $m\times n$ matrix $A$ to obtain $\tilde{A} = A + E$, how "far" are the 
singular values and singular spaces of $\tilde{A}$ from those of $A$? A satisfactory answer (which is asymptotically 
the best possible in the worst case) is provided by a theorem due to Wedin.  

\begin{lemma}\label{lemma:weyl_ineq}(Weyl's Inequality \cite{Weyl12})
    Let $A, E$ be $m\times n$ matrices with $m\geq n$, and let $\tA=A+E$.
    Then, for each $i\in [n]$, it holds that
    \[\inabs{\sigma_i(\tA)-\sigma_i(A)}\leq \lnorm{E},\]
    where $\sigma_i$ denotes the $i\textsuperscript{th}$ largest singular value.
\end{lemma}

\begin{lemma} \label{th:wedin} {\bf Perturbation of singular spaces}  (Wedin~\cite{Wedin72}; see Theorem V.4.4 in \cite{SS90}).
Let $A, E$ be $m\times n$ matrices with $m\geq n$, and let $\tA = A + E$, and suppose their singular value decomposition is as follows:
\[  
A = \insquare{\begin{array}{ccc}
        U_1 &  U_2 & U_3
    \end{array}} 
    \insquare{\begin{array}{cc}
        \Sigma_1 &  0 \\
        0 & \Sigma_2 \\
        0 & 0
    \end{array}} 
    \insquare{\begin{array}{c}
        V_1^\top \\
        V_2^\top
    \end{array}}, \quad
\tilde{A} = \insquare{\begin{array}{ccc}
    \tU_1 &  \tU_2 & \tU_3
\end{array}} 
\insquare{\begin{array}{cc}
    \tilde{\Sigma_1} &  0 \\
    0 & \tilde{\Sigma_2} \\
        0 & 0
\end{array}} 
\insquare{\begin{array}{c}
    \tV_1^\top \\
    \tV_2^\top
\end{array}} 
\]
where the first block of the above decomposition corresponds to the top $r$ singular values and the second block to the bottom $(n-r)$ singular values. 

Let $\delta>0$ be such that $\delta \leq \min_{i\in [r]} (\Sigma_1)_{i,i} - \max_{j\in [n-r]} (\Sigma_2)_{j,j}$.
Then, it holds that
\[ \dist(\inangle{U_1}, \inangle{\tU_1}) \leq \frac{2\lnorm{E}}{\delta},\quad  \dist(\inangle{V_1}, \inangle{\tV_1}) \leq \frac{2\lnorm{E}}{\delta}. \]

\end{lemma}
\begin{proof}
    Theorem V.4.4 in \cite{SS90} shows that the above distances are bounded by $\frac{\lnorm{E}}{\eta}$, where $\eta \eqdef \min_{i\in [r]} (\Sigma_1)_{i,i} - \max_{j\in [n-r]} (\Sigma_2)_{j,j}$.
    The result then follows from Weyl's inequality (Lemma~\ref{lemma:weyl_ineq}), which says that for each $i\in [r]$, $\inabs{(\Sigma_1)_{i,i}-(\tilde{\Sigma}_1)_{i,i}}\leq \lnorm{E}$, and so $\eta \geq \delta - \lnorm{E}$.
    
    We remark that result over $\C$ (as is proven in~\cite{SS90}) also implies the same over $\R$.
\end{proof}



As an immediate corollary we have:

\begin{corollary}\label{cor:sing_nullspacePerturb} (lemma G.5 in \cite{GHK15})
Let $A, E$ be $m \times n$ matrices with $m \geq n$. Suppose that $A$ has rank $r$ and the smallest (non-zero) singular value of 
$A$ is given by $\sigma_{r}(A)$. Let $S, \tS$ (resp. $T,\tT$) be the subspaces spanned by the top $r$ right (resp. left) singular vectors of $A$ and 
$\tA = A + E$ respectively. 
Then we have:
    $$ \dist(S,\tS) \leq  \frac{2\lnorm{E}}{ \sigma_r(A)}, \quad \dist(T, \tT) \leq  \frac{2\lnorm{E}}{ \sigma_r(A)}. $$
\end{corollary}


The above bounds are hold true for even when the perturbation matrix $E$ is chosen in a worst-case/adversarial fashion
and can be rather pessimistic. In many applications, $E$ is more like a random matrix in which the case the perturbation of singular 
values and spaces will be significantly less. 


\begin{lemma}\label{lem:randomOpNorm} (cf. \cite{Latala05, Szarek91})
For a random real matrix $E \sim (\CAL{N}\inparen{0, \rho^2})^{m \times n}$, we have that almost surely   
    $$ \lnorm{E} = \Theta \inparen{\rho \cdot \sqrt{m+n}}. $$
\end{lemma}

\begin{lemma}
For $i \in \inbrace{1, 2}$, let $S_i \subseteq \R^{n_i}$ be a subspace of dimensions $s_i \leq n_i $ and let $\Proj_{S_i} : \R^{n_i} \mapsto \R^{n_i}$
be the corresponding projection map. 
Then for a random matrix $E \sim (\CAL{N}\inparen{0, \rho^2})^{n_2 \times n_1}$, we have that almost surely   
    $$ \lnorm{\Proj_{S_2} \cdot E \cdot \Proj_{S_1}} = \Theta \inparen{\rho \cdot \sqrt{s_1 + s_2}}. $$
\end{lemma}
\begin{proof}
The probability density function for random Gaussian matrices is invariant under orthogonal change of basis of either the domain space 
or the target space. Consequently we can assume without loss of generality that $S_1$ (resp. $S_2$) is spanned by the first $s_1$ (resp. $s_2$) 
canonical unit vectors of $\R^{n_1}$ (resp. of $\R^{n_2}$). Then $ \lnorm{(\Proj_{S_2} \cdot E \cdot \Proj_{S_1})} = \lnorm{B}$, 
where $B \in \R^{s_2 \times s_1}$ is the top-left ($s_2 \times s_1$)-dimensional submatrix of $E$. Thus 
$B \sim (\CAL{N}\inparen{0, \rho^2})^{s_2 \times s_1}$ and hence the conclusion follows from an application of lemma \ref{lem:randomOpNorm}. 
\end{proof}


%


\subsubsection{Perturbation bounds for Eigenvalues and Eigenvectors}\label{sec:eigenPerturb}

If we {\em slightly} perturb a matrix $A \in \R^{n \times n}$ to obtain $\tilde{A} = A + E$, how "far" are the 
eigenvalues and eigenvectors of $\tilde{A}$ from those of $A$? We show quantitative bounds when all eigenvalues of $A$ are simple (have multiplicity 1), and $\lnorm{E}$ is {\em small}.

First, we state a result that shows that the eigenvalues of any complex matrix vary continuously with the error $\lnorm{E}$.
For $\lambda\in \C, \eps>0$, we define $\cD(\lambda, \eps) \eqdef \inbrace{\zeta\in \C: \inabs{\zeta-\lambda} \leq \eps}$.

\begin{theorem}\label{thm:eigenvalue_cont_pert} (See Theorem IV.1.1 in~\cite{SS90}) 
    Let $A, E\in \C^{n\times n}$ and $\tA = A+E$.
    Let $\lambda$ be an eigenvalue of $A$ with algebraic multiplicity $m$.
    Then, for any (small enough) $\eps>0$, there exists a $\delta>0$, such that if $\lnorm{E}<\delta$, the disk $\cD(\lambda, \eps)$ contains exactly $m$ eigenvalues of $\tA$.
\end{theorem}
\begin{proof}[Proof Idea]
    The theorem follows by noticing that the the characteristic polynomial of any matrix is a continuous function of the matrix entries, and then applying Rouch\'e's Theorem.
\end{proof}

\begin{theorem}\label{thm:bauer_fike} (Bauer-Fike~\cite{BF60})
    Let $A =  X \cdot \Lambda \cdot X^{-1} \in \C^{n\times n}$ be a diagonalizable matrix with $X\in \C^{n\times n}$, and $\Lambda\in \C^{n\times n}$ diagonal.
    Let $E\in \C^{n\times n}$ and $\tA = A+E$.
    Then, for each eigenvalue $\tlb$ of $\tA$, there is an eigenvalue $\lb$ of $A$ such that $\inabs{\tlb-\lb} \leq \kappa(X)\cdot \lnorm{E}$, where $\kappa(X) = \lnorm{X}\cdot \lnorm{X^{-1}}$ is the condition number of $X$.
\end{theorem}

\begin{lemma}{\bf Perturbation of Eigenvalues.}\label{lemma:eigenval_pert}
	Let $A =  X \cdot \Lambda \cdot X^{-1} \in \C^{n\times n}$ be a diagonalizable matrix, with $X\in \C^{n\times n}$, and $\Lambda = \diag(\underbrace{\lambda_1, \dots, \lambda_1}_{m_1\text{ times}},\dots,\underbrace{\lambda_k, \dots, \lambda_k}_{m_k\text{ times}})$ and $\sum_{i=1}^k m_i=n$.
    Let $\delta = \min_{i,j\in [k], i\ne j} \inabs{\lb_i-\lb_j} > 0$, and let $\kappa(X) = \lnorm{X}\cdot \lnorm{X^{-1}}$ be the condition number of $X$.

	Let $E\in \C^{n\times n}$ be such that $\kappa(X)\cdot \lnorm{E} < \delta/2$, and let $\tA = A+E$.
	Then, 
	\begin{enumerate}
		\item The eigenvalues of $\tA$ can be grouped into $k$ groups $\tlb_{1,1},\dots,\tlb_{1,m_1},\dots,\tlb_{k,1},\dots,\tlb_{k,m_k}$ such that for each $i\in [k], j\in [m_i]$, it holds that $\inabs{\tlb_{i,j}-\lb_i} \leq \kappa(X)\cdot \lnorm{E} < \delta/2$.
		\item Suppose that all of $A, E, X, \Lambda\in \R^{n\times n}$ are real matrices. Let $i\in [n]$ be such that $\lb_i$ has multiplicity 1 (that is $m_i=1$). Then, it holds that $\tlb_{i,1}$ is real. 
	\end{enumerate}

\end{lemma}
\begin{proof}
	\begin{enumerate}
		\item For each $i\in [k]$, let $\cD_i = \cD(\lb_i, \kappa(X)\cdot \lnorm{E}) \subsetneq \cD(\lb_i, \delta/2)$.
			We know that the $\cD_i$'s are disjoint from each other, and for each $i\in [k]$, $\cD_i$ contains exactly $m_i$ eigenvalues of $A_0 = A$.
			By the Bauer-Fike theorem (Theorem~\ref{thm:bauer_fike}), we know that for each $\tau\in [0,1]$, each eigenvalue of $A_{\tau} \eqdef A + \tau E$ lies in $\cup_{i\in[k]}\cD_i$. 
			Then, by Theorem~\ref{thm:eigenvalue_cont_pert}, it must hold that for each $\tau\in [0,1]$,  each $\cD_i$ contains exactly $m_i$ eigenvalues of $A_\tau$, since the eigenvalues can "never jump" from one $\cD_i$ to another. We leave the formal details of the last sentence to the reader.
    
    	\item Suppose that all of $A, E, X, \Lambda \in \R^{n\times n}$. Without loss of generality, we assume that $\lb_1$ is the eigenvalue multiplicity 1.
    		It holds that  $\inabs{\tlb_{1,1}-\lb_1}<\delta/2$, and for each $i\in [k], i\not=1, j\in [m_i]$, that 
    		\[\inabs{\tlb_{i,j}-\lb_1}\geq \inabs{\lb_i-\lb_1}- \inabs{\lb_i-\tlb_{i,j}} > \delta - \delta/2 = \delta/2. \]
    		That is, $\tlb_{1,1}$ is the unique eigenvalue of $\tA$ that is within $\delta/2$ distance of $\lb_1$.
    		But, the complex conjugate $\bar{\tilde{\lambda}}_{1,1}$ is also an eigenvalue of $\tA$ satisfying $\inabs{\bar{\tilde{\lambda}}_{1,1} - \lb_1} = \inabs{\tlb_{1,1}-\lb_1} < \delta/2$.
    		Hence, $\tlb_{1,1}=\bar{\tilde{\lambda}}_{1,1} \in \R$.
	\end{enumerate}
\end{proof}

 \begin{lemma}\label{lemma:eigenvec_pert} {\bf Perturbation of Eigenvectors of Simple Eigenvalues for Real Matrices.}
     Let $A =  X \cdot \Lambda \cdot X^{-1} \in \R^{n\times n}$ be a real diagonalizable matrix, with $X\in \R^{n\times n}$, and where $\Lambda = \diag(\underbrace{\lambda_1, \dots, \lambda_1}_{m_1\text{ times}},\dots,\underbrace{\lambda_k, \dots, \lambda_k}_{m_k\text{ times}})\in \R^{n\times n}$ with $\sum_{i=1}^k m_i=n$.
     Let $\delta = \min_{i,j\in [k], i\ne j} \inabs{\lb_i-\lb_j}$, and let $\kappa(X) = \lnorm{X}\cdot \lnorm{X^{-1}}$ be the condition number of $X$.
    
     Let $E\in \R^{n\times n}$ be such that $\kappa(X)\cdot \lnorm{E} < \delta/2$, and let $\tA=A+E$. Then,
     \begin{enumerate}
         \item The eigenvalues of $\tA$ can be grouped into $k$ groups $\tlb_{1,1},\dots,\tlb_{1,m_1},\dots,\tlb_{k,1},\dots,\tlb_{k,m_1}\in \C$ such that for each $i\in [k], j\in [m_i]$, it holds that $\inabs{\tlb_{i,j}-\lb_i} \leq \kappa(X)\cdot \lnorm{E}$.
       
         \item Let $\vecx_1,\dots,\vecx_n \in \R^n$ be the columns of $X$, and suppose that $\lambda_1$ is a simple eigenvalue of $A$ (that is, it has multiplicity $m_1=1$), with $\vecx_1$ as the corresponding eigenvector.
        
         Let $\tlb_1 \eqdef \tlb_{1,1}$ be the unique eigenvalue of $\tA$ such that $\inabs{\tlb_1-\lb_1}<\delta/2$.
         Then, $\tlb_1$ is real.
         Further, suppose that $\tilde{\vecx}\in \R^n, \lnorm{\tilde{\vecx}}=1$ is an eigenvector of $\tA$ with eigenvalue $\tlb_1$.
         Then, there exists $\zeta\in \inbrace{-1,1}$ such that 
         \[ \lnorm{\tilde{\vecx}-\zeta \cdot \frac{\vecx_1}{\lnorm{\vecx_1}}}\leq \frac{4\kappa(X)\lnorm{E}}{\delta}.\]  
     \end{enumerate}
 \end{lemma}
 \begin{proof}
	The first part follows from first part of Lemma~\ref{lemma:eigenval_pert}.
	Similarly, in the second part, the fact that $\tlb_{1}$ is real follows form the second part in Lemma~\ref{lemma:eigenval_pert}.
		
		Now, suppose that $\tilde{\vecx}\in \R^n, \lnorm{\tilde{\vecx}}=1$ is an eigenvector of $\tA$ with eigenvalue $\tlb_1$.
 	Let $\vecalpha \in \R^n$ be such that $\tilde{\vecx} = X\vecalpha = \sum_{i=1}^n \alpha_i \vecx_i$.
 	Then, we know
     \[ \tlb_1 X\vecalpha = \tlb_1 \tilde{\vecx} = \tA \tilde{\vecx} = (A+E)X\vecalpha = X\Lambda \vecalpha + E\tilde{\vecx}.\]
     Rearranging, we get 
     \[ (\Lambda-\tlb_1 I)\vecalpha = -X^{-1}E\tilde{\vecx}. \]
     We know that for each $i\in [k], i\not=1$, \[ \inabs{\lb_i-\tlb_1} \geq \inabs{\lb_i-\lb_1} - \inabs{\lb_1-\tlb_1} > \delta/2 > 0.\]
     Setting $P = \diag\inparen{0, \frac{1}{\Lambda_{2,2}-\tlb_1}, \dots, \frac{1}{\Lambda_{n,n}-\tlb_1}} \in \R^{n\times n}$, we know $\lnorm{P} \leq \frac{2}{\delta}$, and 
     \[\tilde{\vecx} - \alpha_1 \vecx_1 = X\inparen{\vecalpha - \insquare{\alpha_1,0,\dots,0}^\top} =  -XP X^{-1}E\tilde{\vecx}, \]
     and \[\lnorm{\tilde{\vecx} - \alpha_1 \vecx_1} \leq \lnorm{XP X^{-1}E\tilde{\vecx}} \leq \frac{2\kappa(X)\lnorm{E}}{\delta}\eqdef \gamma.\]
	We can assume that $\gamma<1$, or else the statement we wish to prove is true trivially by triangle inequality.
    This, in particular implies that $\alpha_1\not=0$, and  
    \begin{align*}
        \lnorm{\tilde{\vecx}-\frac{\alpha_1\vecx_1}{\lnorm{\alpha_1\vecx_1}}} &\leq  \lnorm{\tilde{\vecx}-\alpha_1\vecx_1}+ \lnorm{\alpha_1\vecx_1-\frac{\alpha_1\vecx_1}{\lnorm{\alpha_1\vecx_1}}} \\ &= \lnorm{\tilde{\vecx}-\alpha_1\vecx_1}+ \inabs{1-\lnorm{\alpha_1\vecx_1}}\\& \leq 2\lnorm{\tilde{\vecx}-\alpha_1\vecx_1} \\ &\leq 2\gamma.
    \end{align*}
    Choosing $\zeta = \alpha_1/\inabs{\alpha_1}$, this proves the desired result.
 \end{proof}

\subsubsection{Perturbation bounds for Pseudo-Inverse}\label{sec:pseudoinv_pert}

\begin{lemma}\label{lemma:pseudo_inv_pert}(Wedin~\cite{Wedin73}; see Theorem III.3.9 in~\cite{SS90})
    Let $A, E$ be $m\times n$ matrices with $m\geq n$, and let $\tA=A+E$.
    If $\rank(A)=\rank(\tA)=n$, then
    \[\lnorm{A^\dag - \tA^{\dag}} \leq \sqrt{2}\lnorm{A^\dag}\lnorm{\tA^{\dag}}\lnorm{E}.\]
\end{lemma}

\begin{corollary}\label{corr:pseudo_inv_perturbation}
Let $A, E$ be $m\times n$ matrices with $m\geq n$, and let $\tA=A+E$.
If $\rank(A)=n$ and $\lnorm{E}\leq \sigma_n(A)/2$, then we have 
\[\lnorm{A^\dag - \tA^{\dag}}\leq 3\lnorm{A^{\dag}}^2\lnorm{E}.\]
\end{corollary}
\begin{proof}
    By Weyl's inequality (Lemma~\ref{lemma:weyl_ineq}), we know $\sigma_n(\tA) \geq \sigma_n(A) - \lnorm{E} \geq \sigma_n(A)/2 > 0$.
    Hence, $\rank(\tA)=n$, and $\lnorm{\tA^\dag} = 1/\sigma_n(\tA) \leq 2/\sigma_n(A)$.
    Plugging this into Lemma~\ref{lemma:pseudo_inv_pert}, we get the desired result.
\end{proof}


\subsection{Anti-Concentration of Gaussian Linear Forms}

We show an anti-concentration bound for linear forms in independent Gaussian random variables.

\begin{lemma}\label{lemma:gaussian_anti_conc}
    Let $\vecx = (x_1,\dots,x_n)\in \R^n$ be such that for each $i\in[n]$, $x_i\sim\cN(0,1)$ is chosen independently, and let $\vecy = \vecx/\lnorm{\vecx}$.
    Then, for any $\veca\in\R^n$, $\veca\not=0$ and any $\delta\geq 0$, it holds that 
        \[\Pr\insquare{\inabs{\veca^\top \vecy} \leq \frac{\delta\lnorm{\veca}}{6\sqrt{n+ \ln\frac{1}{\delta}}}} \leq \delta.\]
\end{lemma}
\begin{proof}
    Let $\veca \in \R^n, \veca\not=0$.
    We can assume $\lnorm{\veca} = 1$.
    \begin{enumerate}
        \item Note that $\veca^\top\vecx \sim \cN(0,1)$, and so for any $\eta \geq 0$, \[\Pr\insquare{\inabs{\veca^\top \vecx}} = \Pr_{g \sim \cN(0,1)}\insquare{\inabs{g} \leq \eta} = \frac{1}{\sqrt{2\pi}}\int_{-\eta}^{\eta} e^{-t^2}dt \leq \frac{1}{\sqrt{2\pi}} \cdot 2\eta \leq \eta. \]
        \item By concentration bounds on chi-square random variables~\cite{LM00}, we have for any $c\geq 1$ that \[\Pr\insquare{\lnorm{\vecx}> 3c\sqrt{n}}\leq e^{-c^2n}.\]
    \end{enumerate}
    Combining the above inequalities, we get 
    \begin{align*}
    \Pr\insquare{\inabs{\veca^\top \vecy} \leq \frac{\delta}{6\sqrt{n+ \ln\frac{1}{\delta}}}} &\leq \Pr\insquare{\inabs{\veca^\top \vecx} \leq \frac{\delta}{2}} + \Pr\insquare{\lnorm{\vecx}> 3\sqrt{n}\cdot \sqrt{1+\frac{1}{n}\ln\frac{1}{\delta}}} \\ & \leq \frac{\delta}{2} + \delta e^{-n} \leq \delta.\qedhere
    \end{align*}
\end{proof}

\section{Subspace Clustering}\label{sec:scrProofs}

In this section, we will consider the problem of Subspace Clustering.
Following the outline in Section~\ref{sec:scOverview}, we will show a reduction to vector space decomposition.
In Section~\ref{subsec:sc_tensor_rec}, we solve a problem which we call Robust Recover from Symmetric Tensor Power, which we later use as a subroutine in our subspace clustering algorithm. Then, in Section~\ref{subsec:sc_noiseless} and Section~\ref{subsec:sc_robust} we analyze the noiseless and the robust case of subspace clustering.

\subsection{Robust Recovery from Symmetric Tensor Power}\label{subsec:sc_tensor_rec}

For any set $A\subseteq \R^n$, and any $d\in \N$, we define the set $\tensored{A}{d} \eqdef \setdef{(\veca \cdot \vecx)^d}{\veca \in A} \subseteq \R[\vecx]^{=d}$, where $\vecx=(x_1,\dots,x_n)$ are formal variables.
We consider the following problem:

~\\{\bf Robust Recovery from Symmetric Tensor Power (RRSTP).}\ 
Let $A = \inbrace{\veca_1,\dots,\veca_N}\subseteq \R^n$ be a set of $N$ points, and let $d\in \N$.
We are given as input a subspace $\tU \subseteq \R[\vecx]^{=d}$ such that $\dist\inparen{\tU, \inangle{\tensored{A}{d}}}$ is "small," and our goal is to \emph{efficiently} find a subspace $\tV\subseteq \R[\vecx]^{=1}$ such that $\dist(\tV, \inangle{A})$ is "small."\\

~Note that we are working with the Bombieri inner product over $\R[\vecx]^{=d}$ and $\R[\vecx]^{=(d-1)}$.
Informally speaking, the goal is to (approximately) recover the subspace $\inangle{A}$, when given as input the subspace $\inangle{\tensored{A}{d}}$ (approximately).


For the rest of this section, we fix a set $A = \inbrace{\veca_1,\dots,\veca_N}$.
Let $V \eqdef \inangle{A}$ and $r = \dim(V)$, and let $U \eqdef \inangle{\tensored{A}{d}}$ and $R = \dim(U)$.

\subsubsection{RRSTP: Noiseless Case}\label{subsec:rrstp_noiseless}
We first consider the noiseless case, in which the input is the space $U$.
Let $(\vecu_1,\dots,\vecu_R)$ be an arbitrary orthonormal basis of the space $U$.

\begin{definition}\label{defn:rrstp_par_der}
    We define the partial derivative map $T: \R^n\to (\R[\vecx]^{=(d-1)})^R$ by 
        \[ T(\vecc) = \inparen{\sum_{i=1}^n c_i\cdot \partial_{x_i}\vecu_1, \dots, \sum_{i=1}^n c_i\cdot \partial_{x_i}\vecu_R} = \inparen{\partial_{\vecc}\vecu_1,\dots,\partial_{\vecc}\vecu_R}.\]
\end{definition}

\begin{lemma}\label{lemma_rrstp_null_space_ans}
    \[\orth{\ker(T)} = \inangle{A}.\]
\end{lemma}
\begin{proof}
    We show the equivalent fact that $\ker(T) = \orth{\inangle{A}}$, which is implied by the following: 
    \begin{align*}
    	\vecc\in\ker(T) &\iff \partial_{\vecc}\vecu_i = 0 \text{ for all } i\in [R]
    	\\ &\iff \partial_{\vecc}\vecu = 0 \text{ for all } \vecu \in U
    	\\&\iff \partial_{\vecc} (\veca_i\cdot \vecx)^d = 0 \text{ for all } i\in [N]
    	\\&\iff \vecc \cdot \veca_i = 0 \text{ for all } i\in [N]
    	\\&\iff \vecc \in \orth{\inangle{A}}. \qedhere
    \end{align*}
\end{proof}

The above lemma gives a natural \emph{algorithm for the noiseless case}:
Compute the map $T$ (using any orthonormal basis of $\inangle{\tensored{A}{d}}$), and output $\orth{\ker(T)}$.
Further this algorithm is efficient: if the input is given as an orthonormal basis of $\inangle{\tensored{A}{d}}$, of size $R\cdot \binom{n+d-1}{d} \leq R\cdot n^d$, the algorithm runs in time $\poly(n^d)$.

\subsubsection{RRSTP: Robust Case}

Next, we provide a robust version of the above algorithm.
We are given as input a vector space $\tU\subseteq \R[\vecx]^{=d}$ of dimension $R$, such that $\dist(\tU, U)$ is small.
Let $(\tilde{\vecu}_1, \dots, \tilde{\vecu}_R)$ be an arbitrary orthonormal basis of $\tU$.
We define a noisy version of the partial derivative operator in Definition~\ref{defn:rrstp_par_der}.
\begin{definition}\label{defn:rrstp_par_der_noisy}
    We define the partial derivative map $\tT: \R^n\to (\R[\vecx]^{=(d-1)})^R$ by 
        \[ \tT(\vecc) = \inparen{\sum_{i=1}^n c_i\cdot \partial_{x_i}\tilde{\vecu}_1, \dots, \sum_{i=1}^n c_i\cdot \partial_{x_i}\tilde{\vecu}_R} = \inparen{\partial_{\vecc}\tilde{\vecu}_1,\dots,\partial_{\vecc}\tilde{\vecu}_R}.\]
\end{definition}

Our algorithm is then formally described as Algorithm~\ref{alg:rrstp}.

\begin{algorithm}[H]
    \caption{Robust Recovery from Symmetric Tensor Power.} \label{alg:rrstp}
    \begin{algorithmic}
        \STATE \textbf{Input}: $\tU$ is a subspace of dimension $R$, where $\tU\subseteq \R[\vecx]^{=d}$ with $d\in \N$ and $\vecx=(x_1,\dots,x_n)$ .
        \STATE \textbf{Assumptions}: There is a set $A
        \subseteq \R^n$ of size $N$, such that $\dist(\tU, \inangle{\tensored{A}{d}}) \leq \epsilon$.
        
        \STATE \textbf{Output}: Subspace $\tV\subseteq\R^n$ such that $\dist(\tV, \inangle{A})$ is small.
        
    \end{algorithmic}
    \begin{algorithmic}[1]
        \STATEx 

        \STATE Let $(\tilde{\vecu}_1, \dots,\tilde{\vecu}_R)$ be an orthonormal basis for $\tU$ (with respect to the Bombieri inner product).
        
        \STATE Let $\tilde{T}: \R^n\to (\R[\vecx]^{=(d-1)})^R$ be the (directional-derivative) map defined as in Definition~\ref{defn:rrstp_par_der_noisy}.
        
        \STATE Let $r$ be such that $R=\binom{r+d-1}{d}$.
        
        \STATE Let $\tilde{V}\subseteq \R^n$ be the space spanned by the right singular vectors of $\tT$, corresponding to the top $r$ singular values.

        \STATE Output $\tV$.
    \end{algorithmic}
\end{algorithm}

We shall prove that the algorithm gets the following guarantees. We shall work with the extra assumption that $R = \binom{r+d-1}{d}$, which for example is satisfied when $N$ is large and the set $A$ is chosen in some random manner.

\begin{proposition}\label{prop:rrstp}
    Let $d\in \N$, and let $A = \inbrace{\veca_1,\dots,\veca_N}\subseteq \R^n$ be a set of $N$ points.
    Let $V=\inangle{A}$ be of dimension $r$, and let $U = \inangle{\tensored{A}{d}} \subseteq \R[\vecx]^{=d}$ be of dimension $R$.
    Suppose that:
    \begin{enumerate}
    	\item $R = \binom{r+d-1}{d}$.
    	\item $\tU\subseteq \R[\vecx]^{=d}$ is a subspace such that $\dist(\tU, U)\leq \epsilon<1$.
    \end{enumerate}
    Then, Algorithm~\ref{alg:rrstp}, on input $\tU$, runs in time $\poly(n^d)$, and outputs a subspace $\tV\in \R^n$ such that
    \[\dist(\tV, V) \leq 4\epsilon \sqrt{r}.\]
\end{proposition}

Before we prove our main proposition, we show that the choice of basis does not affect the above algorithm in any way.

\begin{lemma}\label{lemma:rrstp_ind_basis_choice}
	The singular value decomposition of the operator $T$ (resp. $\tT$) does not depend on the choice of the orthonormal basis $(\vecu_1,\dots,\vecu_R)$ (resp. $(\tilde{\vecu}_1, \dots,\tilde{\vecu}_R)$) for the vector space $U$ (resp. $\tU$).
\end{lemma}
\begin{proof}
	For any $\vecc, \vecd\in \R^n$, we have that
	\[ \inangle{T \vecc,\ T \vecd} = \sum_{k=1}^R \inangle{\partial_{\vecc}\vecu_k, \partial_\vecd\vecu_k}_B.\]
	This is the Hilbert-Schmidt inner product between $\partial_{\vecc}$ and $\partial_{\vecd}$ over the vector space $\Lin(U, \R[\vecx]^{=(d-1)})$, which does not depend on the choice of orthonormal basis of $U$.
	Hence, the singular value decomposition of $T$ is independent of this choice of basis as well.
	
	The same proof shows the result for $\tU$ as well.
\end{proof}

%
We will also need the following singular value lower bound.
\begin{lemma}\label{lemma:rrstp_sing_bound}
	Let $T$ be as defined in Definition~\ref{defn:rrstp_par_der} with respect to an arbitrary orthonormal basis $(\vecu_1,\dots,\vecu_R)$ for $U$. 
	If $R = \binom{r+d-1}{d}$, then it holds that \[\sigma_r(T) = d\cdot \sqrt{\frac{R}{r}}.\]
\end{lemma}
\begin{proof}	
	Let $\vecc\in \R^n, \lnorm{\vecc}=1$ be such that $\vecc\in \orth{\ker(T)} = \inangle{A} = V$ (see Lemma~\ref{lemma_rrstp_null_space_ans}).
	We wish to lower bound
	\[ \inangle{T(\vecc), T(\vecc)} = \sum_{i=1}^R \norm{\partial_{\vecc}\vecu_i}_B^2.\]
	Note that by Lemma~\ref{lemma:rrstp_ind_basis_choice}, the singular value does not depend on the choice of $(\vecu_1,\dots,\vecu_R)$.
	Hence, we will choose a convenient basis to work with.
	
	First, consider an orthonormal basis $(\vecv_1=\vecc, \dots, \vecv_r)$ of $V$.
	Then, based on the variables $\vecx=(x_1,\dots,x_n)$, we define the variables $\vecy=(y_1,\dots,y_r)$ by $y_i=\vecv_i\cdot \vecx$ for each $i\in [r]$.
	This allows us to view $\inangle{\tensored{A}{d}} = U \subseteq \R[\vecy]^{=d}$ in a natural way: for each polynomial $p(\vecx)\in U$, there is a corresponding polynomial $q(\vecy)\in \R[\vecy]^{=d}$ such that $p(\vecx) = q(\vecv_1\cdot \vecx, \dots, \vecv_r\cdot \vecx)$.
	Furthermore, since the Bombieri norm is preserved under isometries, it holds that $\norm{p(\vecx)}_B = \norm{q(\vecy)}_B$ (where the norms are with respect to the spaces $\R[\vecx]^{=d}$ and $\R[\vecy]^{=d}$ respectively).
	
	Now, since $\dim(U) = R=\binom{r+d-1}{d}$, we know $U = \R[\vecy]^{=d}$, and so we can choose the orthonormal basis $(q_\alpha(\vecy) =\sqrt{\frac{d!}{\vecalpha!}} \vecy^{\vecalpha} )_{\vecalpha\in \N_d^r}$.
	Then,
	\begin{align*}
		\inangle{T(\vecv_1), T(\vecv_1)} &= \sum_{\vecalpha\in \N_d^r}\norm{\partial_{\vecv_1} q_{\vecalpha}(\vecv_1\cdot \vecx, \dots, \vecv_r\cdot \vecx)}_B^2.
		\\&= \sum_{\vecalpha\in \N_d^r} \norm{\partial_{y_1}q_\alpha(\vecy)}_B^2
		\\&= \sum_{\vecalpha\in \N_d^r: \alpha_1>0} \frac{d!}{\vecalpha!}\cdot \alpha_1^2\cdot \frac{\vecalpha!}{\alpha_1\cdot (d-1)!}
		\\&= d\cdot \sum_{\vecalpha\in \N_d^r} \alpha_1 = d\cdot \frac{Rd}{r}.
		\qedhere
	\end{align*}

%
%
%
\end{proof}

\begin{proof}[Proof of Proposition~\ref{prop:rrstp}]
Observe that by Lemma~\ref{lemma:rrstp_ind_basis_choice} , we can work with any orthonormal basis for the vector space $\tU$, and the corresponding operator $\tT$.
By the Canonical Decomposition (Theorem~\ref{thm:cs_decomp}), we can choose a basis $(\vecu_1,\dots,\vecu_R)$ for $U$, and $(\tilde{\vecu}_1,\dots,\tilde{\vecu}_R)$ for $\tU$, such that $\norm{\tilde{\vecu}_i-\vecu_i}_B \leq 2\epsilon$ for each $i\in [R]$.
Let $T$ and $\tilde{T}$ be the operators as defined in Definition~\ref{defn:rrstp_par_der} and Definition~\ref{defn:rrstp_par_der_noisy} with respect to these basis, and let $M$ and $\tM$ be the matrices corresponding to these operators.
Then, by Lemma~\ref{lemma:prelims_bomb_der}, we have
\begin{align*}
	\fnorm{\tM\vecc-M\vecc}^2 = \sum_{i=1}^R \sum_{j=1}^n\norm{\partial_{x_j}(\tilde{\vecu}_i-\vecu_i)}_B^2 = d^2 \cdot \sum_{i=1}^R \norm{\tilde{\vecu}_i-\vecu_i}_B^2 \leq 4\epsilon^2d^2R.
\end{align*}
Now by Corollary~\ref{cor:sing_nullspacePerturb} and Lemma~\ref{lemma:rrstp_sing_bound}, we get 
\[ \dist(\tV, V) \leq \frac{2\cdot \fnorm{\tM\vecc-M\vecc}}{\sigma_r(M)} \leq \frac{2\cdot 2\epsilon d\sqrt{R}}{d \sqrt{\frac{R}{r}}} = 4\epsilon\sqrt{r}.\]

\emph{Runtime:} We observe that $R\leq \dim(\R[\vecx]^{=d}) = \binom{n+d-1}{d} = \poly(n^d)$.
Hence, the map $\tT$ and its singular value decomposition can also be computed in time $\poly(n^d)$.
\end{proof}

\subsection{Subspace Clustering: Noiseless Case}\label{subsec:sc_noiseless}

We begin by considering the noiseless version of the subspace clustering problem.
Recall that we are given a set of $N$ points 
$A =\{ \veca_1,\veca_2,\ldots,\veca_N \} \subseteq \R^n$, which admit a partition 
    $$ A = A_1 \uplus A_2 \uplus \ldots \uplus A_s, $$
such that the points in each $A_j$ span a \emph{low-dimensional} space $\inangle{A_j}$.
Our goal is to find this partition.

For each $i\in [N]$, we define $\ell_i \in \R[\vecx]^{=1}$ as the linear form $\ell_i(\vecx) = (\veca_i\cdot\vecx)$ in the formal variables $\vecx=(x_1,\dots,x_n)$.
For any $d \in \N$, we define the set $\tensored{A}{d} \eqdef \setdef{(\veca \cdot \vecx)^d}{\veca \in A} \subseteq \R[\vecx]^{=d}$.
Proceeding as in Section~\ref{sec:scOverview}, we devise an algorithm (Algorithm~\ref{alg:sc_noiseless}) for this problem, via a reduction to vector space decomposition.

\begin{algorithm}[H]
    \caption{Subspace Clustering: Noiseless Case.} \label{alg:sc_noiseless}
    \begin{algorithmic}
        \STATE \textbf{Input}: $(A, d)$ where $A\subseteq \R^n$ is a set of size $N$, and $d\geq 2$ is a positive integer.
        \STATE \textbf{Assumptions}: The set $A$ admits a partition $A = A_1 \uplus A_2 \uplus \ldots \uplus A_s$ such that each $A_i$ spans a \emph{low-dimensional} subspace.
        
        \STATE \textbf{Output}: The partition $(A_1,\dots,A_s)$.
        
    \end{algorithmic}
    \begin{algorithmic}[1]
        \STATEx
        
        \STATE Compute the spaces  $U = \inangle{\tensored{A}{d}}, V = \inangle{\tensored{A}{(d-1)}}$ and the tuple of operators $\opB = (B_1,\dots,B_n)\in \Lin(U,V)^n$, where $B_i$ corresponds to the operator $\partial_{x_i}:U\to V$.
        
        \STATE\label{algstep:sc_noiseless_rvsd} Run RVSD Algorithm (Noiseless Case) on $(U,V,\opB)$, and obtain the spaces $\inangle{\tensored{A_j}{d}}$, for each $j\in [s]$.

        \STATE\label{algstep:sc_noiseless_tensor_recover} For each $j\in [s]$, run RRSTP Algorithm (Noiseless Case; see Section~\ref{subsec:rrstp_noiseless}) on $\inangle{\tensored{A_j}{d}}$ to obtain $\inangle{A_j}$.

        \STATE For each $j\in [s]$, compute $A_j = A \cap \inangle{A_j}$.

        \STATE Output $(A_1,\dots,A_s)$.
    \end{algorithmic}
\end{algorithm}

Next, we state our assumptions, and then analyze our algorithm.

\begin{definition}\label{defn:sc_non_degen_cond} (Subspace Clustering: Non-degeneracy conditions)
We say that the partition $A = A_1 \uplus A_2 \uplus \ldots \uplus A_s,$ is non-degenerate with respect to the integer $d\geq 2$ if the following conditions are satisfied:
\begin{enumerate}
        \item \[\inangle{\tensored{A}{(d-1)}} = \inangle{\tensored{A_1}{(d-1)}} \oplus \inangle{\tensored{A_2}{(d-1)}} \oplus \ldots \oplus \inangle{\tensored{A_s}{(d-1)}}.\]
        
        \item\label{cond:irred_sc} For each $j\in [s]$, the space $\inangle{\tensored{A_j}{d}}$ is irreducible with respect to the action of first order partials $\opPartials{1}$.
        That is, we cannot write $\inangle{\tensored{A_j}{d}} = U_{j,1}\oplus U_{j,2}$ and $\inangle{\tensored{A_j}{(d-1)}} = V_{j,1}\oplus V_{j,2}$, with all $U_{j,1},U_{j,2},V_{j,1},V_{j,2}$ non-zero, such that $\opPartials{1}$ maps $U_{j,1}$ into $V_{j,1}$ and $U_{j,2}$ into $V_{j,2}$.
\end{enumerate}
    
\end{definition}

\begin{theorem}\label{thm:sc_noiseless}
    Let $A \subseteq \R^n$ be set of size $N$, and $d \geq 2$ be an integer. 
    Let \[ A = A_1 \uplus A_2 \uplus \ldots \uplus A_s \]
    be a partition of $A$ that it is non-degenerate with respect to $d$ (see Definition~\ref{defn:sc_non_degen_cond}).
    Then, on input $(A,d)$, Algorithm~\ref{alg:sc_noiseless} runs in time $\poly(N,n^d)$, and outputs the partition $(A_1,\dots,A_s)$.
\end{theorem}

In the remainder of this section, we shall analyze our algorithm and prove the above theorem.


\subsubsection{Structure of The Adjoint Algebra}

In this section, we analyze Step~\ref{algstep:sc_noiseless_rvsd} of Algorithm~\ref{alg:sc_noiseless}.

Fix some $d\in \N$, such that the non-degeneracy conditions in Definition~\ref{defn:sc_non_degen_cond} are satisfied.
Let \[ U \eqdef \inangle{\tensored{A}{d}} = \inangle{\ell_1^d,\ \ell_2^d,\ \ldots,\ \ell_N^d} \subseteq \R[\vecx]^{=d}\] and \[ V \eqdef \inangle{\tensored{A}{(d-1)}} = \inangle{\ell_1^{d-1},\ \ell_2^{d-1},\ \ldots,\ \ell_N^{d-1}} \subseteq \R[\vecx]^{=(d-1)}.\]
Let $\opB = (B_1,\dots,B_n) \in \Lin(U,V)^n$ be the $n$-tuple of operators corresponding to the action of first-order partial derivatives; that is, $B_i$ corresponds to the operator $\partial_{x_i}$.
The adjoint algebra is then
\[\adj = \inbrace{(D,E): \partial_{x_i}\cdot D = E\cdot \partial_{x_i} \text{ for all }i\in[n] }\subseteq \Lin(U,U)\times \Lin(V,V).\]

To show the correctness of Step~\ref{algstep:sc_noiseless_rvsd}, it is sufficient to show that under the non-degeneracy condition, the adjoint algebra has dimension $s$.

\begin{proposition}\label{prop:sc_adj_dim_s}
    \[ \dim(\adj) = s.\]
\end{proposition}

Note that the above proposition along with Proposition~\ref{prop:adj_diag_unique_decomp} implies the uniqueness of the decomposition of $U$ as a direct sum of $s$ subspaces.

\begin{corollary}\label{corr:uniqueness_sc}
    Let $A \subseteq \R^n$ be set of size $N$, and $d \geq 2$ be an integer. 
    Let \[ A = A_1 \uplus A_2 \uplus \ldots \uplus A_s \]
    be a partition of $A$ that it is non-degenerate (see Definition~\ref{defn:sc_non_degen_cond}) with respect to $d$.

    Then, the decomposition \[\inangle{\tensored{A}{d}} = \inangle{\tensored{A_1}{d}} \oplus \inangle{\tensored{A_2}{d}} \oplus \ldots \oplus \inangle{\tensored{A_s}{d}}\] of $ \inangle{\tensored{A}{d}} $ under the action of $\opPartials{1}$ is unique.
\end{corollary}

We start by proving a few lemmas characterizing the structure of the adjoint algebra.

\begin{lemma}\label{lemma:adjDiagonal}
    Let $(D, E)\in \adj$ be any element in the adjoint algebra.
    Then, there exist field constants $c_1, c_2, \ldots, c_N\in \R$, such that for all $i \in [N]$ we have
        $$ D \cdot \ell_{i}^d = c_i \cdot \ell_{i}^d, \quad \text{and~} E \cdot \ell_{i}^{d-1} = c_i \cdot \ell_{i}^{d-1}. $$
\end{lemma}
\begin{proof}
    Without loss of generality, suppose that $i=1$.
    Further, after making a suitable change of variables, we can assume that $\ell_1 = x_1$ (note that the space of partial derivatives is closed under change of variables).

    Now, for any $j\in [N]\setminus\inbrace{1}$, we must have by the definition of the adjoint algebra that
    \[ \partial_{x_j} \cdot D \cdot x_1^d = E \cdot \partial_{x_j} \cdot x_1^d = 0.\]
    This means that $ (D \cdot x_1^d) \in \R[\vecx]^{=d} $ is a homogeneous degree $d$ polynomial depending only on the variable $x_1$.
    Hence, there exists $c_1\in \R$ such that $(D \cdot x_1^d) = c_1 \cdot x_1^d$.
    Consequently, we also have 
    \[E\cdot x_1^{d-1} = \frac{1}{d}\cdot  E\cdot\partial_{x_1} \cdot x_1^{d} = \frac{1}{d}\cdot \partial_{x_1}\cdot D\cdot x_1^{d} =  \frac{1}{d}\cdot \partial_{x_1}\cdot c_1 x_1^{d} = c_1 \cdot x_1^{d-1}. \qedhere\]
\end{proof}

\begin{lemma}\label{lemma:c1Equalc2}
    Let $(D, E)\in \adj$ be any element in the adjoint algebra, and let the constants $c_1, c_2, \ldots, c_N\in \R$ be as in Lemma~\ref{lemma:adjDiagonal}.
    Suppose that $I \subseteq [N]$ is a minimal set such that $\setdef{\ell_i^{d-1}}{i \in I}$ are linearly dependent.
    Then, $c_{i} = c_{i^{\prime}}$ for all $i, i^{\prime} \in I$.
\end{lemma}
\begin{proof}
    Without loss of generality we can assume that $I = \{ 1, 2, \ldots, r\}$.
    By the minimality of $I$, let $\alpha_1,\dots,\alpha_r\not=0$ be such that
    \begin{equation}\label{eqn:lincomb0}
        \alpha_1 \cdot \ell_1^{d-1} + \alpha_2 \cdot \ell_2^{d-1} + \ldots + \alpha_r \cdot \ell_r^{d-1} = 0.
    \end{equation}
    Then, we have
        \begin{eqnarray*}
            E \cdot (\alpha_1 \cdot \ell_1^{d-1} + \dots + \alpha_r \cdot \ell_r^{d-1})   &=  & 0 \\
            \implies \alpha_1 \cdot c_1\cdot \ell_1^{d-1} + \dots + \alpha_r \cdot c_r \cdot  \ell_r^{d-1}   &=  & 0 \\
            \implies \alpha_2 \cdot (c_2 - c_1) \cdot \ell_2^{d-1} + \dots + \alpha_r \cdot (c_r - c_1) \cdot \ell_2^{d-1}   &=  & 0 \quad (\text{using~\ref{eqn:lincomb0}}).
        \end{eqnarray*}
    Since each $\alpha_i\not=0$, the minimality of $I$ implies that for each $i\in [r]\setminus\inbrace{1}$, we have $c_i=c_1$.
\end{proof}

\begin{lemma}\label{lemma:sc_same_ci_in_part}
    Let $(D, E)\in \adj$ be any element in the adjoint algebra, and let the constants $c_1, c_2, \ldots, c_N\in \R$ be as in Lemma~\ref{lemma:adjDiagonal}.
    Fix any $j\in [s]$, and suppose that $\inangle{\tensored{A_j}{d}}$ is irreducible (see Condition~\ref{cond:irred_sc} in Definition~\ref{defn:sc_non_degen_cond}) under the action of $\opB$.
    Let $I_j = \inbrace{i\in [N]: \ell_i\in A_j}$.
    Then, it holds that $c_i=c_{i'}$ for each $i,i'\in I$.    
\end{lemma}
\begin{proof}
    Fix any $(D, E)\in \adj$ and let $c_1, c_2, \ldots, c_N\in \R$ be as in Lemma~\ref{lemma:adjDiagonal}.
    Let $I_j = \inbrace{i\in [N]: \ell_i\in A_j}$ and $C_j = \inbrace{c_i: i\in I_j}$, and let $t = \inabs{C_j}$.

    Suppose for the sake of contradiction that $t>1$, and $C_j = \inbrace{\tilde{c}_1,\dots,\tilde{c}_t}$.
    Let $U_j = \inangle{\tensored{A_j}{d}}, V_j = \inangle{\tensored{A_j}{(d-1)}}$.
    Now, for each $k\in [t]$ define $U_{j,k} = \inangle{\inbrace{\ell_i^d : i\in I_j,\ c_i = \tilde{c}_k}},$ $ V_{j,k} = \inangle{\inbrace{\ell_i^{d-1} : i\in I_j,\ c_i = \tilde{c}_k}}$.
    Then, by definition $\opB$ maps each $U_{j,k}$ into $V_{j,k}$.
    Further, it holds that $V_j = V_{j,1}\oplus\dots\oplus V_{j,t}$ is a direct sum: if we concatenate the bases of $(V_{j,k})_{k\in [t]}$, the list must be linearly independent, or else by Lemma~\ref{lemma:c1Equalc2}, some of the $\tilde{c}'s$ must be equal, which is false by definition.
    This also implies the weaker condition that $U_j = U_{j,1}\oplus\dots \oplus U_{j,t}$ is a direct sum, and this contradicts the irreducibility assumption.
\end{proof}

Next, we prove our main proposition regarding the dimension of the adjoint algebra.

\begin{proof}[Proof of Proposition~\ref{prop:sc_adj_dim_s}]   
    Since $U,V$ admit a decomposition into $s$ subspaces, the dimension of the adjoint algebra is at least $s$ (recall the adjoint algebra always contains the scaling maps; see comment after Definition~\ref{defn:adj_alg}). Suppose for the sake of contradiction that $\dim(\adj) = t > s$.

    By Lemma~\ref{lemma:adjDiagonal}, it holds that for each $(D,E)\in \adj$, there exist $c_1,\dots,c_N\in \R$ such that for each $i\in [N]$, we have $D\cdot \ell_i^d = c_i\cdot \ell_i^d$ and $E\cdot \ell_i^{d-1} = c_i\cdot \ell_i^{d-1}$.
    Now, since $\dim(\adj) = t$, there exists an element $(D,E)\in \adj$ such that at least $t$ of $c_1,\dots,c_N$ are distinct (for example, any generic element satisfies this).
    On the other hand, the irreducibility of each component $\inangle{\tensored{A_j}{d}}$ under the action of $\opB$, along with Lemma~\ref{lemma:sc_same_ci_in_part}, implies that there can be at most $s$ distinct elements among $c_1,\dots,c_N$.
    This is a contradiction.
\end{proof}

\subsubsection{Completing the Proof}

\begin{proof}[Proof of Theorem~\ref{thm:sc_noiseless}]
    The above subsection shows that Step~\ref{algstep:sc_noiseless_rvsd} correctly obtains the subspaces $\inangle{\tensored{A_j}{d}}$ for $j\in [s]$.
    Then, Step~\ref{algstep:sc_noiseless_tensor_recover} correctly obtains the subspaces $\inangle{A_j}$ for each $j\in [s]$ (see Section~\ref{subsec:rrstp_noiseless}).
    
    Now, observe that the non-degeneracy condition $\inangle{\tensored{A}{(d-1)}} = \inangle{\tensored{A_1}{(d-1)}} \oplus \inangle{\tensored{A_2}{(d-1)}} \oplus \ldots \oplus \inangle{\tensored{A_s}{(d-1)}}$ implies that for each $i\in [N]$, there is a unique $j\in [s]$ such that $\veca_i \in A_j$.
    Hence, it holds that $A_j = \inangle{A_j}\cap A$.
    This completes the proof of correctness.
    
    \textbf{Runtime:} Our algorithm deals with the set $A$ of size $N$, and works with elements in the space $\R[\vecx]^{=d}$.
    The efficiency of the RVSD algorithm, and the efficiency of the RRSTP algorithm implies that Algorithm~\ref{alg:sc_noiseless} runs in time $\poly(N,n^d)$.
\end{proof}

\subsection{Subspace Clustering: Robust Case}\label{subsec:sc_robust}

In the noisy version of the subspace clustering problem, we are given the set $A$ approximately, and we wish to cluster the points such that each cluster is close to a low-dimensional subspace.
More formally, the problem is described as follows:

\begin{problem}\label{prob:sc}
	We are given as input an integer $s\in \N$, and a set $\tA = \inbrace{ \tveca_1, \tveca_2, \ldots, \tveca_N} \subseteq \R^n$, where each $\tilde{\veca}_i$ is \emph{close} to an (unknown) point $\veca_i\in \R^n$, such that the resulting set of points $A =\inbrace{ \veca_1,\veca_2,\ldots,\veca_N } \subseteq \R^n$ can be clustered using $s$ low-dimensional subspaces, i.e.  \[ A = A_1 \uplus A_2 \uplus \ldots \uplus A_s, \]
	where each $A_j$ satisfies $\dim(A_j) \leq t$.
	Our goal is to \emph{efficiently} find an $s$-tuple of subspaces $\tilde{\vecW} = (\tW_1,\tW_2,\ldots,\tW_s)$ such that (upto reordering) for each $j\in [s]$, it holds that $\dist(\tW_j, \inangle{A_j})$ is small.
\end{problem}

Proceeding as in Section~\ref{sec:scOverview}, we devise an algorithm (Algorithm~\ref{alg:sc}) for this problem, via a reduction to robust vector space decomposition.

\begin{algorithm}[H]
    \caption{Subspace Clustering.} \label{alg:sc}
    \begin{algorithmic}
        \STATE \textbf{Input}: $(\tA, d, s, m_d, m_{d-1})$ where $\tA\subseteq \R^n$ is a set of size $N$, and $d\geq 2, m_d, m_{d-1}$ are positive integers.
        \STATE \textbf{Assumptions}: There is a set $A\subseteq \R^n$ of size $N$, such that
        	\begin{itemize}
        		\item For each point $\veca\in A$, there a unique point $\tveca\in \tA$, such that $\lnorm{\tveca-\veca}\leq \epsilon$.
        		\item The set $A$ admits a partition $A = A_1 \uplus A_2 \uplus \ldots \uplus A_s$, where each $\inangle{A_i}$ is of dimension at most $t$.
        	\end{itemize}

        \STATE \textbf{Output}: $\tilde{\vecW} = (\tW_1,\tW_2,\ldots,\tW_s)$ such that (upto reordering) for each $j\in [s]$,  $\dist(\tW_j, \inangle{A_j})$ is small.
        
    \end{algorithmic}
    \begin{algorithmic}[1]
        \STATEx
        
        \STATE Compute the matrices $M_{\tA, d}$ and $M_{\tA, d-1}$ as in Definition~\ref{defn:sc_m_ad}. Let $\tU$ (resp. $\tV$) be the subspace spanned by the left singular vectors of $M_{\tA, d}$ (resp. $M_{\tA, d-1}$), corresponding to the top $m_d$ (resp. $m_{d-1}$) singular values.
        
        \STATE \label{algstep:sc_rvsd} Let $W_1 = \R[\vecx]^{=d}, W_2 = \R[\vecx]^{=(d-1)}$, and let $\opB = (B_1,\dots,B_n)\in \Lin(W_1,W_2)^n$, where $B_i$ corresponds to the operator $\partial_{x_i}:W_1\to W_2$.
        \STATEx Run RVSD Algorithm on $(W_1,W_2, s, \tU, \tV, \opB)$\footnotemark (see Algorithm~\ref{alg:rvsd}; use projections as defined in Section~\ref{subsec:rvsd_common_op}), and let the output be $\tvecU = (\tU_1,\dots,\tU_s)$. 

        \STATE\label{algstep:sc_tensor_recover} For each $j\in [s]$, run RRSTP (Algorithm~\ref{alg:rrstp}) on $\tU_j$, and let the output be $\tW_j$.

        \STATE Output $\tilde{\vecW} = (\tW_1,\dots,\tW_s)$.
    \end{algorithmic}
\end{algorithm}
\footnotetext{See Remark~\ref{remark:sc_inputs} for the parameter $\tau\in(0,1)$} 

The above algorithm gets the following guarantees:

\begin{theorem}\label{thm:sc}
	Let $A = \inbrace{\veca_1,\dots,\veca_N} \subset \R^n$ be a finite set of $N$ points of unit norm, which can partitioned as $A = A_1 \uplus \cdots \uplus A_s$, where each $\inangle{A_i}$ is subspace of dimension at most $t$.
     
    Let $d\geq 2$ be an integer, let $\vecU = (U_1,\dots,U_s)$ (resp. $\vecV = (V_1,\dots,V_s)$) be an $s$-tuple of subspaces with $U_j =  \inangle{\tensored{A_j}{d}}$ (resp. $V_j =  \inangle{\tensored{A_j}{d-1}}$) for each $j\in [s]$.
    Let $U = \inangle{\vecU}$ (resp. $V = \inangle{\vecV}$) have dimension $m_d$ (resp. $m_{d-1}$).
    
    Suppose that:
    \begin{itemize}
    	\item $U = U_1\oplus\dots\oplus U_s$, $V = V_1\oplus\dots\oplus V_s$, and for each $j\in [s]$, it holds that $\dim(U_j) = \binom{\dim(\inangle{A_j})+d-1}{d}$, $\dim(V_j) = \binom{\dim(\inangle{A_j})+d-2}{d-1}$.
    	\item $\sigma_A$ is the minimum of $\sigma_{m_d}(M_{A,\ d})$ and $\sigma_{m_{d-1}}(M_{A,\ d-1})$, where $M_{A,d}$ (resp. $M_{A,\ d-1}$) is the matrix whose columns are the polynomials $(\veca_i\cdot\vecx)^d$ (resp. $(\veca_i\cdot\vecx)^{d-1}$) (see Definition~\ref{defn:sc_m_ad}). 
    	\item $\kappa(\vecU)$ denotes the condition number of the tuple of subspaces $\vecU$ (see Section~\ref{sec:prelims}).
    	\item $\sigma_{-(s+1)}(\adjmap)$ is the $(s+1)\textsuperscript{th}$ smallest singular value of the adjoint algebra map (see Definition~\ref{defn:adj_alg_map}), corresponding to the action of $\opB=(B_1,\dots,B_n)$ on $\vecU, \vecV$, where $B_i$ corresponds to the operator $\partial_{x_i}$. 
    \end{itemize}
    
   	Let $\tA = \inbrace{\tveca_1, \tveca_2, \ldots, \tveca_N}  \subseteq \R^n$ be a set of unit norm vectors such that $\lnorm{\veca_i - \tveca_i} \leq \epsilon$ for each $i \in [N]$.
   	Let $\delta>0$.
   	Then, Algorithm~\ref{alg:sc}, on input $(\tA, d, s, m_d, m_{d-1})$, runs in time $\poly(N, n^d)$ and outputs $\tilde{\vecW} = (\tW_1,\dots,\tW_s)$ such that with probability at least $1-\delta$, it holds (upto reordering) that for each $j\in [s]$,
	
	\begin{align*}
		\dist(\tW_j, \inangle{A_j}) 
		&\leq O\inparen{t^{2}\sqrt{N}d^2\cdot   \frac{s^{2}}{\delta} \sqrt{s+\ln\frac{s^2}{\delta}}  \cdot \kappa(\vecU)^3\cdot \frac{1}{\sigma_A}\cdot\frac{1}{\sigma_{-(s+1)}(\adjmap)}\cdot \epsilon}
		 \\&= \poly\inparen{t, N, d, s, \frac{1}{\delta},\ \kappa(\vecU),\ \frac{1}{\sigma_A},\ \frac{1}{\sigma_{-(s+1)}(\adjmap)}}\cdot \epsilon.
	\end{align*}

\end{theorem}

\begin{remark}\label{remark:sc_inputs}
	In Theorem~\ref{thm:sc}, we assume that the algorithm gets as input numbers $m_d =\dim\inparen{\inangle{\tensored{A}{d}}}$, and $m_{d-1} = \dim\inparen{\inangle{\tensored{A}{d-1}}}$ as input.
	Further, we implicitly assume that the algorithm knows a parameter $\tau\in (0,1)$ needed for the RVSD Algorithm (Algorithm~\ref{alg:rvsd}) that lies in the correct range (see Theorem~\ref{thm:rvsd}).
	
	As mentioned in Remark~\ref{remark:param_tau}, for the final algorithm, we can iterate over the parameter $\tau$, and similarly over $m_d, m_{d-1}$, and stop when the output $\tilde{\vecW} = (\tW_1,\tW_2,\ldots,\tW_s)$ gives a valid \emph{clustering} of the points in the set $\tA$ into $s$ low-dimensional subspaces; that is:
	\begin{enumerate}
		\item For each $j\in [s]$, it holds $\dim(\tW_j)\leq t$.
		\item For each $i\in [N]$, there exists $j\in [s]$ such that $\tilde{\veca}_i$ is \emph{close} to $\tW_j$.
	\end{enumerate}
	The actual inputs to the algorithm in this case will be $(\tA, d,s, t)$.
	The blow-up in the run time due to this iteration is $O(m_d\cdot m_{d-1}\cdot \log_2\inparen{\kappa(\vecU)\cdot t}) = \poly\inparen{s, t^d, \log_2\inparen{\kappa(\vecU)}}$.
	Assuming that we have good upper bounds on the condition number $\log_2\inparen{\kappa(\vecU)}$, as is required for the error bounds of Theorem~\ref{thm:sc}, we can safely ignore this in the theorem statement.
	
	Similar to the other parameters above, we may also iterate over $s$, if we know that the input to the problem can be clustered into a relatively small number of subspaces (also see Remark~\ref{remark:rvsd_rem}).
\end{remark}

The remainder of this section is devoted to the proof of Theorem~\ref{thm:sc}.
The runtime guarantees follow from the guarantees of each of the individual steps., and we shall omit the details for that.

We shall fix the set $\tA = \{ \tveca_1, \tveca_2, \ldots, \tveca_N \} \subseteq \R^n$.
Let $A=\inbrace{\veca_1,\dots,\veca_N} = A_1\uplus\dots\uplus A_s$ be as above, and let $m_d = \dim\inparen{\inangle{\tensored{A}{d}}}$ and $m_{d-1} = \dim\inparen{\inangle{\tensored{A}{d}}}$.

\subsubsection{Closeness of Subspaces}\label{sec:sc_sub_close}

\begin{definition}\label{defn:sc_m_ad}
	Consider the space $\R[\vecx]^{= d}$ with the Bombieri inner product, and of dimension $m_d = \binom{n+d-1}{d}$ (see Section~\ref{sec:prelims}).
	We define $M_{A,\ d} \in \R^{m_d\times N}$ (resp $M_{\tA,\ d}$) to be the matrix whose $i\textsuperscript{th}$ column is the polynomial $(\veca_i\cdot \vecx)^d \in \R[\vecx]^{= d}$ (resp. $(\tveca_i\cdot \vecx)^d$) written with respect to the Bombieri basis.
	
	Similarly, we also define $m_{d-1} = \binom{n+d-2}{d}$, and $M_{A,\ d-1} \in \R^{m_{d-1}\times N}$ (resp. $M_{\tA,\ d-1}$) with columns $(\veca_i\cdot \vecx)^{d-1}$ (resp. $(\tveca_i\cdot \vecx)^{d-1}$).
\end{definition}

Then, $U \eqdef \inangle{\tensored{A}{d}} \subseteq \R[\vecx]^{= d}$ corresponds to the column-span of the matrix $M_{A,\ d}$.
Let $\tU$ be the $m_d$-dimensional subspace of $M_{\tA,\ d}$ closest to $U$: this equals the vector space spanned by the left-singular vectors of $M_{\tA,\ d}$, corresponding the top $m_d$ singular values.

\begin{lemma}\label{lemma:sc_distanceToTensored}
    \[\dist(U, \tU) \leq \frac{2 \sqrt{N}d \epsilon}{\sigma_{m_d}(M_{A,\ d})},\]
    where $\sigma_{m_d}(M_{A,\ d})$ is the $m_d\textsuperscript{th}$ largest singular value of $M_{A,\ d}$.
\end{lemma}
    


\begin{proof}
    We have 
    \begin{align*}
        \fnorm{M_{A,\ d}-M_{\tA,\ d}}^2 &= \sum_{i = 1}^{N} \norm{(\veca_i\cdot\vecx)^d - (\tveca_i\cdot\vecx)^d}_B^2 \\
        &= \sum_{i = 1}^{N} \norm{(\veca_i \cdot \vecx - \tveca_i \cdot \vecx) \inparen{\sum_{j=0}^{d-1} (\veca_i \cdot \vecx)^j (\tveca_i \cdot \vecx)^{d-1-j}}}_B^2 \\
        &\le \sum_{i=1}^{N} \norm{\veca_i \cdot \vecx - \tveca_i \cdot \vecx}_B^2 \norm{\inparen{\sum_{j=0}^{d-1} (\veca_i \cdot \vecx)^j (\tveca_i \cdot \vecx)^{d-1-j}}}_B^2 \\
        &\le N\epsilon^2d^2
    \end{align*} where we use the submultiplicativity and triangle-inequality of the Bombieri norm, and the fact that $\norm{(\vecb \cdot \vecx)^k}_B = \norm{\vecb}_2^k$ for all $\vecb\in\R^n, k\in \N$.
    Now, by Corollary~\ref{cor:sing_nullspacePerturb}, we get 
    \[
        \dist(U,\Tilde{U}) \le \frac{2 \fnorm{M_{A,\ d}-M_{\tA,\ d}}}{\sigma_{m_d}(M_{A,\ d})} \leq \frac{2 \sqrt{N}d \epsilon}{\sigma_{m_d}(M_{A,\ d})}. \qedhere
    \]
\end{proof} 

In a similar manner, let $V=\inangle{\tensored{A}{d-1}}$, and let $\tV$ be space spanned by the left singular values of $M_{\tA, d-1}$, corresponding to the top $m_{d-1}$ singular values.
Then, we get
\begin{lemma}\label{lemma:sc_distanceToTensored2}
    \[\dist(V, \tV) \leq \frac{2 \sqrt{N}d \epsilon}{\sigma_{m_{d-1}}(M_{A,\ d-1})}.\]
\end{lemma}


\subsubsection{Using RVSD and RRSTP}

Let $\vecU = (U_1,\dots,U_s)$, where $U_j = \inangle{\tensored{A_j}{d}}$ 
for each $j\in [s]$.
Also, let $(\tU_1,\dots,\tU_s)$,  be the output of RVSD in Step~\ref{algstep:sc_rvsd} of Algorithm~\ref{alg:sc}.

Let $\delta>0$. 
Then, by Corollary~\ref{corr:rvsd_common_op}, we get that with probability at least $1-\delta$, (upto reordering) for each $j\in [s]$,
 
\[ \dist(\tU_j, U_j) \leq O\inparen{t^{3/2}\cdot \kappa(\vecU)^3\cdot  s^{2} \sqrt{s+\ln\frac{s^2}{\delta}}\cdot \frac{\eps_1+\eps_2}{\delta}\cdot\frac{\lnorm{\opB}}{\sigma_{-(s+1)}(\adjmap)}},\]
where we have that
\begin{enumerate}
	\item $\sigma_{-(s+1)}(\adjmap)$ is the $(s+1)\textsuperscript{th}$ smallest singular value of the relevant adjoint algebra map.
	\item $\opB = (B_1,\dots,B_n)\in \Lin\inparen{\R[\vecx]^{= d},\ \R[\vecx]^{= (d-1)}}^n$, with $B_i$ being the map corresponding to the operation $\partial_{x_i}$.
		Note that by Definition~\ref{defn:joint_op} and Lemma~\ref{lemma:prelims_bomb_der}, we have $\lnorm{\opB} = d$.
	\item By Lemma~\ref{lemma:sc_distanceToTensored} and  Lemma~\ref{lemma:sc_distanceToTensored2}, $\dist(\tU, U)\leq \eps_1 = \frac{2 \sqrt{N}d \epsilon}{\sigma_{m_{d}}(M_{A,\ d})}$ and $\dist(\tV, V)\leq \eps_2 = \frac{2 \sqrt{N}d \epsilon}{\sigma_{m_{d-1}}(M_{A,\ d-1})}$ .
\end{enumerate}
Simplifying, we get
\begin{align*}
    \dist(& \tU_j, U_j) \leq \\ & O\inparen{t^{3/2}\sqrt{N}d^2\cdot   \tfrac{s^{2}}{\delta} \sqrt{s+\ln\tfrac{s^2}{\delta}}  \cdot \kappa(\vecU)^3\cdot \tfrac{1}{\min\inbrace{\sigma_{m_{d}}(M_{A,\ d}),\ \sigma_{m_{d-1}}(M_{A,\ d-1})} }\cdot\tfrac{1}{\sigma_{-(s+1)}(\adjmap)}\cdot \epsilon}.
\end{align*}

Now, for each $j\in [s]$, let $W_j = \inangle{A_j}$.
Then, assuming that $\dim(U_j) = \binom{\dim(W_j)+d-1}{d}$, we have by Proposition~\ref{prop:rrstp}, that the output $\tW_j$ of the RRSTP algorithm satisfies
\begin{align*}
     \dist(&\tW_j, W_j) \leq \\ &O\inparen{t^{2}\sqrt{N}d^2\cdot   \tfrac{s^{2}}{\delta} \sqrt{s+\ln\tfrac{s^2}{\delta}}  \cdot \kappa(\vecU)^3\cdot \tfrac{1}{\min\inbrace{\sigma_{m_{d}}(M_{A,\ d}),\ \sigma_{m_{d-1}}(M_{A,\ d-1})} }\cdot\tfrac{1}{\sigma_{-(s+1)}(\adjmap)}\cdot \epsilon}.
\end{align*}

\subsection{Singular Value Analysis for The Adjoint Algebra}

Analysis of the singular values of the adjoint algebra operator can be tedious. However, for subspace clustering we can obtain substantial lower bounds for $\sigma_{-(s + 1)}(\adjmap)$ which reveal how the geometry of the original subspaces affect the robustness of the algorithm. We dedicate section \ref{sec:singval_AdjointOperator} to this analysis. We proceed by defining a special inner product for linear maps on sums of subspaces. This allows an inductive approach to separately bound the contributions of the diagonal and off-diagonal blocks of $\adjmap$ to its smallest non-zero singular value. In particular we get the following theorem.

\begin{theorem}[Theorem \ref{thm:adjoint_algebra_sc_singular_values} restated]\label{thm:sing_val_of_adjoint_algebra_operator} 
For subspaces $U, V$ of dimension $t$, let $f_d$ be defined as $f_d\left(U, V\right) = \frac{d}{t} \left[\sum_{k\in[t]} \sin^2 \theta_k + (d - 1) \sin^2 \theta_{t}\right]$ for $d \ge 2$, where $\theta_1 \geq \theta_2 \geq \cdots \geq \theta_t$ are the canonical angles between $U, V$. Then, if $\adjmap$ represents the adjoint algebra map corresponding to the subspace clustering problem for subspaces $\inangle{A_1}, \inangle{A_2}, \ldots, \inangle{A_s}$ with parameter $d$, we have
            $$\sigma^2_{-(s + 1)}(\adjmap) \geq \insquare{\tfrac{d}{\kappa(\vecU, \vecV)^2}}^2 \cdot \min\{\sigma_\diag, \sigma_\offdiag \}$$
            
        \noindent where the above quantities are defined as follows:
        \begin{align*}
               \sigma_\diag &\eqdef \sqrt{\frac{d}{t^*+d-1}} \inparen{1 - \sqrt{1 - \tfrac{1}{d} \tfrac{t^*}{t^* + d - 1}}}, \\
            \sigma_\offdiag &\eqdef \min_{j \neq k} \sqrt{\frac{d}{t_k+d-1}}\inparen{1 - \sqrt{1 - \tfrac{1}{d} \tfrac{t_{k}}{t_{k} + d - 1} \cdot f_d(\inangle{A_j}, \inangle{A_k})}},
        \end{align*}
        
        \noindent where $\kappa(\vecU,\vecV) = \max\inbrace{\sigma_1(\vecU), \sigma_1(\vecV)}/\min\inbrace{\sigma_{-1}(\vecU), \sigma_{-1}(\vecV)}$, and $t_1, \ldots, t_s$ are dimensions of $\inangle{A_1}, \ldots, \inangle{A_s}$ respectively, with $t^* = \max_{i \in [s]}t_i$.

\end{theorem}

Refer to section \ref{sec:singval_AdjointOperator} for the proof of Theorem \ref{thm:sing_val_of_adjoint_algebra_operator}.

\subsection{Smoothed Analysis of Subspace Clustering}

We analyse our algorithm for subspace clustering in a smoothed setting.
We first describe the input model for our problem. For simplicity, we assume that each of the subspaces has the same dimension. 
\begin{enumerate}
    \item We have a tuple of $s$ hidden subspaces of $\R^n$, $\vecW = (W_1, W_2, \ldots, W_s)$, each of dimension $t$. Let $P_1, P_2, \ldots, P_s \in \R^{n \times t}$ be matrices with orthonormal columns, such that the column span of $P_i$ is $W_i$. Each subspace $W_i$ is perturbed by perturbing $P_i$ by a random Gaussian matrix $G_i \sim \mathcal{N}(0,\rho^2/n)^{n \times t}$. Let $\hat{P}_i = P_i + G_i$, and $\hat{W}_1, \hat{W}_2, \ldots, \hat{W}_s$ be the column spans of $\hat{P}_1, \hat{P}_2, \ldots, \hat{P}_s$ respectively. 
    
    \item Sample sets of points $A_1, A_2, \ldots, A_s$ from $\hat{W}_1, \hat{W}_2, \ldots, \hat{W}_s$ respectively, of unit norm. For each $i \in [s]$, perturb each point in $A_i$ with respect to $\hat{W}_i$ to get the set of points $\hat{A}_i$. Formally, this means perturbing points in $A_i$ by $\hat{B}_i \cdot v$, where $\hat{B}_i$ is an $n \times t$ matrix describing an orthonormal basis for $\hat{W}_i$ and $v \sim \mathcal{N}(0,\rho^2/t)^t$, and normalizing. Let $\hat{A} = \hat{A}_1 \cup \hat{A}_2 \cup \cdots \cup \hat{A}_s$.
    
    \item For each $\veca \in \hat{A}$, add noise (and normalize) to get point $\veca^\prime$ such that $\norm{\veca - \veca^\prime}_2 \le \epsilon$. We are given $\hat{A}^\prime$, the set of noise-added points.
\end{enumerate}

Given the set of points $\hat{A}^\prime$, the goal is to recover subspaces $\tvecW = (\tW_1, \tW_2, \ldots, \tW_s)$ such that $\dist(\hat{W}_j, \tW_j)$ is small for each $j\in[s]$. 

Algorithm \ref{alg:sc} gets the following guarantees in the above smoothed setting. 

\begin{theorem}\label{thm:smoothed_sc}
    Let $\vecW = \inparen{W_1, \ldots, W_s}, \hat{A}^\prime$ and $\hat{A} = \hat{A}_1 \uplus \cdots \uplus \hat{A}_s$ be as generated above, with smoothening parameter $\rho \in (0,1)$. Let $\hat{\vecU} = (\hat{U}_1, \ldots, \hat{U}_s)$ and $\hat{\vecV} = (\hat{V}_1, \ldots, \hat{V}_s)$ be $s$-tuples such that $\hat{U}_j = \inangle{\tensored{\hat{A}}{d}_j}, \hat{V}_j = \inangle{\tensored{\hat{A}}{d-1}_j}$ for each $j \in [s]$, for some $d \ge 2$. Further, for each $j\in[s]$, let $M_{\hat{A}_j,d} = C_{j,d} \cdot M_{\hat{B}_j,d}$ and  $M_{\hat{A}_j,d_1} = C_{j,d-1} \cdot M_{\hat{B}_j,d-1}$ where $C_{j,d}$ (resp. $C_{j,d-1}$) is a matrix with columns as an orthonormal basis of $\hat{U}_j$ (resp. $\hat{V}_j$), and $B_j \subset \R^t$.
    Let $\hat{U} = \inangle{\hat{\vecU}}, \hat{V} = \inangle{\hat{\vecV}}$, and $m_d = \dim(\hat{U}), m_{d-1}=\dim(\hat{V})$.
    Let $\delta > 0$ and $t > d$. Then, Algorithm \ref{alg:sc} on input $(\hat{A}^\prime, d, s, m_d, m_{d-1})$ outputs $\tvecW = (\tW_1, \ldots, \tW_s)$ such that with probability $1 - \delta - s^2 \exp(-\Omega(\rho^2n))$,

\[ \dist(\tW_j, \hat{W}_j) \leq O\inparen{\sqrt{N} \cdot   \frac{s^{2}}{\delta} \sqrt{s+\ln\frac{s^2}{\delta}}  \cdot \kappa ^ 6 \cdot \frac{1}{\sigma}\cdot t^{3/4} \cdot d^{3/4} \cdot n^{1/2} \cdot \frac{\epsilon}{\rho}}\]
where $$\kappa = \frac{\max\inbrace{\sigma_1(\vecU), \sigma_1(\vecV)}}{\min\inbrace{\sigma_{-1}(\vecU), \sigma_{-1}(\vecV)}},\enspace,  \sigma = \min_{j\in[s]}\min\inbrace{\sigma_{-1}(M_{\hat{B}_j,d}), \sigma_{-1}(M_{\hat{B}_j,d-1})}, \text{and~} N = |\hat{A}^\prime|.$$
\end{theorem}

The rest of this section is devoted to the proof of Theorem \ref{thm:smoothed_sc}.

We first show that the perturbed subspaces generated by the above procedure are well separated, with high probability.

\begin{lemma}\label{lemma:sc_distancePerturbedSubspaces}
    Let $V$ be a subspace of $\R^n$ of dimension $t < n$. Let $\hat{V}$ be the perturbed subspace generated from $V$ by perturbing an orthonormal basis of $V$ by vectors sampled from $\mathcal{N}(0,\rho^2/n)^n$. Then, if $U$ is another subspace of dimension $t$, we have, with high probability, 
    \begin{equation*}
        \dist(U,\hat{V}) \ge \frac{1}{5} \cdot \frac{\rho}{\sqrt{1+\rho^2}} \cdot \sqrt{1-\frac{t}{n}}    \end{equation*}
\end{lemma}

\begin{proof}
    Without loss of generality, we can assume that $U = \inangle{\vece_1, \vece_2, \ldots, \vece_t}$. Let $V = \inangle{\vecv_1, \vecv_2, \ldots, \vecv_t}$ be an orthonormal basis for $V$, and $\hat{V} = \inangle{\hat{\vecv}_1, \hat{\vecv}_2, \ldots, \hat{\vecv}_t}$, where $\hat{\vecv}_i = \vecv_i + \vecb_i$, and $\vecb_i \sim \mathcal{N}(0, \rho^2/n)^n$. We have
    \begin{align*}
        \dist(U, \hat{V}) &= \norm{\Proj_U - \Proj_{\hat{V}}}_2 \\
        &\ge \frac{\norm{\Proj_U \cdot \hat{\vecv}_1 - \Proj_{\hat{V}} \cdot \hat{\vecv}_1}_2}{\norm{\hat{\vecv}_1}_2} \\
        &= \frac{\norm{\Proj_U \cdot \hat{\vecv}_1 - \hat{\vecv}_1}_2}{\norm{\hat{\vecv}_1}_2}.
    \end{align*}
    Since $\vecv_1$ has unit norm, we have that 
    \begin{align*}
        \Pr\insquare{\norm{\hat{\vecv}_1}^2_2 \ge 1 + 5\rho^2} \le \exp(-\Omega(\rho^2 n)).
    \end{align*}
    Also, since $U = \inangle{\vece_1, \vece_2, \ldots, \vece_t}$, if $\hat{\vecv}_1 = \sum_{i \in [n]} \hat{v}_{1i} \vece_i$, then $\Proj_U \cdot \hat{\vecv}_1 = \sum_{i \in [t]} \hat{v}_{1i} \vece_i$. Therefore,

    \begin{align*}
        \norm{\Proj_U \cdot \hat{\vecv}_1 - \hat{\vecv}_1}^2_2 = \sum_{i = t + 1}^{n} \hat{v}^2_{1i}.
    \end{align*}
    Thus, we have 
    \begin{align*}
        \Pr\insquare{\norm{\Proj_U \cdot \hat{\vecv}_1 - \hat{\vecv}_1}^2_2 \le \frac{2\rho^2}{5} \cdot \inparen{1 - \frac{t}{n}}} \le \exp(-\Omega(\rho^2 n)).
    \end{align*}
    Therefore, we have, with probability at least $1 - \exp(-\Omega(\rho^2 n))$, 
    \begin{align*}
        \frac{\norm{\Proj_U \cdot \hat{\vecv}_1 - \hat{\vecv}_1}^2_2}{\norm{\hat{\vecv}_1}^2_2} &\ge \frac{2\rho^2}{25(1+\rho^2)}\cdot\inparen{1 - \frac{t}{n}} \\
        &\ge \frac{1}{25}\cdot\frac{\rho^2}{1+\rho^2} \cdot \inparen{1 - \frac{t}{n}}.
    \end{align*}
    
\end{proof}

As a corollary, we get a lower bound on the smallest non-zero singular value of the adjoint algebra for subspace clustering, in the smoothed setting.

\begin{corollary}\label{corr:sc_singval}
    Let the subspaces $\hat{W}_1 = \inangle{\hat{A}_1}, \ldots, \hat{W}_s = \inangle{\hat{A}_s}$ be generated as given above. Let $\vecU = \inparen{\inangle{\tensored{\hat{A}_1}{d}}, \ldots, \inangle{\tensored{\hat{A}_s}{d}}}$, $\vecV = \inparen{\inangle{\tensored{\hat{A}_1}{d-1}}, \ldots, \inangle{\tensored{\hat{A}_s}{d-1}}}$, for $d \ge 2$. Let $\adjmap$ be the adjoint algebra map corresponding to action of the derivative maps on $\vecU, \vecV$. Then, we have, with probability at least $1 - s^2 \exp(-\Omega(\rho^2 n))$, 
    \begin{align*}
        \sigma_{-(s+1)}^2(\adjmap) \ge \frac{1}{50\kappa^4} \cdot \frac{d^{5/2}}{(t+d-1)^{3/2}} \cdot  \frac{\rho^2}{1+\rho^2} \cdot \frac{1}{n}. 
    \end{align*} where $\kappa = \kappa(\hat{\vecU},\hat{\vecV})$.

\end{corollary}

\begin{proof}
    By Theorem \ref{thm:sing_val_of_adjoint_algebra_operator}, we have 
        $$\sigma^2_{-(s + 1)}(\adjmap) \geq \frac{d^2}{\kappa^4} \cdot \min\{\sigma_\diag, \sigma_\offdiag\},$$ where 
        \begin{align*}
            \sigma_\diag &=  \sqrt{\frac{d}{t+d-1}} \displaystyle\left( 1 - \sqrt{1 - 
                          \frac{1}{d} \cdot \frac{t}{d - 1 + t} }  \right)\\
            &\ge \frac{1}{2} \cdot \frac{t}{d^{1/2}\cdot(t+d-1)^{3/2}} 
        \end{align*} 
        and 
        \begin{align*}
            \sigma_{\offdiag} &= \min_{i \ne j} \sqrt{\frac{d}{t+d-1}} \cdot \left( 1 - \sqrt{1 - \frac{1}{d} \cdot \frac{t}{d - 1 
                         + t} \cdot f_d\left(\hat{W}_i, \hat{W}_j\right) }  \right) \\
                        &\ge \frac{1}{2} \cdot \frac{t}{d^{1/2}\cdot(t+d-1)^{3/2}} \cdot \min_{i \ne j} f_d(\hat{W}_i, \hat{W}_j).
        \end{align*} 
        For subspaces $U, V$ with canonical angles $\theta_1 \ge \theta_2 \ge \cdots \theta_t \ge 0$, we have  
        \begin{align*}
            f_d\left(U, V\right) = \frac{d}{t} \left[\sum_{k\in[t]} \sin^2 \theta_k + (d - 1) \sin^2 \theta_{t}\right] \ge \frac{d}{t}\cdot\sin^2(\theta_1) = \frac{d}{t} \cdot \dist(U,V)^2.
        \end{align*}
        Thus, by Lemma \ref{lemma:sc_distancePerturbedSubspaces}, we have, with probability $1 - s^2 \exp(-\Omega(\rho^2n))$,
        \begin{align*}
            f_d(\hat{W}_i,\hat{W}_j) \ge \frac{1}{25}\cdot\frac{d}{t} \cdot \frac{\rho^2}{1+\rho^2} \cdot (1 - \frac{t}{n}).
        \end{align*} for all $i, j \in [s]$ such that $i \ne j$.
        Since $t \le n - 1$, we have the required expression.
\end{proof}

Using the above, and the fact that $\vecW, \tvecW$ will be close with high probability, we have our result.  

\begin{proof}[Proof of Theorem \ref{thm:smoothed_sc}]
    By Corollary \ref{corr:sc_singval} and Theorem \ref{thm:sc}, if Algorithm \ref{alg:sc} returns subspaces $\tW = (\tW_1, \ldots, \tW_s)$, we have, with probability $1 - \delta - s^2 \exp(-\Omega(\rho^2 n))$, 
        \[ \dist(\tW_j, W_j) \leq O\inparen{\sqrt{N} \cdot   \frac{s^{2}}{\delta} \sqrt{s+\ln\frac{s^2}{\delta}}  \cdot \kappa ^ 5 \cdot \frac{1}{\sigma^\prime}\cdot t^{3/4}d^{3/4}n^{1/2} \cdot \frac{\epsilon}{\rho}}\] and $\sigma^\prime = \min\inbrace{\sigma_{m_d(M_{\hat{A},d})}, \sigma_{m_{d-1}(M_{\hat{A},d-1})}}$.

    Note that we can write $M_{\hat{A},d}$ as 
    \begin{align*}
        M_{\hat{A},d} = \begin{bmatrix}
            C_{1,d} & \dots & C_{s,d}
        \end{bmatrix} \cdot \begin{bmatrix}
            M_{\hat{B}_1,d} &  & \\
            & \ddots & \\
            & & M_{\hat{B}_s,d}
        \end{bmatrix}.
    \end{align*}

    We can similarly write an expression for $M_{\hat{A},d-1}$. Thus, we get that 

    \begin{align*}
        \sigma^\prime \ge \frac{1}{\kappa} \min_{j \in [s]} \min\inbrace{\sigma_{-1}(M_{\hat{B}_j, d}), \sigma_{-1}(M_{\hat{B}_j, d-1})}.
    \end{align*}
    This gives us the required result.
    
\end{proof}
\section{Learning Mixtures of Gaussians}\label{sec:gaussianMixturesReduction}

In this section we give the proofs of technical claims pertaining to section \ref{sec:gmOverview}.

\begin{lemma}\label{lem:equalsDirectSum} [\cite{GKS20}]
    {\bf Blessing of Dimensionality for sums of powers of quadratics.}
    Let positive $n, k, d, s $ be non-negative integers satisfying the following constraints. 
        $$ s \cdot \binom{n + k - 1}{k} \ll \binom{n + 2d - k - 1}{2d - k}. $$
    Then with high probability over the random choice of $p_1(\vecx), p_2(\vecx), \ldots, p_s(\vecx) \in \R[\vecx] $ we have
        $$ \inangle{\opL \cdot p} = \inangle{\vecy^{=k} \cdot q_1^{d-k}} \oplus \inangle{\vecy^{=k} \cdot q_2^{d-k}} \oplus \ldots \oplus \inangle{\vecy^{=k} \cdot q_s^{d-k}}, $$
    where for each $i \in [s]$, $q_i(\vecy) \in \R[\vecy]$ is a restriction of $p_i$, and $\inabs{y} = n$. 
\end{lemma}

\begin{remark}
    The constraints/bounds on the parameters for which the conclusion of the above lemma holds is likely suboptimal and the conclusion is likely to hold for a larger range.  
\end{remark}

\begin{lemma}\label{lem:direct_sum} [\cite{GKS20}]
    {\bf Blessing of Dimensionality for shifted spaces of powers of independent quadratics.}
    Let positive $n, k, e, s $ be non-negative integers satisfying the following constraints. 
        $$  s \cdot \binom{n + k - 1}{k} \ll \binom{n + 2e + k - 1}{2e + k}.$$
    Then with high probability over the random choice of $q_1(\vecy), q_2(\vecy), \ldots, q_s(\vecy) \in \R[\vecy] $, $\inabs{y} = n$ we have
        $$ \inangle{\vecy^{=k} \cdot q_1^e} + \inangle{\vecy^{=k} \cdot q_2^e} + \ldots + \inangle{\vecy^{=k} \cdot q_s^e} = \inangle{\vecy^{=k} \cdot q_1^e} \oplus \inangle{\vecy^{=k} \cdot q_2^e} \oplus \ldots \oplus \inangle{\vecy^{=k} \cdot q_s^e} $$
\end{lemma}

\begin{remark}
    The constraints/bounds on the parameters for which the conclusion of the above lemma holds is likely suboptimal and the conclusion is likely to hold for a larger range.  
\end{remark}

\begin{lemma}\label{lem:rednMain}
    The adjoint algebra $A$ for the action of $\opB$ has dimension $s$. In other words for any $(D, E) \in A$ there exists $(\lambda_1, \lambda_2, \ldots, \lambda_s) \in \R^s$ such that for all $i \in [s], f(\vecy) \in \R[\vecy]^{=k}, g(\vecy) \in \R[\vecy]^{=(k+2)}$
        $$ D \cdot (q_i(\vecy)^e \cdot f(\vecy)) = \lambda_i \cdot (q_i(\vecy)^e \cdot f(\vecy)) $$
    and 
        $$ E \cdot (q_i(\vecy)^{e-1} \cdot g(\vecy)) = \lambda_i \cdot (q_i(\vecy)^{e-1} \cdot g(\vecy)).$$
\end{lemma}

\begin{corollary}\label{cor:unique}
    The decomposition of $U$ under the action of $\opB$ is unique.
\end{corollary}

We need a couple of technical lemmas. 
\begin{lemma}\label{lem:multiple}
    Fix any $(D, E)$ in $A$. Then there exist $R_1, R_2, \ldots R_s \in \inangle{q_1^{d-k}, q_2^{d-k}, \ldots, q_s^{d-k}}$ such that for all $f(\vecy) \in \R[\vecy]^{=k}$ we have
        $D \cdot (q_i^e \cdot f(\vecy)) = R_i \cdot f(\vecy)$.
\end{lemma}

\begin{lemma}\label{lem:existsBl}
    Let $q_1(\vecy), q_2(\vecy), \ldots, q_s(\vecy)$ be independently chosen random quadratic forms.  
    For all $i \in [s]$, there exists $B_i \in \opB$ and $\ell_i \in \R[\vecy]^{=k}$ 
    such that $B_i \cdot (q_i^e \cdot \ell_i^k) = B_i \cdot \ell_i = 0$ and moreover 
    for all $j \in [s] \setminus \{i\}$, we have 
        $B_i \cdot (q_j^e \cdot \ell_i^k ) \neq 0 $.
\end{lemma}

\begin{proof}[Proof of Lemma \ref{lem:rednMain}]
    Let $(D, E)$ be any element of the adjoint algebra $A$.
    By lemma \ref{lem:multiple}, there exist $R_1, R_2, \ldots R_s \in \inangle{q_1^{d-k}, q_2^{d-k}, \ldots, q_s^{d-k}}$ such that for all $f(\vecy) \in \R[\vecy]^{=k}$ we have
        $D \cdot (q_i^e \cdot f(\vecy)) = R_i \cdot f(\vecy)$.
    Let $R_{i} = c_{i1} q_1^e + c_{i2} q_2^e + \ldots + c_{is} q_s^e$.
    For an arbitrary $i \in [s]$, let $B_i \in \opB, \ell_i \in \R[\vecy]^{=k}$ be as provided by lemma \ref{lem:existsBl}. 
    Now 
        \begin{eqnarray*}
            B_i \cdot (q_i^e \cdot \ell_i^k) & = & 0 \\
            \implies E \cdot B_i \cdot (q_i^e \cdot \ell_i^k) & = & 0 \\
            \implies B_i \cdot D \cdot (q_i^e \cdot \ell_i^k) & = & 0 \\
            \implies B_i \cdot (R_i \cdot \ell_i^k) & = & 0 \\
            \implies B_i \cdot ( (\sum_{j \in [s]} c_{ij} q_j^e) \cdot \ell_i^k) & = & 0 \\
            \implies \sum_{j \in [s]} c_{ij} \cdot B_i \cdot (q_j^e \cdot \ell_i^k) & = & 0 \\
            \implies \forall j \in [s]: c_{ij} \cdot (B_i \cdot (q_j^e \cdot \ell_i^k)) & = & 0 \quad (\text{as $V_j$'s form a direct sum}) \\
            \implies \forall j \in ([s] \setminus \{i\}): c_{ij} & = & 0 \quad (\text{using the property of $B_i, \ell_i$ from lemma \ref{lem:existsBl}})
        \end{eqnarray*}
    Thus $R_i(\vecy) = c_{ii} \cdot q_i(\vecy)^{e}$ and so that for all $f \in \R[\vecy]^{=k}, g \in \R[\vecy]^{=(k+2)}$ we have
    $D \cdot q_i^e \cdot f = c_{ii} \cdot q_i^e \cdot f$ and $E \cdot q_i^{e-1} \cdot g = c_{ii} \cdot q_i^{e-1} \cdot g$. In particular, the adjoint algebra for $\opB$ is $s$-dimensional.
\end{proof}


\noindent We will now provide the proofs of lemmas \ref{lem:multiple} and \ref{lem:existsBl}. For this we will in turn need a couple more preliminary observations. \\

\noindent {\bf Existence of nice basis.}
    Let us first observe that the space of homogeneous polynomials admits a basis consisting of {\em very simple} 
    polynomials for pretty much any definition of what is a very simple polynomial. In particular let us work with the specific notions below of what a very simple polynomial is.  
    If a polynomial $g(\vecx) \in \R[\vecx]^{=k}$ is of the form $g(\vecx) = \ell(\vecx)^k$
    for some linear form $\ell(\vecx) \in \R[\vecx]^{=1}$, let us call it as a power of linear form and in short refer to it as a homogeneous-$ \Sigma \bigwedge$ polynomial. The following is a classical result.

\begin{theorem}\label{thm:ellison}
    {\bf Ellison \cite{ellison}.}
    There exists a basis of $\R[\vecx]$ consisting of homogeneous-$ \Sigma \bigwedge$ polynomials.
\end{theorem}
    
\noindent We will need variants of the above theorem for some other notions of what constitutes 
    a very simple polynomial. If a polynomial $g(\vecx) \in \R[\vecx]^{=k}$ is of the form 
        $$ g(\vecx) = \ell_1(\vecx)\cdot \ell_2(\vecx) \cdot \ldots \cdot \ell_k(\vecx), $$
    $\ell_1(\vecx), \ell_2(\vecx), \ldots, \ell_k(\vecx) \in \R[\vecx]^{=1}$ are coprime linear forms 
    let us call it as a product of coprime linear forms and in short refer to it as a homogeneous-$\Pi^{\text{coprime}} \Sigma $ polynomial.
    For an integer $r \geq 2$ we further call such a polynomial as a homogeneous-$\Pi^{\text{coprime, $r$-span}} \Sigma $ polynomial if we also have $\dim(\inangle{\ell_1(\vecx), \ell_2(\vecx), \ldots, \ell_k(\vecx)})=r$.

\begin{claim}\label{clm:niceBasis}
    There exists a basis of $\R[\vecx]^{=k}$ consisting of homogeneous-$\Pi^{\text{coprime, $2$-span}} \Sigma $ polynomials. 
\end{claim}
\begin{proof}
    By induction on the number of variables $n = \lvert \vecx \rvert$. \\
    
\noindent {\bf Base case.}    For the base case of $n = 2$, let 
    $\alpha_1, \alpha_2, \ldots, \alpha_{k+1} \in \R$ be any set of $k+1$ distinct field elements. 
    Let $p(x_1, x_2) \eqdef (x_1 + \alpha_1 x_2) \cdot (x_1 + \alpha_2 x_2) \cdot \ldots \cdot (x_1 + \alpha_{k+1} x_2)$. Let $ p_i \in \R[\vecx]^{=k}$ be defined as 
        $p_i(x_1, x_2) \eqdef \prod_{j \neq i} (x_1 + \alpha_j x_2) = \frac{p}{(x_1 + \alpha_i x_2)}$.
    We claim that the $p_i$'s ($i \in [k+1]$) form the required basis of $\R[\vecx]^{=k}$. Each of them
    is clearly a product of coprime linear forms and so it suffices to show they are linearly independent. 
    Suppose that 
        $$ c_1 p_1 + c_2 p_2 + \ldots + c_{k+1} p_{k+1} = 0. $$
    Making the substitution $x_1 = -\alpha_1 x_2 $ in the above identity we get that 
    \begin{align*}
        c_1 \cdot (-\alpha_1 x_2 + \alpha_2 x_2) \cdot (-\alpha_1 x_2 + \alpha_3 x_2) \cdot \ldots & \cdot (-\alpha_1 x_2 + \alpha_{k+1} x_2) \\ &+ c_2 \cdot 0 + c_3 \cdot 0 + \ldots + c_{k+1} \cdot 0  = 0
    \end{align*}
    from which we infer that $c_1 = 0$. Similarly we can infer $c_i = 0$ for all $i \in [k+1]$ implying that the 
    $p_i$'s are linearly independent, as required. \\

\noindent {\bf Inductive step.} 
    Suffices to show that any monomial $m \in \R^[\vecx]^{=k}$ can be expressed as a linear combination
    of homogeneous-$\Pi^{\text{coprime, $r$-span}} \Sigma $ polynomials. Suppose that $m=x_1^{e_1} \cdot x_2^{e_2} \cdot \ldots \cdot x_n^{e_n}$ where $\sum_{i \in [n]} e_i = k$. By theorem \ref{thm:ellison} let 
    $m_1 \eqdef x_2^{e_2} \cdot x_3^{e_3} \cdot \ldots \cdot x_{n}^{e_{n}}$ admit a representation as a 
    sum of powers of linear forms as 
        $$ m_1 = \sum_{i} \ell_{i}(x_2, x_3, \ldots, x_n)^{r}, \mathrm{where~} r = e_2 + e_3 + \ldots + e_{n}. $$
    So it suffices to show that for any $i$, $(x_1^{e_1} \cdot \ell_{i}^{r})$ can be expressed as a sum of homogeneous-$\Pi^{\text{coprime, $2$-span}} \Sigma $ polynomials.
    By making a suitable change of variables we can assume that $\ell_{i} = x_2$ and then we can infer the previous statement from the base case above. 
\end{proof}

\noindent Now consider a linear form $\ell(\vecx) \in \R[\vecx]^{=1}$ that divides a product of two polynomials 
    $q(\vecx)$ and $p(\vecx)$. If $q$ is coprime to $\ell$ we can infer that $\ell$ divides $p$. Now suppose that 
    $\ell$ divides $\partial_{1} (q(\vecx) \cdot p(\vecx))$. We would still like to infer that $\ell$ divides $p$. 
    This is not true in general
    but with some mild conditions (a slightly more general version of) it does hold. 
    
\begin{claim}\label{clm:zeroLinearCombination}
    Let $q(\vecx) \in \R[\vecx]^{=d}$ be a homogeneous polynomial of degree $d$ and $p(\vecx) \in \R[\vecx]^{=k}$
    be a homogeneous polynomial of degree $k$. Let $\ell(\vecx) \in \R[\vecx]^{=1}$ be a linear form. For 
    $\veca = (a_1, a_2, \ldots, a_n) \in \R^n$, let  
    $(\veca \cdot \partial) \eqdef (a_1 \partial_1 + a_2 \partial_2 + \ldots + a_n \partial_n) \in \inangle{\partial^{=1}}$ be a first order differential operator. Suppose that there exists positive constants $c_1, c_2 \in \R_{>0}$ such that
        \begin{equation}\label{eqn:lcz}
            \ell \ \vert \ (c_1 ((\veca \cdot \partial) q) \cdot p + c_2 q \cdot ((\veca \cdot \partial) p)). 
        \end{equation}
    If 
        $$ \gcd( q, (\veca \cdot \partial) q) = 1 \quad (\text{mod~} \ell(\vecx)) $$
    then $\ell \  \vert \ p $.
\end{claim}
\begin{proof}
    By making a suitable change of variables if needed, we can assume without loss of generality that $\ell(\vecx) = x_1$ and 
    $\veca = (0, 1, 0, 0, \ldots , 0)$ so that $(\veca \cdot \partial)$ is simply $\partial_2$, i.e. 
    the derivative with respect to the variable $x_2$. Let $\hat{p} \eqdef p(0, x_2, x_3, \ldots, x_n)$
    and $\hat{q} \eqdef q(0, x_2, x_3, \ldots, x_n)$. The conclusion of the above claim can be restated 
    as claiming that $\hat{p} = 0$. Suppose not. Then we can write $\hat{p}$ as  
        $\hat{p} = \hat{q}^r \cdot f,$
    for some polynomial integer $r \geq 0$ and some polynomial $f \in \R[x_2, x_3, \ldots, x_n]$ which is not divisible by $\hat{q}$.
    Now equation (\ref{eqn:lcz}) implies that 
        \begin{eqnarray*}
            c_1 \cdot (\partial_{2} \hat{q}) \cdot  \hat{p} + c_2 \cdot \hat{q} \cdot (\partial_{2} \hat{p}) & = & 0 \\
            \implies (c_1 + r c_2) \cdot (\partial_{2} \hat{q}) \cdot f + c_2 \cdot \hat{q} \cdot (\partial_2 f) & = & 0
        \end{eqnarray*}
    from which we can infer that $\hat{q}$ divides $(c_1 + r c_2) \cdot (\partial_{2} \hat{q}) \cdot f$. But 
    $(c_1 + r c_2)$ is positive and hence nonzero and by assumption $ \gcd( \hat{q}, (\partial_2) \hat{q}) = 1 $
    so that $\hat{q}$ must divide $f$, a contradiction. Thus we must have $\hat{p} = 0$, or equivalently that 
    $\ell(\vecx)$ divides $p$.     
\end{proof}

\noindent A nondegenerate quadratic form (specifically, one which has rank at least $5$) satisfies the 
    desired property above. 

\begin{corollary}
    Let $\ell(\vecx) \in \R[\vecx]^{=1}$ be a linear form. For 
    $\veca = (a_1, a_2, \ldots, a_n) \in \R^n$, let  
    $(\veca \cdot \partial) \eqdef (a_1 \partial_1 + a_2 \partial_2 + \ldots + a_n \partial_n) \in \inangle{\partial^{=1}}$ be a first order differential operator.
    If $q(\vecx) \in \R[\vecx]^{=2}$ is a quadratic form of rank at least $5$ 
    and $\ell(\veca) = 0$ then it must hold that 
    $$ \gcd( q, (\veca \cdot \partial) q) = 1 \quad (\text{mod~} \ell(\vecx)). $$
    Consequently, for any polynomial $p(\vecx) \in \R[\vecx]^{=k}$, if $\ell(\vecx)$ divides 
    $(\veca \cdot \partial) (q \cdot p)$ then $\ell(\vecx)$ divides $p$.
\end{corollary}
\begin{proof}
    By making a suitable change of variables we can assume without loss of generality that 
    $\ell(\vecx) = x_1$ and that the differential operator $ (\veca \cdot \partial)$ is $\partial_2$ and that 
    $ q(\vecx) = x_1 \cdot p_1(\vecx) + x_2 \cdot p_2(x_3, x_4, \ldots , x_n) + c_3 \cdot x_3^2 \ldots + c_r \cdot x_r^2$, 
    where $p_1(\vecx), p_2(\vecx)$ are linear forms and $r$ is the rank of $q$. Then we have:
    \begin{eqnarray*}
        \gcd( q, (\veca \cdot \partial) q) \quad (\text{mod~} \ell(\vecx)) & = & \gcd(x_2 \cdot p_2 + c_3 \cdot x_3^2 \ldots + c_r \cdot x_r^2, p_2) \\
            & = & \gcd(c_3 \cdot x_3^2 \ldots + c_r \cdot x_r^2, p_2) \\
            & = & 1, \quad (\text{as~} (c_3 \cdot x_3^2 \ldots + c_r \cdot x_r^2) \text{ is irreducible for~} r \geq 5)
    \end{eqnarray*}    
    The second conclusion follows from claim \ref{clm:zeroLinearCombination} above.
\end{proof}

\noindent We are now ready to prove our main technical lemmas. 

\begin{proof}[Proof of Lemma \ref{lem:multiple}]
    Suffices to show that there exist constants $c_{11}, c_{12}, \ldots , c_{1s} \in \R$ such that 
    $D \cdot q_1^e \cdot f = (c_{11} q_1^e + c_{12} q_2^e + \ldots + c_{1s} q_s^e) \cdot f$ for all polynomials $f \in \R[\vecx]^{=k}$.
    Using the existence of nice bases of $\R[\vecx]^{=k}$ as provided by claim \ref{clm:niceBasis}, it suffices to prove the lemma for polynomials $f(\vecx) \in \R[\vecx]^{=k}$ of the form 
        $$ f(\vecx) = \ell_{1} \cdot \ell_2 \cdot \ldots \cdot \ell_k, $$
    where the $\ell_i$'s are coprime linear forms spanning only $2$ dimensions. 
    By making a suitable change of variables we can assume that each $\ell_i$ is a linear form over the variables $x_1$ and $x_2$ only.
    Let 
    \begin{align*}
        D \cdot (q_1^e \cdot f ) &= q_1^{e} \cdot p_1 + q_2^{e} \cdot p_2 + \cdots + q_s^e \cdot p_s, \quad \mathrm{and~}\\ D \cdot (q_1^e \cdot x_4 \ell_2 \ell_3 \ell_k) &= q_1^{e} \cdot \hat{p}_1 + q_2^{e} \cdot \hat{p}_2 + \cdots + q_s^e \cdot \hat{p}_s.
    \end{align*}
    where the $p_i$'s and $\hat{p}_j$'s are in $\R[\vecx]^{=k}$. Now since $f$ is a polynomial over only $x_1$ and $x_2$ we have 
    $\partial_3 f = 0$ which implies that
        \begin{eqnarray*}
            x_4 \cdot \partial_3 \cdot (q_1^e \cdot \ell_1 \ell_2 \ldots \ell_k) & = & \ell_1 \cdot \partial_3 (q_1^e \cdot x_4 \ell_2 \ell_3 \cdot \ell_k) \\
            \implies E \cdot (x_4 \cdot \partial_3 \cdot (q_1^e \cdot \ell_1 \ell_2 \ldots \ell_k)) & = & E \cdot (\ell_1 \cdot \partial_3 (q_1^e \cdot x_4 \ell_2 \ell_3 \cdot \ell_k)) \\
            \implies (x_4 \cdot \partial_3) (D \cdot (q_1^e \cdot \ell_1 \ell_2 \ldots \ell_k)) & = & (\ell_1 \cdot \partial_3) (D \cdot (q_1^e \cdot x_4 \ell_2 \ell_3 \cdot \ell_k)) \\
            \implies \sum_{i \in [s]} ((x_4 \cdot \partial_3) \cdot (q_i^e p_i) - (\ell_1 \cdot \partial_3) \cdot (q_i^e \cdot \hat{p}_i )) & = & 0 \\
            \implies \forall i \in [s]: ((x_4 \cdot \partial_3) \cdot (q_i^e p_i) - (\ell_1 \cdot \partial_3) \cdot (q_i^e \cdot \hat{p}_i )) & = & 0 \quad  \\
            (\text{as $(x_4 \cdot \partial_3) \cdot (q_i^e p_i) \in V_i$ and $ (\ell_1 \cdot \partial_3) \cdot (q_i^e \cdot \hat{p}_i ) \in V_i$} &\text{and} & \text{the $V_i$'s form a direct sum}) \\
            \implies \forall i \in [s]: \ell_1 & \vert & (x_4 \cdot \partial_3) \cdot (q_i^e p_i) \\
            \implies \forall i \in [s]: \ell_1 & \vert & (\partial_3) \cdot (q_i^e p_i) \\
            \implies \forall i \in [s]: \ell_1 & \vert & p_i.
        \end{eqnarray*}
    Similarly, we can show for all $j \in [k]$ and $i \in [s]$ that $\ell_j \vert p_i$. The $\ell_j$'s are 
    coprime and hence $\forall i \in [s]: (\prod_{j \in [k]} \ell_j ) \vert p_i$. But the $p_i$'s are of degree $k$ and hence we must have $p_i = c_{1i} \prod_{j \in [k]} \ell_{j}$ for some $c_{1i} \in \R$. 
    Moreover it also follows that the $c_{1i}$'s are in fact independent of the choice of $f$ and so we must have that $D \cdot (q_1^e \cdot f) = (c_{11} q_1^e + c_{12} q_2^e + \ldots + c_{1s} q_s^e) \cdot f$ for all polynomials $f \in \R[\vecx]^{=k}$. 
\end{proof}

\begin{proof}[Proof of Lemma \ref{lem:existsBl}]
Assume without loss of generality that $i=1$. Now, quadratic forms correspond to symmetric matrices which can be diagonalized over $\R$ by orthogonal matrices. 
So Suppose that $q_1(\vecy) = c_1 \ell_{1}(\vecy)^2 + c_2 \ell_{2}(\vecy)^2 + \ldots + c_n \ell_{n}(\vecy)^2$, where the $\ell_{j}$'s are pairwise orthogonal linear forms. By making a suitable orthonormal change of variables we can assume without loss of generality that 
    $$ q_1(\vecy) = c_1 y_1^2 + c_2 y_2^2 + \ldots + c_n y_n^2. $$
Consider the operator $B_1 \eqdef (c_1 y_1 \partial_2 - c_2 y_2 \partial_1) $ and let us apply it to polynomials
of the form $q_j(\vecy)^e \cdot x_3^k $. We have  
    $$         B_1 \cdot (q_j^e \cdot x_3^k) = e \cdot q_j^{e-1} \cdot x_3^k \cdot (c_1 \cdot y_1 \cdot (\partial_2 q_j) - c_2 \cdot y_2 \cdot (\partial_1 q_j) ) $$
Thus $B_1 \cdot (q_1^e \cdot x_3^k) = 0 $ and $B_1 \cdot (q_j^e \cdot x_3^k) \neq 0 $ unless 
$y_1 \vert (\partial_1 q_j)$ and $y_2 \vert (\partial_2 q_j)$. For $j \neq 1$, with probability 1, this latter condition does not hold when $q_j$ is chosen randomly and independent of $q_1$ and so $B_1 \cdot (q_j^e \cdot x_3^k) \neq 0 $. 
\end{proof}

\section{Learning Arithmetic Circuits in the Presence of Noise}\label{sec:aclProofs}

In this section, we will consider the problem of learning arithmetic circuits in the presence of noise, and follow the sketch in Section~\ref{sec:ov_ckts}.
Our algorithm is given below as Algorithm~\ref{alg:ckt_reconstruction} and it gets the guarantees given by Theorem~\ref{thm:robustCircuitReconstruction}.
Throughout this section, we shall assume that all vector spaces are inner product spaces, and we will work with appropriate orthonormal bases for the vector spaces.

\begin{algorithm}[H]
    \caption{Reconstructing the children of addition gates in the presence of noise.}\label{alg:ckt_reconstruction}
    \begin{algorithmic}
        \STATE \textbf{Input}: $(\tf(\vecx), s, d_U, d_V)$, where $\tf(\vecx) \in \R[\vecx]^{=d}$ is a polynomial, and $s, d_U, d_V$ are positive integers. 
        \STATE \textbf{Assumptions}: $f(\vecx) = T_1(\vecx) + T_2(\vecx) + \ldots + T_s(\vecx) \in \R[\vecx]^{=d}$ is a polynomial such that each $T_i(\vecx)$ belongs to a circuit class $\mathcal{C}$ that admits operators $\opL, \opB$ as follows:
        	There are vector spaces $W_1,W_2$, and a collection $\opL$ of linear maps $L : \R[\vecx]^{=d} \to W_1$, and a collection $\opB$ of linear maps $\opB:W_1\to W_2$ such that:
        	\begin{itemize}
        		 \item Let $U = \inangle{\opL\cdot f}$, and $U_i = \inangle{\opL\cdot T_i}$ for each $i\in [s]$. Then, $U=U_1\oplus\dots\oplus U_s.$ 
				\item Let $V = \inangle{\opB\cdot U}=\inangle{\opB\cdot\opL\cdot f}$,  and $V_i = \inangle{\opB\cdot U_i}=\inangle{\opB\cdot\opL\cdot T_i}$ for each $i\in [s]$. Then, $V=V_1\oplus\dots\oplus V_s.$
				\item The decomposition of $(U,V)$ under $\opB$ is strongly unique i.e. $\dim(\adjfull{\opB}{U}{V}) = s$.
			\end{itemize}
			The given integer inputs $d_U$ and $d_V$ are the dimensions of $U$ and $V$ respectively.
			The polynomial $\tilde{f}(\vecx)$ is such that $\tilde{f}(\vecx) = f(\vecx) + \eta(\vecx)$, with $\norm{\eta} \leq \epsilon$.
            
        \STATE \textbf{Output}: $\tT_1, \tT_2, \ldots, \tT_s \in \R[\vecx]^{=d}$ such that $\norm{T_i-\tT_i}$ is "small" for each $i\in [s]$ (upto reordering). 
    \end{algorithmic}
    
    \begin{algorithmic}[1]
        \STATEx
        \STATE Compute $\tU\subseteq W_1$ spanned by top $d_U$ left-singular vectors of $\tM$, the matrix with columns $(L\cdot \tilde{f})_{L\in \opL}$.
        
        \STATE Compute $\tV\subseteq W_2$ spanned by top $d_V$ left-singular vectors of $\tN$, the matrix with columns $(B\cdot L\cdot \tilde{f})_{B\in \opB,\ L\in \opL}$.
        
        \STATE Run RVSD algorithm on $(W_1,W_2,s,\tU,\tV,\opB)$\footnotemark; let the output be $\tvecU = (\tU_1,\dots,\tU_s)$, where $\tU = \tU_1\oplus\dots\oplus \tU_s$.
        

		\STATE Let $\hat{L}:\R[\vecx]^{=d}\to W_1^t$ be as in Definition~\ref{defn:joint_op}, where $\inabs{\opL} = t$.
            \Statex For each $i\in [s]$, let $\tP_i:W_1\to W_1$ be the map which is identity on $\tU_i$, zero on each $\tU_j$ for $j\not=i$, and zero on $\orth{\tU}$.
            \STATEx For each $i\in [s]$, let $ \Id_t\otimes \tP_i : W_1^t\to W_1^t$ be the map given by $\Id_t\otimes \tP_i \cdot (\vecw_1,\dots,\vecw_t) = (\tP_i \cdot \vecw_1, \dots, \tP_i \cdot \vecw_t)$.
		\STATEx For each $i\in [s]$, compute $\tT_i =  \hat{L}^\dag \cdot (\Id_t\otimes \tP_i)\cdot \hat{L}\cdot \tf$.
%
        \STATE Output $\tT_1, \tT_2, \ldots, \tT_s$.
    \end{algorithmic}
\end{algorithm}
\footnotetext{See Remark~\ref{remark:lac_inputs} for the parameter $\tau\in(0,1)$}

\begin{theorem}\label{thm:robustCircuitReconstruction}
    Let $f(\vecx) = T_1(\vecx)+ T_2(\vecx) + \cdots + T_s(\vecx)$ with $\norm{f}=1$ be a polynomial such that each $T_i \in \R[\vecx]^{=d}$ belongs to a circuit class $\mathcal{C}$ that admits operators $\opL, \opB$ as follows.
    
    Let $W_1,W_2$ be vector spaces, let $\opL$ be a collection of linear maps $L:\R[\vecx]^{=d} \to W_1$, and let $\opB$ be a collection of linear maps $B:W_1 \to W_2$, such that:
    \begin{itemize}
        \item Let $U = \inangle{\opL\cdot f}$, and $U_i = \inangle{\opL\cdot T_i}$ for each $i\in [s]$. Then, $U=U_1\oplus\dots\oplus U_s.$ 
		\item Let $V = \inangle{\opB\cdot U}=\inangle{\opB\cdot\opL\cdot f}$,  and $V_i = \inangle{\opB\cdot U_i}=\inangle{\opB\cdot\opL\cdot T_i}$ for each $i\in [s]$. Then, $V=V_1\oplus\dots\oplus V_s.$
		\item The decomposition of $(U,V)$ under $\opB$ is strongly unique i.e. $\dim(\adjfull{\opB}{U}{V}) = s$.
    \end{itemize}
    Consider the following:
    \begin{itemize}
    	\item Let $\vecU=(U_1,\dots,U_s)$, $\vecV=(V_1,\dots,V_s)$. Let $d_U = \dim(U)$, $d_V=\dim(V)$, and let $d^* = \max_{i \in [s]} \dim(U_i)$, $d_* = \min_{i \in [s]} \dim(U_i)$.
    	\item Let $M$ and $N$ be matrices with columns $(L \cdot f)_{L \in \opL}$ and $(B \cdot L \cdot f)_{B \in \opB,\ L \in \opL}$ respectively, and let $\sigma_{M,N} = \min\inbrace{\sigma_{d_U}(M), \sigma_{d_V}(N)}$.
    	\item Let $\adjmap$ be the adjoint algebra map corresponding to $(U,V,\opB)$ (see Definition~\ref{defn:adj_alg_map}; consider the restriction of maps in $\opB$ to $\Lin(U,V)$).
    	\item Let the collections $\opL, \opB$ be normalized such that $\lnorm{\opL} = 1$ and $\lnorm{\opB}=1$. Further, let $\hat{L}$ (see Definition~\ref{defn:joint_op}) be injective, and let $\kappa(\opL)$ be the condition number $\kappa(\hat{L})$.

        \item Let $\delta>0$ be arbitrary, and let $\theta = 10^6 \cdot \sqrt{\frac{{d^*}^3}{d_*}}\cdot \frac{s^{5/2}}{\delta} \sqrt{s+\ln\frac{s^2}{\delta}}  \cdot \kappa(\vecU)^5\cdot\kappa(\opL)\cdot\frac{1}{\sigma_{M,N}}\cdot \frac{1}{\sigma_{-(s+1)}(\adjmap)}  \cdot \epsilon < 1$.
    \end{itemize}
    
    Suppose that $\tf(\vecx) = f(\vecx) + \eta(\vecx)$ such that $\norm{\eta} \leq \epsilon$.
    Then, Algorithm~\ref{alg:ckt_reconstruction}, on input $(\tilde{f}, s, d_U, d_V)$, runs in time $\poly(n^d, \dim W_1, \dim W_2)$, and outputs $\tT_1, \tT_2, \ldots, \tT_s$ such that with probability at least $1 - \delta$, it holds (upto reordering) that for each $i\in [s]$,
    \[
    \norm{T_i-\tT_i} \leq \theta.
    \]
\end{theorem}
\begin{remark}\label{remark:lac_inputs}
    We note that it is possible to iterate over the parameter $\tau\in (0,1)$ (which is the input to the RVSD algorithm), and also the parameters $d_U, d_V, s$, assuming that there is a way to check the validity of the polynomials $\tT_1,\dots,\tT_s$ obtained (which is usually the case in applications).
    The reader is referred to  Remark~\ref{remark:param_tau}, Remark~\ref{remark:rvsd_rem}, and Remark~\ref{remark:sc_inputs} for more details. 
\end{remark}


In the remainder of this section, we shall prove Theorem~\ref{thm:robustCircuitReconstruction}.
It is not hard to see that the runtime is $\poly(n^d, \dim W_1, \dim W_2)$: we work with orthonormal bases for all the vector spaces throughout, and also assume (without loss of generality) that all operators in $\opL$ and $\opB$ are linearly independent. We will omit the details for this.

Throughout, we will be following the notation defined in in the statement of Theorem~\ref{thm:robustCircuitReconstruction}.
Also, let the operators $\hat{L}, \hat{B}$ be as defined in Definition~\ref{defn:joint_op}.

\subsection{Applying Robust Vector Space Decomposition}

As defined in Algorithm~\ref{alg:ckt_reconstruction}, let $\tM$ and $\tN$ be the matrices whose columns are $(L\cdot \tilde{f})_{L\in \opL}$ and $(B\cdot L\cdot \tilde{f})_{B\in \opB,\ L\in \opL}$ respectively.
Let $\tU$ (resp. $\tV$) be the vector space spanned by the top $d_U$ (resp. $d_V$) left singular vectors of $\tM$ (resp. $\tN$).

\begin{lemma}\label{lemma_ckt_recon_tUdistU}
	
    \[
        \dist(U,\tU) \le \frac{2\epsilon}{\sigma_{d_U}(M)}.
    \]
    \[
        \dist(V,\tV) \le \frac{2\epsilon}{\sigma_{d_V}(N)}.
    \]
\end{lemma}
\begin{proof}
	Observe that $U$ (resp. $V$) is the column space of the matrix $M$ (resp. $N$).
	Then, by Corollary~\ref{cor:sing_nullspacePerturb},
	\[
		\dist(U, \tU) \le \frac{2 \lnorm{M - \tM}}{\sigma_{d_U} (M)}, \quad \textbf\dist(V, \tV) \le \frac{2 \lnorm{N - \tN}}{\sigma_{d_V} (N)}.
	\]
	Further, we have 
	\[ \fnorm{M - \tM}^2 = \sum_{L \in \opL} \norm{ L \cdot (f-\tf)}^2 = \norm{\hat{L} \cdot \eta}^2 \leq  \lnorm{\opL}^2 \cdot \norm{\eta}^2 \leq  \lnorm{\opL}^2\cdot \epsilon^2=\epsilon^2.\]	
	Similarly,	
	\[ \fnorm{N - \tN}^2 = \sum_{B\in \opB}\sum_{L \in \opL} \norm{ B\cdot L \cdot \eta}^2 \leq   \sum_{L\in \opL} \lnorm{\opB}^2 \cdot\norm{L\cdot \eta}^2 \leq \lnorm{\opB}^2\cdot \lnorm{\opL}^2 \cdot \epsilon^2=\epsilon^2.\qedhere\]
\end{proof}

Now, suppose that the output of the RVSD algorithm is $\tvecU = (\tU_1,\dots,\tU_s)$.
Then, Lemma~\ref{lemma_ckt_recon_tUdistU}, along with Corollary~\ref{corr:rvsd_common_op} gives the following:
For any $\delta>0$, with probability at least $1-\delta$, we have (upto reordering) that for each $i\in [s]$,
\[ \dist(U_i,\tU_i) \leq 15000\cdot {\sqrt{\frac{{d^*}^3}{d_*}}\cdot \frac{s^{2}}{\delta} \sqrt{s+\ln\frac{s^2}{\delta}}  \cdot \kappa(\vecU)^3\cdot\frac{1}{\sigma_{M,N}}\cdot \frac{1}{\sigma_{-(s+1)}(\adjmap)}  \cdot \epsilon} \eqdef \gamma .
\]
Further, assuming that $2\gamma\sqrt{s}\cdot \kappa(\vecU)<1$, we know $\tU=\tU_1\oplus\dots\oplus\tU_s$; note that this assumption follows from the assumption that $\theta<1$ in the theorem statement.

\subsection{Recovering the Polynomials}

It remains to analyze the final step of the algorithm.
We shall assume that $\gamma\sqrt{s}\cdot \kappa(\vecU)<1/4$ (which follows from the assumption that $\theta<1$ in the theorem statement).

For each $i\in [s]$, let $P_i:W_1\to W_1$ (resp. $\tP_i$) be the map which is identity on $U_i$ (resp. $\tU_i$), zero on each $U_j$ (resp. $\tU_j$) for $j\not=i$, and zero on $\orth{U}$ (resp. $\orth{\tU}$); note that these exist since $U=U_1\oplus\dots\oplus U_s$ and $\tU=\tU_1\oplus\dots\oplus\tU_s$.

Suppose that $\inabs{\opL}=t$ (that is, $\opL$ has $t$ operators).
For each $i\in [s]$, let $ \Id_t\otimes P_i : W_1^t\to W_1^t$ (resp. $\Id_t\otimes \tP_i$) be the map given by $\Id_t\otimes P_i \cdot (\vecw_1,\dots,\vecw_t) = (P_i \cdot \vecw_1, \dots, P_i \cdot \vecw_t)$ (resp. $\Id_t\otimes \tP_i \cdot (\vecw_1,\dots,\vecw_t) = (\tP_i \cdot \vecw_1, \dots, \tP_i \cdot \vecw_t)$).

\begin{lemma}\label{lemma:ckt_recon_Pi_tPi}
	For all $i\in [s]$, it holds that
	 \[ \lnorm{P_i} \leq \kappa(\vecU), \quad \lnorm{P_i-\tP_i} \leq 14\gamma\sqrt{s}\cdot \kappa(\vecU)^2.\]	 
\end{lemma}
\begin{proof}
	Without loss of generality, we assume $W=\R^n$ with the usual inner product.
	For each $i\in [s]$, let $d_i=\dim(U_i)$; by the canonical decomposition (Theorem~\ref{thm:cs_decomp}), we find an orthonormal basis $\vecu_{i,1},\dots,\vecu_{i,d_i}\in \R^n$ of $U_i$ and $\tilde{\vecu}_{i,1},\dots,\tilde{\vecu}_{{i,d_i}} \in \R^n$ of $\tU_i$.
Let $M_U\in \R^{n\times d_U}$ (resp. $M_{\tU}$) be the $\vecU$ (resp. $\tvecU$)-associated matrix with columns $(\vecu_{i,j})_{i\in [s], j\in [d_i]}$ (resp. $(\tilde{\vecu}_{i,j})_{i\in [s], j\in [d_i]}$). 
Then, by the properties of the canonical decomposition, we have $\lnorm{M_U-M_{\tU}}\leq 2\gamma\sqrt{s}$ (note that this is essentially the same as Claim~\ref{claim:rrsm_direct_sum_NtN}).

Now, let $\Lambda_i \in \R^{d_U\times d_U}$ be a diagonal matrix defined as follows: let the $d_U$ diagonal elements be split into $s$ groups, of sizes $d_1,\dots,d_s$ respectively; define $\Lambda_i$ to have all ones in the $i\textsuperscript{th}$ group, and zero otherwise.
Then, we can write $P_i = M_U\cdot \Lambda_i \cdot M_U^\dag$ and $\tP_i = M_{\tU}\cdot \Lambda_i \cdot M_{\tU}^\dag$.
This gives us
\begin{enumerate}
	\item \[\lnorm{P_i} \leq \kappa(M_U)\cdot \lnorm{\Lambda_i} = \kappa(\vecU)\cdot 1. \]
	
    \item We assumed that $2\gamma\sqrt{s} \leq \frac{1}{2\kappa(\vecU)}\leq \frac{\sigma_d(M_U)}{2} < 1$.
    Then, $\lnorm{M_{\tU}}\leq \lnorm{M_U} + 2\gamma\sqrt{s} \leq 2 \lnorm{M_U}$. 
    Further, by Corollary~\ref{corr:pseudo_inv_perturbation}, we get
	\begin{align*}
		\lnorm{P_i-\tP_i} &\leq \lnorm{ (M_U-M_{\tU})  \cdot \Lambda_i\cdot {M_{U}^\dag}} + \lnorm{M_{\tU}  \cdot \Lambda_i\cdot (M_{U}^\dag- M_{\tU}^\dag) }
  \\&\leq \lnorm{ M_U-M_{\tU}}  \cdot \lnorm{\Lambda_i}\cdot \lnorm{M_{U}^\dag} + \lnorm{M_{\tU}}  \cdot \lnorm{\Lambda_i}\cdot \lnorm{M_{U}^\dag- M_{\tU}^\dag }
		\\&\leq  2\gamma\sqrt{s}\cdot 1 \cdot \frac{1}{\sigma_d(M_U)} + 2\lnorm{M_U}\cdot 1\cdot \frac{3\cdot 2\gamma\sqrt{s}}{\sigma_d(M_U)^2}.
   \\&\leq \frac{2\gamma\sqrt{s}}{\sigma_d(M_U)} + \frac{12\gamma\sqrt{s}\cdot \kappa(\vecU)}{\sigma_d(M_U)} \leq 14\gamma\sqrt{s}\cdot \kappa(\vecU)^2.
\qedhere	
 \end{align*}
\end{enumerate}
\end{proof}

\begin{proof}[Proof of Theorem~\ref{thm:robustCircuitReconstruction}]
    Fix any $i\in [s]$.
    Observe that for each $L\in \opL$, we have $L\cdot T_i = P_i\cdot L\cdot T_i = P_i\cdot (L\cdot T_1+\dots+L\cdot T_s) = P_i\cdot L\cdot f.$
    Hence, $\hat{L} \cdot T_i = (\Id_t\otimes P_i) \cdot \hat{L}\cdot f$, and since $\hat{L}$ is injective,
    \[ T_i = \hat{L}^\dag \cdot (\Id_t\otimes P_i) \cdot \hat{L}\cdot f. \]
    This gives us
    \begin{align*}
        \lnorm{T_i - \tT_i} &= \lnorm{\hat{L}^\dag \cdot (\Id_t\otimes P_i) \cdot \hat{L}\cdot f - \hat{L}^\dag \cdot (\Id_t\otimes \tP_i)\cdot \hat{L}\cdot \tf }
        \\&\leq \lnorm{\hat{L}^\dag} \cdot \lnorm{ (\Id_t\otimes P_i) \cdot \hat{L}\cdot f - (\Id_t\otimes \tP_i)\cdot \hat{L}\cdot \tf} 
        \\&\leq \lnorm{\hat{L}^\dag} \cdot \inparen{\lnorm{ (\Id_t\otimes P_i) \cdot \hat{L}\cdot (f-\tf) } + \lnorm{(\Id_t\otimes P_i-\Id_t\otimes \tP_i)\cdot \hat{L}\cdot \tf}  }
        \\&\leq \kappa(\opL) \cdot \inparen{\lnorm{ \Id_t\otimes P_i}\cdot \lnorm{f-\tf } + \lnorm{\Id_t\otimes P_i-\Id_t\otimes \tP_i}\cdot \lnorm{\tf}}
        \\&= \kappa(\opL) \cdot \inparen{\lnorm{P_i}\cdot \lnorm{f-\tf} + \lnorm{P_i-\tP_i}\cdot \lnorm{\tf}} 
        \\&\leq \kappa(\opL) \cdot \inparen{\kappa(\vecU)\cdot\epsilon + 14\gamma\sqrt{s}\cdot\kappa(\vecU)^2\cdot 2}
        \\& \leq 30 \cdot \kappa(\opL)\cdot \gamma\sqrt{s}\cdot \kappa(\vecU)^2. 
    \end{align*}
    Here we used Lemma~\ref{lemma:ckt_recon_Pi_tPi},  $\norm{\tf} \leq \norm{f}+\norm{\eta}\leq 1+\epsilon \leq 2$, and $\gamma>\epsilon$.
    Plugging in the value of $\gamma$, we get the desired result.
\end{proof}

\section{Analysis of the RRSM Algorithm}\label{sec:rrsm_analysis}


In this section, we will analyze our algorithm for the Robust Recovery from Scaling Maps problem, and prove Theorem~\ref{thm:rrsm}.

In the following subsections, first we fix some notation, and give a simple proof of Lemma~\ref{lemma:rrsm_simp_bound}.
Then, we will show some general properties about random matrices in the space of scaling maps.
We analyze the algorithm in Section~\ref{subsec:rrsm_second_alg}.
Finally, we prove the direct sum property in Section~\ref{sec:rrsm_direct_sum}.


\subsection{Notation}\label{subsec:rrsm_analysis_notation}

Let $\vecU = (U_1,\dots, U_s)$ be an independent $s$-tuple of subspaces of $W$, and $U = U_1\oplus U_2 \oplus \dots U_s \subseteq W$.
Let $\dim(W) = n$, $\dim(U_i) = d_i$, and $\dim(U) = d = d_1+\dots d_s$.
Let $d^*=\max_{i\in [s]}\dim(U_i)$ and $d_*=\min_{i\in [s]}\dim(U_i)$.
Since the algorithm is invariant to an orthogonal basis change, for the sake of the analysis, we can assume $W = \R^n$ with the canonical inner product.

Let $M \in \R^{n\times d}$ be a $\vecU$-associated matrix whose first $d_1$ columns form an orthonormal basis of $U_1$, the next $d_2$ columns form an orthonormal basis of $U_2$, and so on. 
In particular, we have $\kappa(\vecU) = \kappa(M)$.

For $\veclambda=(\lambda_1,\dots,\lambda_s)^\top \in \R^s$, define $\Lambda(\veclambda) = \diag(\underbrace{\lambda_1, \dots, \lambda_1}_{d_1\text{ times}},\dots,\underbrace{\lambda_s, \dots, \lambda_s}_{d_s\text{ times}}) \in \R^{d\times d}$.
For each $i \in [s]$, let $P_i = M\cdot \Lambda(\vece_i)\cdot M^\dag \in \R^{n\times n}$, where $\vece_i\in \R^s$ is the vector whose $i$\textsuperscript{th} coordinate is 1 and all other coordinates are 0. That is, the matrix $P_i$ corresponds to the linear map which is identity on $U_i$, and zero on each $U_j$ for $j\not=i$, and zero on $\orth{U}$.
Then, the matrices $(P_1,\dots, P_s)$ form a basis of the space $S = S(\vecU) \subseteq \R^{n\times n}$ of scaling maps  (when extended to $\R^{n\times n}$ appropriately; see Definition~\ref{defn:scaling_maps} and Definition~\ref{defn:rrsm_proj_maps}).

We define $\vecOp:\R^{n\times n}\to \R^{n^2}$ to be the linear operator that maps $n\times n$ matrices to their corresponding flattened out matrix in $\R^{n^2}$.
Let $\hat{M} \in \R^{n^2\times s}$ be a matrix whose columns are $\vecOp(P_1),\dots, \vecOp(P_s)$ (see Definition~\ref{defn:rrsm_proj_maps}).

~\\The algorithm has access to a subspace $\tU\subseteq \R^n$, and $\tS\subseteq \R^{n\times n}$ such that $\dist(S,\tS) \leq \eps$, and $\dist(U,\tU) <1$.
Assuming $\eps<1$, we know that $\dim(\tU)=d$ and $\dim(\tS) = s$.

\subsection{Condition Number Relations}\label{subsec:rrsm_analysis_simplified_bounds}

Now that we have established some notation, we begin by giving a simple proof of Lemma~\ref{lemma:rrsm_simp_bound}

Recall that $\hat{M} \in \R^{n^2\times s}$ is the matrix whose columns are $\vecOp(P_1),\dots, \vecOp(P_s)$, where for each $i\in [s]$, $P_i = M\cdot \Lambda(\vece_i)\cdot M^\dag \in \R^{d\times d}$.
Then, we have for each $\veclambda\in \R^s$, 
\[\lnorm{\hat{M}\veclambda}= \fnorm{M\cdot \Lambda(\veclambda)\cdot M^\dag} \leq \kappa(M)\cdot \fnorm{\Lambda(\veclambda)}\leq \kappa(M)\cdot\sqrt{\max_{i\in[s]} d_i}\cdot \lnorm{\veclambda},\]
\[\lnorm{\hat{M}\veclambda}= \fnorm{M\cdot \Lambda(\veclambda)\cdot M^\dag} \geq \sqrt{\sum_{i\in [s]}d_i \lambda_i^2} \geq \sqrt{\min_{i\in[s]} d_i}\cdot \lnorm{\veclambda}.\]
The inequality in the second line above follows from Proposition~\ref{prop:matrix_norms}.

In particular, we get that $\lnorm{\hat{M}}\leq \kappa(M)\cdot\sqrt{\max_{i\in[s]} d_i}$, and $\kappa(\hat{M})\leq \kappa(M)\cdot\sqrt{\frac{\max_{i\in[s]} d_i}{\min_{j\in[s]} d_j}}.$ \qed


\subsection{Canonical Decomposition and Random Sampling}

Now, we start analyzing the algorithm.
Observe that the algorithm, in the first step, choose a random matrix $\tA\in \tS$.
Using the canonical decomposition, we show that such a random matrix can be coupled with a random matrix $A\in S$, in a natural way.

Let $k=s-\dim(S\cap \tS)$.
Then, by Theorem~\ref{thm:cs_decomp}, we can find an orthonormal basis \\$E_1,\dots, E_s,\ F_1, \dots, F_k,\ H_1, \dots, H_{n^2-(r+k)}$ of $\R^{n\times n}$, and $\frac{\pi}{2} \geq \theta_1\geq \dots \geq \theta_k > 0$ such that:
    \begin{enumerate}
        \item $E_1,\dots, E_s$ form an orthonormal basis of $S$.
        \item For $i\in [k]$, let $\tE_i = \cos(\theta_i)\cdot E_i + \sin(\theta_i)\cdot F_i$, and for $i\in[s]\setminus[k]$, let $\tE_i=E_i$. Then, $\tE_1,\dots,\tE_s$ form an orthonormal basis of $\tS$.
        \item $\sin(\theta_1) = \dist(S,\tilde{S}) \leq \eps$. 
    \end{enumerate}
    
Under the above decomposition, we can couple random elements of $S$ and $\tS$, as follows:

\begin{definition}\label{defn:rrsm_coupled_rand_mat}(Coupled Random Matrices)
    Let $\vecalpha = (\alpha_1,\dots, \alpha_s)\in \R^s$, with each $\alpha_i\in\mathcal{N}(0,1)$ chosen independently.
    Let $\vecbeta = (\beta_1,\dots,\beta_s)\in \R^s$, with $\beta_i = \alpha_i/\lnorm{\vecalpha}$ for each $i\in [s]$.
    Define $A = \sum_{i=1}^s \beta_i E_i \in S$ and $\tA = \sum_{i=1}^s \beta_i \tE_i \in \tS$.

    Note that these satisfy $\fnorm{A} = \fnorm{\tA} = 1$ almost surely.
\end{definition}

Observe that the distribution of the matrix $\tA$ is independent of the choice of the orthonormal basis for $\tS$, and so the algorithm, in the first step, samples $\tA$ from the same distribution.

\begin{lemma}\label{lemma:A_tA_gap}
    \[ \fnorm{A-\tA} \leq 2\eps.\]
\end{lemma}
\begin{proof}
    This follows from Theorem~\ref{thm:cs_decomp} after observing that the relevant norm on $\R^{n\times n}$ is the Frobenius norm.
\end{proof}


\subsection{Bounding the Spectral Gap}

As we will see later, the correctness of the algorithm depends on the gap between the eigenvalues of $\hat{A}$ being large, and in this subsection we shall establish this fact.

Let $E\in \R^{n^2\times s}$ be the matrix whole columns are given by $\vecOp(E_1),\dots,\vecOp(E_s)$.
Recall that $\hat{M} \in \R^{n^2\times s}$ is a matrix whose columns are $\vecOp(P_1),\dots, \vecOp(P_s)$.

Let $\vecbeta$ be as in Definition~\ref{defn:rrsm_coupled_rand_mat}, and let $\veclambda = (\lambda_1, \dots, \lambda_s)^\top$ be such that  $ E\vecbeta = \vecOp(A) = \vecOp\inparen{M\cdot \Lambda(\veclambda)\cdot M^\dag} =  \hat{M} \veclambda$, or equivalently, $\veclambda = \hat{M}^\dag E \vecbeta$.
Define \[\gap(\veclambda) \eqdef \min\inbrace{\min_{i,j\in[s], i\not=j}\inabs{\lambda_i-\lambda_j},\ \min_{i\in[s]}\inabs{\lambda_i}}.\]

\begin{lemma}\label{rrsm_spectral_gap}
For any $\delta>0$, it holds that 
\[\Pr\insquare{\gap(\veclambda) \leq \frac{\delta}{6\cdot \lnorm{\hat M}\cdot  s^{2} \sqrt{s+\ln\frac{s^2}{\delta}}} } \leq \delta. \]
\end{lemma}
\begin{proof}
    For each $i\in[s]$, let $\vece_i\in \R^s$ be the vector whose $i$\textsuperscript{th} coordinate is 1 and all other coordinates are 0. 
    Then, for each $i,j\in[s], i\not=j,$ it holds that
    \[\inabs{\lambda_i} = \vece_i^\top\hat{M}^\dag E \vecbeta, \quad \inabs{\lambda_i-\lambda_j} = (\vece_i-\vece_j)^\top\hat{M}^\dag E \vecbeta.\]
    Now, observe that
    \begin{enumerate}
        \item For each $i,j\in[s], i\not=j$, 
        \[\lnorm{\vece_i^\top\hat{M}^\dag E} \geq \sigma_s(\hat{M}^\dag E), \quad \lnorm{(\vece_i-\vece_j)^\top\hat{M}^\dag E} \geq \sqrt{2}\sigma_s(\hat{M}^\dag E) \geq \sigma_s(\hat{M}^\dag E).\]
        
        \item Observe that $E, \hat{M}\in \R^{n^2\times s}$ are both rank $s$ matrices, each with column space $S$. Hence, by Lemma~\ref{lemma:pseudo_inv_same_col_space}, we have \[\sigma_s(\hat{M}^\dag E) = \frac{1}{\lnorm{\inparen{\hat{M}^\dag E)}^{-1}}} = \frac{1}{\lnorm{E^\top\hat{M}}} \geq \frac{1}{\lnorm{E^\top}\cdot \lnorm{\hat{M}}} = \frac{1}{\lnorm{\hat{M}}}.\]
    \end{enumerate}
    Then, the result follows from Lemma~\ref{lemma:gaussian_anti_conc}, and a union bound over the $s(s+1)/2 \leq s^2$ inequalities.
\end{proof}

From now on, we shall consider fixed $\delta>0$, and fixed $\lambda_1,\dots, \lambda_s$, and assume that $\gap(\veclambda) >\frac{\delta}{6\cdot \lnorm{\hat M}\cdot  s^{2} \sqrt{s+\ln\frac{s^2}{\delta}}}$ (by Lemma~\ref{rrsm_spectral_gap}, this occurs with probability at least $1-\delta$).

\subsection{Analysis of the Algorithm}\label{subsec:rrsm_second_alg}

\subsubsection{Action on The Space of Scaling Maps}

We consider matrices $\hat{A}, \hat{\tilde{A}} \in \R^{n^2\times n^2}$ as follows: these are the matrices corresponding to the linear maps defined in Definition~\ref{defn:rrsm_map_on_S} and Definition~\ref{defn:rrsm_map_on_tS} (and extending the maps to be zero on $\orth{S}, \orth{\tS}$ respectively), under the usual flattening of matrices by the map $\vecOp$.
Formally, we have that for each $B\in \R^{n\times n}$,

\[\hat{A}\cdot \vecOp(B) = \vecOp\inparen{\Proj_S(A\cdot \Proj_S(B))}, \]
\[\hat{\tilde{A}}\cdot \vecOp(B) = \vecOp\inparen
{\Proj_{\tS}(\tA\cdot \Proj_{\tS}(B))}, \]
where $\Proj_S, \Proj_{\tS}:\R^{n\times n}\to \R^{n\times n}$ are the orthogonal projections onto $S, \tS$ respectively.


\begin{lemma}\label{lemma:hatA_thatA_gap}
    \[ \lnorm{\hat{A} - \hat{\tilde{A}}} \leq 4\eps. \]
\end{lemma}
\begin{proof}
    For each $B\in \R^{n\times n}$, we have
    \begin{align*}
        \lnorm{(\hat{\tilde{A}} - \hat{A})\cdot \vecOp(B)} &= \fnorm{\Proj_S(A\cdot \Proj_S(B)) - \Proj_{\tS}(\tA\cdot \Proj_{\tS}(B))} 
        \\&\leq \fnorm{(\Proj_S-\Proj_{\tS})(A\cdot \Proj_S(B))} + \fnorm{\Proj_{\tS}((A-\tA)\cdot \Proj_S(B))}
        \\&\quad+\fnorm{\Proj_{\tS}(\tA \cdot( \Proj_S-\Proj_{\tS})(B))}.
    \end{align*}
    Then, using that $\lnorm{\Proj_S-\Proj_{\tS}} = \dist(S,\tS) \leq \eps$, and Lemma~\ref{lemma:A_tA_gap}, we have for $\fnorm{B}=1$ that
    \[\lnorm{(\hat{\tilde{A}} - \hat{A})\cdot \vecOp(B)} \leq \eps + 2\eps + \eps = 4\eps.\qedhere\]
\end{proof}

\subsubsection{Perturbation Bound}


Let $\gamma \eqdef 300\cdot\kappa(\hat{M})\cdot  \lnorm{\hat M}^2\cdot  s^{2} \sqrt{s+\ln\frac{s^2}{\delta}}\cdot \frac{\eps}{\delta}$ be such that $\gamma<1$, or else Theorem~\ref{thm:rrsm} holds trivially.
Let $ \eta \eqdef 100\cdot\kappa(\hat{M})\cdot  \lnorm{\hat M}\cdot  s^{2} \sqrt{s+\ln\frac{s^2}{\delta}}\cdot \frac{\eps}{\delta}$. 
Then, we have $\eta \leq \frac{\gamma}{3\lnorm{\hat{M}}} < 1$, since $\lnorm{\hat{M}}\geq \fnorm{P_1}\geq 1$.

Following Definition~\ref{defn:rrsm_map_on_S}, we can write
\[\hat{A} = \hat{M} \cdot \diag(\lambda_1,\dots,\lambda_s) \cdot \hat{M}^\dag.\]
This shows that $\hat{A}$ has exactly $s$ distinct non-zero eigenvalues (each with multiplicity one), which are equal to $\lambda_1,\dots,\lambda_s$, and these have eigenvectors $\vecOp(P_1),\dots, \vecOp(P_s)$ respectively.

\begin{lemma}\label{lemma:rrsm_hatA_eigendecomp}
\begin{enumerate}
	\item The matrix $\hat{\tilde{A}} \in \R^{n^2\times n^2}$ has $s$ non-zero, distinct eigenvalues, each of which is real and occurs with multiplicity 1.
		The eigenvalue $0$ occurs with multiplicity $n^2-s$, and so all the eigenvalues (and hence eigenvectors) are real.
	\item Let $\vecOp(\tP_1),\dots,\vecOp(\tP_s) \in \R^{n^2}$ be eigenvectors of $\hat{\tilde{A}}$ corresponding to the distinct non-zero eigenvalues, with $\fnorm{\tP_i}=1$ for each $i\in [s]$.
		Then, (upto reordering) for each $i\in [s]$, there exists $\zeta_i\in \inbrace{-1,1}$ such that
\[ \fnorm{\zeta_i\tP_i - \frac{P_i}{\fnorm{P_i}}} \leq \eta. \]
\end{enumerate}
\end{lemma}
\begin{proof}
	\begin{enumerate}
		\item Let $\theta>0$ be any number such that $\sigma_s(\hat{M})\leq \theta \leq \sigma_1(\hat{M})$.
		We can write \[\hat{A} = \hat{M}' \cdot \diag(\lambda_1,\dots,\lambda_s, \underbrace{0, \dots, 0}_{n^2-s \text{ times}}) \cdot (\hat{M}')^{-1},\] where $\hat{M}' \in \R^{n^2\times n^2}$ is the matrix whose first $s$ columns are the same as that of $\hat{M}$, and the remaining columns are an arbitrary orthogonal basis of $\orth{S}$ with each column having $\ell_2$-norm equal to $\theta>0$.
		The singular values of $\hat{M}'$ are precisely the singular values of $\hat{M}$, along with the value $\theta$ which occurs as a singular value $n^2-s$ times.
		By the choice of $\theta$, we have $\kappa(\hat{M}') = \kappa(\hat{M})$.
		
		Now, since $\kappa(\hat{M}') \cdot (4\epsilon) < \gap(\veclambda)/2$ (which is implied by $\eta< 1$), by Lemma~\ref{lemma:hatA_thatA_gap} and Lemma~\ref{lemma:eigenval_pert} we know that the eigenvalues of $\hat{\tilde{A}}$ can be written as $\tlb_1,\dots,\tlb_{n^2} \in \C$ such that
		\begin{itemize}
			\item For each $i\in [s]$, it holds that $\tlb_i\in \R$, and that $\inabs{\tlb_i-\lb_i} < \gap(\veclambda)/2$.
			In particular, the eigenvalues $(\tlb_i)_{i\in [s]}$ are distinct.
			\item For each $i\in [n^2]\setminus[s]$, it holds that $\inabs{\tlb_i}< \gap(\veclambda)/2$. In fact, since $\hat{\tilde{A}}$ is defined to be zero on $\orth{\tS}$, we know that the eigenvalue $0$ occurs with multiplicity at least $n^2-s$. Hence, each such $\tlb_i=0$.
		\end{itemize}		
		The above implies the first part of the Lemma.
		\item This follows directly from Lemma~\ref{lemma:eigenvec_pert}: \[ \fnorm{\zeta_i\tP_i - \frac{P_i}{\fnorm{P_i}}} \leq \frac{4\kappa(\hat{M}')\cdot 4\eps}{\gap(\veclambda)} = \frac{16\kappa(\hat{M})\cdot \eps}{\gap(\veclambda)} \leq \eta. \qedhere\]\qedhere
	\end{enumerate}
\end{proof}

\subsubsection{Recovering the Components}

Fix some $i\in[s]$, and let $d_i = \dim(U_i)$.  

We know that $\rank(P_i) = d_i$, and so by Lemma~\ref{lemma:rrsm_hatA_eigendecomp} and Weyl's inequality (Lemma~\ref{lemma:weyl_ineq}), the first $d_i$ singular values of $\tP_i$ (equal to those of $\zeta_i\tP_i$) are at least $\frac{\sigma_{d_i}(P_i)}{\fnorm{P_i}} - \eta \geq \frac{1}{\lnorm{\hat{M}}}-\eta$, where as the remaining ones are at most $\eta$.
Here, we used the following observations 
\begin{enumerate}
	\item $\fnorm{P_i} = \lnorm{\vecOp(P_i)} \leq \lnorm{\hat{M}}$.
	\item $\sigma_{d_i}(P_i)\geq 1$ since the map $P_i$ equals the identity map on the $d_i$-dimensional space $U_i$.
\end{enumerate}
Hence, if $\eta< \tau \leq \frac{1}{\lnorm{\hat{M}}} - \eta$, then the algorithm sets $\tU_i$ to be the span of the left singular vectors of $\tP_i$ (or equivalently, of $\zeta_i\tP_i$), corresponding to the top $d_i$ singular vectors.
For instance, this happens when $\eta \leq \frac{1}{3}\cdot\frac{1}{\lnorm{\hat{M}}}$ and $\frac{1}{3}\cdot\frac{1}{\lnorm{\hat{M}}}  < \tau \leq \frac{2}{3}\cdot\frac{1}{\lnorm{\hat{M}}} $.
Then, by Lemma~\ref{lemma:rrsm_hatA_eigendecomp} and Wedin's theorem (see Corollary~\ref{cor:sing_nullspacePerturb}), we get that: 
\[\dist(\tU_i, U_i) \leq \frac{2\cdot \fnorm{ \zeta_i\tP_i - \frac{P_i}{\fnorm{P_i}}}}{\sigma_{d_i}\inparen{\frac{P_i}{\fnorm{P_i}}}}
\leq \frac{2\eta\cdot \fnorm{P_i}}{\sigma_{d_i}(P_i)} \leq 2\eta\cdot \lnorm{\hat{M}} \leq \gamma. \pushQED{\qed}\qedhere\popQED\]



    

\subsubsection{Runtime Analysis}\label{subsec:rrsm_analysis_runtime}

In this section, we analyze the runtime of the algorithm.
Recall that $n = \dim(W),\ d=\dim(U)=\dim(\tU),\ s = \dim(S)=\dim(\tS)$, and let $d_i=\dim(\tU_i)$.
The input is given as $N=dn+sd^2$ field elements, consisting of an orthonormal basis of $\tU\subseteq W$ and an orthonormal basis of $\tS\in \Lin(\tU,\tU)$.
The time taken by each step of the algorithm is as follows:
\begin{enumerate}
    \item The random map $\tA\in \tS$ can be computed as a $d\times d$ matrix in time $O(sd^2)$, by taking a random linear combination of the basis elements of $\tS$.
    \item The map $\hat{\tilde{A}}\in \Lin(\tS,\tS)$ can be computed as an $s\times s$ matrix with respect to the orthonormal basis of $\tS$: If the basis is $\tilde{\vecs}_1,\dots, \tilde{\vecs}_s\in \R^{d\times d}$, for $i,j\in [s]$, the $(i,j)\textsuperscript{th}$ entry of this matrix equals $\inangle{\tA\cdot \tilde{\vecs}_i,\ \tilde{\vecs}_j}_F$, and can be computed in time $O(d^\omega)$, where $\omega$ is the matrix multiplication constant. So the total time taken in this step is $O(s^2d^\omega)$.
    \item The eigen-decomposition of $\hat{\tilde{A}}$ can be computed in time $O(s^3)$.
    \item Computing the eigenvectors $\tP_1,\dots,\tP_s$ as matrices in $\R^{d\times d}$, by taking appropriate linear combinations of the $\tS$ basis elements takes time $O(s\cdot sd^2)$.
    \item For each $i\in [s]$, the singular-value decomposition of $\tP_i$ can be computed in time $O(d^3)$, and then computing the basis vectors of $\tU_i$ as vectors in $\R^n$ takes time $O(d\cdot dn)$. Therefore, the total time to compute the basis elements for all the $\tU_i's$ is $O(sd^3+sd^2n)$.
\end{enumerate}

Finally, we get that the total runtime is $O(s^3+s^2d^\omega+sd^3+sd^2n) = O(N^{5/3})$, using $s\leq d\leq n$.


\subsection{Direct Sum Property}\label{sec:rrsm_direct_sum}

We give a short proof of Proposition~\ref{prop:rrsm_direct_sum}.
It suffices to show that $\tU_1, \dots, \tU_s$ form a direct sum.
The sum then equals $\tU$ by counting dimensions; note that $\dim(U)=\dim(\tU)$ and $\dim(U_i)=\dim(\tU_i)$ for each $i\in[s]$.

For each $i\in [s]$, by the canonical decomposition (Theorem~\ref{thm:cs_decomp}), we can find an orthonormal basis $\vecu_{i,1},\dots,\vecu_{i,d_i}\in \R^n$ of $U_i$ and $\tilde{\vecu}_{i,1},\dots,\tilde{\vecu}_{{i,d_i}} \in \R^n$ of $\tU_i$ 
such that for each $\alpha_1,\dots,\alpha_{d_i}\in \R$, we have \[\lnorm{\sum_{j\in d_i}\alpha_j\vecu_{i,j}-\sum_{j\in d_i}\alpha_j\tilde{\vecu}_{i,j}}\leq 2\gamma\cdot \sqrt{\sum_{j\in[d_i]}\alpha_j^2}.\]

Let $N\in \R^{n\times d}$ (resp. $\tN$) be the $\vecU$ (resp. $\tvecU$)-associated matrix with columns $(\vecu_{i,j})_{i\in [s], j\in [d_i]}$ (resp. $(\tilde{\vecu}_{i,j})_{i\in [s], j\in [d_i]}$).

\begin{claim}\label{claim:rrsm_direct_sum_NtN}
    \[ \lnorm{N-\tN}\leq 2\gamma\sqrt{s}. \]
\end{claim}
\begin{proof}
    Let $\vecalpha\in \R^d$ be indexed by $i\in [s], j\in [d_i]$.
    Then, we have
    \begin{align*}
        \lnorm{(\tN-N)\cdot \vecalpha} &= \lnorm{\sum_{i\in [s], j\in [d_i]} \alpha_{i,j}(\tilde{\vecu}_{i,j}-\vecu_{i,j}) }
        \\&\leq \sum_{i\in [s]} \lnorm{ \sum_{j\in [d_i]} \alpha_{i,j}(\tilde{\vecu}_{i,j}-\vecu_{i,j}) }
        \\&\leq \sum_{i\in [s]} 2\gamma\cdot \sqrt{\sum_{j\in [d_i]} \alpha_{i,j}^2 }
        \\&\leq 2\gamma \cdot \sqrt{s} \cdot \sqrt{\sum_{i\in [s], j\in [d_i]} \alpha_{i,j}^2 }.
    \end{align*}
\end{proof}

Now, using Claim~\ref{claim:rrsm_direct_sum_NtN} along with Weyl's inequality (Lemma~\ref{lemma:weyl_ineq}), we get
\[\sigma_d(\tN) \geq \sigma_d(N) -  2\gamma\sqrt{s} \geq \frac{1}{\kappa(\vecU)} -  2\gamma\sqrt{s} > 0,\]
under the assumption $2\gamma\sqrt{s}\cdot \kappa(\vecU)<1$.
This completes the proof.



%

\section{Analysis of the RVSD Algorithm }\label{sec:rvsd_analysis}

In this section, we will analyze Algorithm~\ref{alg:rvsd} and prove Theorem~\ref{thm:rvsd}.

\subsection{Notation}

Let $W_1,W_2$ be two vector spaces with $\dim(W_1) = n_1, \dim(W_2) = n_2$.
Without loss of generality (by an orthogonal transformation), we can assume $W_1=\R^{n_1}$ and $W_2 = \R^{n_2}$ with the canonical inner products on the two spaces.

Let $\vecU = (U_1,\dots, U_s)$ and $\vecV = (V_1,\dots, V_s)$ be independent $s$-tuples of subspaces in $W_1$ and $W_2$ respectively, and let $U = U_1\oplus U_2 \oplus \dots U_s\subseteq \R^{n_1}$, $V = V_1\oplus V_2 \oplus \dots V_s \subseteq \R^{n_2}$ be such that $\dim(U) = d_1, \dim(V) = d_2$.
Let $\opB = (B_1,\dots,B_m)\in \inparen{\R^{n_2\times n_1}}^m$ be an $m$-tuple of operators, with each $B_j$ being the zero map on $\orth{U}$, and such that for each $i\in [s]$, it holds that $\inangle{\opB\cdot U_i}\subseteq V_i$.

The algorithm has access to a subspaces $\tU\subseteq \R^{n_1}$ and $\tV\subseteq \R^{n_2}$ such that $\dist(U,\tU) \leq \eps_1$ and $\dist(V,\tV) \leq \eps_2$; assuming $\eps_1,\eps_2<1$, we have $\dim(\tU) = d_1, \dim(\tV) = d_2$.
We also know $\topB=(\tB_1,\dots,\tB_m)$, with each $\tB_j$ being the zero map on $\orth{\tU}$, such that $\opB$ and $\topB$ are $\eps$-close: Let $\hat{B} \in \R^{mn_2\times n_1}$ (resp. $\hat{\tilde{B}} \in \R^{mn_2\times n_1}$) be matrices formed by stacking the rows of $B_1,\dots B_m$ (resp. $\tB_1,\dots \tB_m$). Then, $\lnorm{\hat{\tilde{B}} - \hat{B}} \leq \eps \lnorm{\hat B}$ (see Definition~\ref{defn:joint_op}).

\subsection{Perturbation Bound on the Adjoint Algebra}

Let $\Proj_{U}, \Proj_{\tU}\in \R^{n_1\times n_1}$ and $\Proj_{V}, \Proj_{\tV}\in \R^{n_2\times n_2}$ denote the orthogonal projections onto $U, \tU, V, \tV$ respectively.
These also give us the orthogonal projection maps on the spaces of linear maps, for example, $\Proj_{\Lin(U,U)}:\R^{n_1\times n_1}\to \R^{n_1\times n_1}$ is given by $\Proj_{\Lin(U,U)}(D) = \Proj_U\cdot D\cdot \Proj_U$, where we think of $\Lin(U,U)$ as a subspace of $\R^{n_1\times n_1}$ in the natural way.

Let $\adjmap:\R^{n_1\times n_1}\times \R^{n_2\times n_2} \to (\R^{n_2\times n_1})^m$ be the map defined as in Definition~\ref{defn:adj_alg_map}, extended to be the zero map on $\orth{\Lin(U,U)}\times \orth{\Lin(V,V)}$, given by 
\[\adjmap(D,E) = \inparen{B_i\cdot \inparen{\Proj_U\cdot D\cdot \Proj_U}-\inparen{\Proj_V\cdot E\cdot \Proj_V}\cdot B_i}_{i = 1}^m\] 

Letting $\adjker=\ker(\adjmap)$. Then, we have that $\adj = \adjker\cap \inparen{\Lin(U,U)\times \Lin(V,V)} \subseteq \R^{n_1\times n_1}\times \R^{n_2\times n_2}$ (see Definition~\ref{defn:adj_alg}).
Here, we think of $\Lin(U,U)$ as a subspace of $\R^{n_1\times n_1}$, by extending each map to be zero on $\orth{U}$, and similarly we think of $\Lin(V,V)$ as a subspace of $\R^{n_2\times n_2}$.

Note that assuming $\dim(\adj) = s$, we have that $\dim(\adjker) = (n_1^2-d_1^2)+(n_2^2-d_2^2)+s$.
We shall use $\sigma_{-(s+1)}(\adjmap)$ 
 to denote $(\dim(\adjker)+1)\textsuperscript{th}$ smallest singular value of $\adjmap$ (this is also the smallest non-zero singular value).
Note that this equals the value $\sigma_{-(s+1)}(\adjmap)$ as defined in Theorem~\ref{thm:rvsd}.

In a similar manner, we define $\tadjmap:\R^{n_1\times n_1}\times \R^{n_2\times n_2} \to (\R^{n_2\times n_1})^m$ as in Definition~\ref{defn:tadj_alg_map}, extended to be the zero map on $\orth{\Lin(\tU,\tU)}\times \orth{\Lin(\tV,\tV)}$, given by
\[\tadjmap(D,E) = \inparen{\tB_i\cdot \inparen{\Proj_{\tU}\cdot D\cdot \Proj_{\tU}}-\inparen{\Proj_{\tV}\cdot E\cdot \Proj_{\tV}}\cdot \tB_i}_{i = 1}^m\]
Further, by Definition~\ref{defn:tadj_alg_map}, letting $\tadjker$ denote the space spanned by the right singular vectors of $\tadjmap$, corresponding to the $(n_1^2-d_1^2)+(n_2^2-d_2^2)+s$ smallest singular values, we have $\tadj = \tadjker\cap \inparen{\Lin(\tU,\tU)\times \Lin(\tV,\tV)}$, where again we think of $\Lin(\tU,\tU)$ and $\Lin(\tV,\tV)$ as subspaces of $\R^{n_1\times n_1}$ and $\R^{n_2\times n_2}$ respectively.

\begin{lemma}\label{lemma:adjmap_gap}
    \[\lnorm{\adjmap-\tadjmap} \leq 2(\eps+\eps_1+\eps_2)\lnorm{\opB}.\]
\end{lemma}
\begin{proof}
    Let $(D,E)\in \R^{n_1\times n_1}\times \R^{n_2\times n_2}$ be such that $\fnorm{D}^2+\fnorm{E}^2=1$.
    Then, by triangle inequality, and sub-multiplicativity of Frobenius norm (see Proposition~\ref{prop:matrix_norms}), we have that 
    \begin{multline*}
        \lnorm{\tadjmap(D,E)-\adjmap(D,E)} \leq \lnorm{\hat{\tilde{B}} - \hat{B}}\cdot \fnorm{D} \\ + \lnorm{\hat{\tilde{B}} - \hat{B}}\cdot \fnorm{E} + 2\lnorm{\hat{B}}\cdot \lnorm{\Proj_{\tU}-\Proj_{U}}\cdot\fnorm{D} + 2\lnorm{\hat{B}}\cdot \lnorm{\Proj_{\tV}-\Proj_{V}}\cdot\fnorm{E}.
    \end{multline*}
    The lemma then follows.
\end{proof}

\begin{lemma}\label{lemma:adjker_dist}
    \[\dist(\tadjker,\adjker) \leq \frac{4(\eps+\eps_1+\eps_2)\lnorm{\opB}}{\sigma_{-(s+1)}(\adjmap)}. \]
\end{lemma}
\begin{proof}
    This follows from Lemma~\ref{lemma:adjmap_gap} and Corollary~\ref{cor:sing_nullspacePerturb}.
\end{proof}

\begin{lemma}\label{lemma:tadj_adj_dist}
    \[\dist(\tadj,\adj) \leq \frac{4(\eps+\eps_1+\eps_2)\lnorm{\opB}}{\sigma_{-(s+1)}(\adjmap)} + 2\eps_1+2\eps_2 \leq \frac{6(\eps+\eps_1+\eps_2)\lnorm{\opB}}{\sigma_{-(s+1)}(\adjmap)}.\]
\end{lemma}
\begin{proof}
    We have 
    \begin{align*}
        \dist(\tadj,\adj) &= \dist\inparen{\tadjker\cap \inparen{\Lin(\tU,\tU)\times \Lin(\tV,\tV)}, \adjker\cap \inparen{\Lin(U,U)\times \Lin(V,V)}}
        \\&\leq \dist(\tadjker, \adjker) + \dist\inparen{\Lin(\tU,\tU)\times \Lin(\tV,\tV), \Lin(U,U)\times \Lin(V,V)}
        \\&\leq \dist(\tadjker, \adjker)+ \dist\inparen{\Lin(\tU,\tU), \Lin(U,U)} + \dist\inparen{\Lin(\tV,\tV), \Lin(V,V)}.
    \end{align*}
    The first term is bounded by Lemma~\ref{lemma:adjker_dist}.
    For the second term, we observe that for any $D\in \R^{n_1\times n_1}$ with $\fnorm{D}=1$, we have
    \[\lnorm{\Proj_{\Lin(\tU,\tU)}(D) - \Proj_{\Lin(U,U)}(D)} = \fnorm{\Proj_{\tU}\cdot D\cdot \Proj_{\tU} - \Proj_{U}\cdot D\cdot \Proj_{U}} \leq 2\eps_1 .\]
    The third term is bounded similarly.

    For the last inequality, it suffices to show that $\sigma_{-(s+1)}(\adjmap)\leq \lnorm{\opB}$: For any $(D,E)$, with $E=0$, $\fnorm{D}=1$, and $D\in \Lin(U_1,U_2)\subseteq \R^{n_1\times n_1}$, we have $(D,E)\in \orth{\adjker}$ and so
    \[\sigma_{-(s+1)}(\tadjmap) \leq \lnorm{\adjmap\cdot (D,E)} = \lnorm{\hat{B}\cdot D} \leq \lnorm{\hat{B}} = \lnorm{\opB}.\]
    Note that we assumed $s\geq 2$, which is without loss of generality, as in the $s=1$ the RVSD problem is trivial.
\end{proof}

This also gives the following:

\begin{lemma}\label{lemma:tadj1_adj1_dist}
    \[\dist(\tadj_1,\adj_1) \leq \frac{6(\eps+\eps_1+\eps_2)\lnorm{\opB}}{\sigma_{-(s+1)}(\adjmap)}.\]
\end{lemma}

\subsection{Applying RRSM and Recovering The Component Subspaces}

By Theorem~\ref{thm:rrsm}, we get the following:
Let $\hat{M}:\R^s\to \Lin(U,U)$ be as defined in Definition~\ref{defn:rrsm_proj_maps}, and suppose that $\tau\in(0,1)$ satisfies that $\frac{1}{3}\cdot\frac{1}{\lnorm{\hat{M}}}  < \tau \leq \frac{2}{3}\cdot\frac{1}{\lnorm{\hat{M}}}$.
Then, for any $\delta>0$, with probability at least $1-\delta$, we get (upto reordering) that for each $i\in [s]$,
\[\dist(U_i, \tU_i) \leq 300\cdot\kappa(\hat{M})\cdot  \lnorm{\hat M}^2\cdot  s^{2} \sqrt{s+\ln\frac{s^2}{\delta}}\cdot \frac{1}{\delta}\cdot \frac{6(\eps+\eps_1+\eps_2)\lnorm{\opB}}{\sigma_{-(s+1)}(\adjmap)}. 
\]

The direct sum property $\tU = \tU_1\oplus\dots\oplus\tU_s$ follows by Proposition~\ref{prop:rrsm_direct_sum}.
\qed


\subsection{Runtime Analysis}

In this section, we analyze the runtime of the Algorithm~\ref{alg:rvsd}.
Recall that $n_1 = \dim(W_1)$, $n_2 = \dim(W_2),\ d_1=\dim(U)=\dim(\tU),\ d_2=\dim(V)=\dim(\tV)$.
The input is given as $N=n_1d_1+n_2d_2+md_1d_2$ field elements, consisting of an orthonormal basis of $\tU\subseteq W_1$ and $\tV\subseteq W_2$, and a description of $m$-tuple $\opB=\Lin(\tU,\tV)^m$ as matrices with respect to the above basis of $\tU$ and $\tV$.
The time taken by each step of the algorithm is as follows.
\begin{enumerate}
    \item The adjoint algebra map can be computed as a $(md_1d_2)\times (d_1^2+d_2^2)$ matrix. Setting $K = md_1d_2\cdot (d_1^2+d_2^2)$, this can be done in time $ O(K)$, as each entry is simply ($\pm 1$ times) some entry of some $B_j, j\in [m]$.
    \item The singular value decomposition of $\tadjmap$ can be computed in time $O(K^3)$, from which a basis of $\tadj$ can be obtained in time $O(s \cdot (d_1^2+d_2^2))$, and further a basis of $\tadj_1$ can be obtained in time $O(sd_1^2)$.
    \item Assuming we run the second algorithm for RRSM (see Algorithm~\ref{alg:rrsm}, Theorem~\ref{thm:rrsm}), the last step takes time $O(s^3+ s^2d_1^\omega+sd_1^3+sd_1^2n_1).$
\end{enumerate}

Finally, we get that the total runtime is $O(\inparen{md_1d_2\cdot (d_1^2+d_2^2)}^3 + sd_1^2n_1) = O(N^6)$, using $s\leq d_1\leq n_1$ and $s\leq d_2\leq n_2$.

\newpage
\section{Singular Values of the Adjoint Algebra Operator for Subspace Clustering}\label{sec:singval_AdjointOperator}

Section \ref{sec:scrProofs} discusses the problem of Subspace Clustering and culminates in theorem \ref{thm:sc} where the robustness of the proposed algorithm is quantified. The quantity $\sigma_{-(s + 1)}(\adjmap)$ appears in the bounds. In this section we will compute a lower bound for this quantity. \\

The high level procedure for getting the bound is as follows. Let $U = U_1 \oplus \ldots \oplus U_s$ and $V = V_1 \oplus \ldots \oplus V_s$. Since the operators $\opB$ map each $U_i$ to $V_i$, the action of the adjoint algebra operator $\adjmap$ on input matrices $(D,E) \in \Lin(U, U) \times \Lin(V, V)$ can be separated into independent actions of smaller adjoint algebra operators $\adjmap_{jk}$ acting on matrices $(D_{jk}, E_{jk}) \in \Lin(U_k, U_j) \times \Lin(V_k, V_j)$, for $k,j \in [s]$.  Now, choosing a slightly modified inner product on these smaller spaces ensures that the map $\adjmap_{jk}^T \adjmap_{jk}$ can be roughly expressed as $(I - \Psi_{jk})$. These new maps $\Psi_{jk}$ depend on the scaled partial derivative operators whose transposes turn out to be a sum of shifts. Here, we make the crucial observation that the derivatives of shifts on a polynomial space of degree $d + 1$ can be converted to shifts of derivatives on a polynomial space of degree $d$. This leads to an intricate inductive argument to calculate the singular values of  $\Psi_{jk}$ by induction on the degree $d$ of the homogeneous polynomial space $U_k$. These singular values naturally fetch the required singular values of $\adjmap$, besides also allowing us to meaningfully relate them to the geometry of the underlying subspaces. \\

We prove the following theorem in the coming sections.
    \begin{theorem}\label{thm:adjoint_algebra_sc_singular_values}
        Let $\adjmap, \vecU$ and $\vecV$ be as in theorem \ref{thm:sc} for clustering the collection $A = \{\veca_1, \dots, \veca_N\}$ into subspaces $\inangle{A_1}, \dots, \inangle{A_s}$. For subspaces $\inangle{A_j}, \inangle{A_k}$ with canonical angles $\theta_1 \ge \cdots \ge \theta_t$ between them, define $f_{jk} = f_d(\inangle{A_j}, \inangle{A_k}) = \frac{d + 1}{t} \left[\sum_{k = 1}^t \sin^2 \theta_k + d \sin^2 \theta_{t}\right]$. Then,
            $$\sigma^2_{-(s + 1)}(\adjmap) \geq \tfrac{(d + 1)^2}{\kappa^4(\vecU, \vecV)} \cdot \min\{\sigma_\diag, \sigma_\offdiag \}$$
            
        \noindent where the above quantities are defined as follows:
        \begin{align*}
               \sigma_\diag & \eqdef \tfrac{\sqrt{(d + 1)(t^* + d)} - \sqrt{(d + 1)(t^* + d) - t^*}}{t^* + d}, \\
            \sigma_\offdiag & \eqdef \min_{j \neq k} \tfrac{\sqrt{(d + 1)(t_k + d)} - \sqrt{(d + 1)(t_k + d) - t_k \cdot f_{jk}}}{t_k + d}.
        \end{align*}
        
        \noindent Here $t_1, \ldots, t_s$ are the dimensions of $\inangle{A}_1, \ldots, \inangle{A}_s$ and $t^* = \max_{i \in [s]}t_i$.
    \end{theorem}
Further, remark \ref{remark:f_d_geometry} provides further interpretation of the relation between the quantity $f_d$ and $\sigma_{-(s + 1)}(\adjmap)$. Note that in the statement of Theorem \ref{thm:adjoint_algebra_sc_singular_values} we assume that $\inangle{\vecU} \subseteq \R[\vecx]^{=d+1}$ and $\inangle{\vecV} \subseteq \R[\vecx]^{=d}$.

\subsection{Adjoint Algebra Operators corresponding to Partial Derivatives on Tensored Spaces}

We will begin by estimating the singular values of the adjoint algebra operator corresponding to a simplified subspace clustering instance when $s = 1$ (refer to section \ref{sec:scOverview}). This means that we study the relevant operators on homogeneous polynomial spaces like $\R[\vecx]^{=d}$, instead of a direct sum of many such spaces. As mentioned in the section overview, we will eventually relate the operators on the summed spaces to simpler operators which we will study here. We will take the aid of a special inner product to ease our calculations. 

\subsubsection{Matrix Representations of Derivatives and Shifts} 

Let $W = \begin{bmatrix} \vecw_1 & \dots & \vecw_m \end{bmatrix}$, where $\vecw_i \in \R^n$ be a matrix with orthonormal columns. Given the set of variables $\vecx = (x_1, \dots, x_n)$ and positive integer $d$ define the following:
    \begin{align}
        \nonumber
        \vecy \eqdef & \enspace W^T \cdot \vecx \text{ a new set of variables,}\\
        \nonumber
        U \eqdef & \enspace \R[\vecy]^{= d + 1} \text{ of dimension } m^{(d + 1)} = \textstyle\binom{m + d}{d + 1}, \\
        V \eqdef & \enspace \R[\vecy]^{= d} \text{ of dimension } m^{(d)} = \textstyle\binom{m + d - 1}{d}. \label{def:m^(d)}
    \end{align}

\noindent For $i \in [n]$ we would like to compute the matrix representation of the scaled partial derivatives $\{ L_i \} = \left\{\frac{\partial_i}{d + 1}\right\}$ on $U$ with respect to $x_i$s. \\

Let $\{p_{\vecalpha}\}_{\vecalpha}$ and $\{q_{\vecbeta}\}_{\vecbeta} $ be the Bombieri basis of $U$ and $V$ respectively with respect to the variables $y$. Note here that 
    \begin{equation*}
     (L_i)_{\vecbeta \vecalpha} = \IPP{q_{\vecbeta}}{L_i p_{\vecalpha}}
                                 = \tfrac{1}{d + 1} \sqrt{\tfrac{(d + 1)!}{\vecalpha!}} \enspace \IPP{q_\vecbeta}{\D_i \vecy^\vecalpha}
                                 = \sqrt{\tfrac{d!}{\vecalpha! (d + 1)}} \enspace \IPP{q_\vecbeta}{\sum_{j = 1}^m \vecalpha_j \vecy^{\vecalpha - j} w_{ji} },
    \end{equation*}
where $w_{ji}$ is the $i$-th coordinate of $\vecw_j$. Now, note that given $\vecalpha$ and $\vecbeta$ there may exist a $j$ for which $\vecalpha_k = \vecbeta_k$ for $k \neq j$, and $\vecalpha_j = \vecbeta_j + 1$. This condition is compactly written as $\vecalpha = \vecbeta \cup j$ or $\vecbeta = \vecalpha - j$. In this case, such a $j$ is unique. If no such $j$ exists, $\vecy^\vecbeta$ and $\vecy^{\vecalpha - j}$ are orthogonal. This gives the matrix representation of $L_i$ as
    \begin{equation}\label{def:derivative_matrix}
    (L_i)_{\vecbeta \vecalpha} = \begin{cases}
                                    w_{ji}\sqrt{\frac{\vecalpha_j}{d + 1}} & \text{ if $\exists j \in [m]$ such that $\vecalpha = \vecbeta \cup j$,}\\
                                    0                                      & \text{ otherwise.}
                                  \end{cases}
    \end{equation}

\noindent Using the above representation we can find the action of $L_i^T$ on $q_\vecbeta$ as follows:
    \begin{align}
    \nonumber
     & L_i^T q_\vecbeta = \sum_\vecalpha (L_i^T)_{\vecalpha \vecbeta} \cdot p_\vecalpha
                      = \sum_{j = 1}^m w_{ji}\sqrt{\tfrac{\vecbeta_j + 1}{d + 1}} \cdot p_{\vecbeta \cup j}
                      = \sum_{j = 1}^m w_{ji} y_j \cdot q_{\vecbeta} \\
               \implies & L_i^T q_\vecbeta =  \left(\sum_{j = 1}^m w_{ji} y_j \right) q_\vecbeta. \label{def:shift_matrix}
    \end{align}

Now, to distinguish the scaled partial derivatives acting on $U$ and $V$ we denote them $\upL_i$ and $\downL_i$ respectively. Then for any $q \in V$ we have 
    \begin{align*}
        \left(\upL_{i} \upL_{j}^T \right) q 
        &= \frac{\D_i}{d + 1} \left(\sum_{k = 1}^m w_{kj} y_k \cdot q\right) \\
        &= \frac{1}{d + 1} \left[ \sum_{k = 1}^m w_{kj} \enspace w_{ki} \cdot q + 
           \sum_{k = 1}^m w_{kj} y_k \cdot \D_i q \right] \\
        &= \left[ \frac{r_{ij}}{d + 1} +    \frac{d}{d + 1} \downL_{j}^T \downL_{i} \right] q,
    \end{align*}
    
\noindent where $r_{ij} \eqdef \sum_{k = 1}^m w_{kj} w_{ki}$. Further, for any monomial $\vecy^\vecalpha$ and the operator $L_i$ (irrespective of the degree of its polynomial domain) we have.
\begin{align*}
        \left( \sum_{i = 1}^n L_i^T L_i \right) \vecy^\vecalpha 
        &= \sum_{j = 1}^m \sum_{i = 1}^n \vecalpha_j w_{ji} \cdot L_i^T  \vecy^{\vecalpha - j} \\                       
        &= \frac{1}{d + 1} \sum_{j, k = 1}^m \vecalpha_j \vecy^{\vecalpha - j + k} \left(\sum_{i = 1}^n  w_{ji} w_{ki} \right) \\
        &= \frac{1}{d + 1} \sum_{j = 1}^m \vecalpha_j \vecy^\vecalpha \\ 
        &= \vecy^\vecalpha.
    \end{align*}

\noindent Hence the above equations give the following:
    \begin{equation}\label{eq:chainrule:derivatives_of_shifts} 
        \upL_i \upL_j^T = \frac{r_{ij}}{d + 1} I +    \frac{d}{d + 1}  \downL_j^T \downL_i \text{  and  }
        \sum_{i = 1}^n L_i^T L_i = I.
    \end{equation}
    
\noindent Subsequently, we use the above identities with appropriate dimensions to get
    \begin{equation}\label{eq:gramsum_Ai_transpose}
        \sum_{i = 1}^n \upL_i \upL_i^T = \frac{I}{d + 1} \sum_{i = 1}^n r_{ii} + \frac{d}{d + 1} \sum_{i = 1}^n \downL_i^T \downL_i 
                         = \frac{m + d}{d + 1} I
                         = \frac{m^{(d + 1)}}{m^{(d)}} I.
    \end{equation}
    
 The above relations will turn out to be crucial in the following sections for the singular value analysis of relevant operators.


\subsubsection{Adjoint Algebra Operator in a Special Inner Product}

Let $x, y$ be vectors in an arbitrary $l$-dimensional vector space $H$. Then for positive numbers $t_1, \dots, t_l$, the bilinear map
    \begin{equation}\label{def:inp_tau_basic}
        \IPT{x}{y} \eqdef \sum_{i = 1}^l \frac{x_i \cdot y_i}{t_i} 
    \end{equation}
defines an inner product. The following lemma builds on this special inner product. 
\begin{lemma}\label{lemma:inequality_singular_values}
    Let $A \in \Lin(H, H)$. Let $\tau_1(A), \dots, \tau_l(A)$ denote the singular values of $A$ with respect to the special inner product as defined in equation \ref{def:inp_tau_basic}. Then for all $i \in [l]$
        \begin{equation*}
            \frac{t_*}{t^*} \en \sigma_i^2(A) \leq \tau^2_i(A) \leq \frac{t^*}{t_*}\en \sigma_i^2(A).
        \end{equation*}
    
    \noindent where $t_* = \min_i t_i$ and $t^* = \max_i t_i$. The above inequality is tight.
\end{lemma}

\begin{proof}
    Let us first prove the above inequality for the largest singular values $\sigma_1(A)$ and $\tau_1(A)$. For any vector $x \in H$ let $\norm{x}_\tau$ denote its norm that arises from the special inner product. Then 
        \begin{align*}
            \norm{x}_\tau^2 = \sum_{i = 1}^l \frac{x_i^2}{t_i}.
        \end{align*}

    \noindent As all the numbers involved above are positive we immediately get the inequality
        \begin{equation*}
             \frac{\norm{x}^2}{t^*} \leq \norm{x}_\tau^2 \leq \frac{\norm{x}^2}{t_*}.
        \end{equation*}

    \noindent Then observe that
        \begin{align*}
            \norm{Ax}_\tau^2
            \leq \frac{1}{t_*} \norm{Ax}^2
            \leq  \frac{1}{t_*} \en \sigma_1^2(A) \norm{x}^2
            \leq \frac{t^*}{t_*} \en \sigma_1^2(A) \norm{x}_\tau^2.
        \end{align*}

    \noindent Similarly we also get $\norm{Ax}^2 \leq \frac{t^*}{t_*} \tau_1^2(A) \norm{x}^2$. These two inequalities give us
        \begin{equation}\label{eq:sigma_tau_1}
             \frac{t_*}{t^*} \en \sigma_1^2(A) \leq \tau^2_1(A) \leq \frac{t^*}{t_*}\en \sigma_1^2(A).
        \end{equation}

    In order to prove the inequality for other singular values we rely on the following characterisation of the $(i + 1)$-st singular value.
        \begin{equation*}
            \sigma_{i + 1}(A) = \min_{B \hspace{1mm} : \hspace{1mm} \rank(B) = i} \sigma_1(A - B).
        \end{equation*}

    The above holds for all inner products. Pick any $i \geq 1$. Let $B$ be the matrix of rank $i$ that minimizes the above for $\tau_{i + 1}(A)$. Then using equation \ref{eq:sigma_tau_1} we get
        \begin{equation*}
            \sigma_{i + 1}^2(A) \leq \sigma_1^2(A - B)
            \leq \frac{t_*}{t^*}  \en \tau^2_1(A - B) 
            = \tau_{i + 1}^2(A).
        \end{equation*}

    \noindent A similar strategy obtains the inequality $\tau_i^2(A) \leq  \frac{t^*}{t_*}\sigma_i^2(A)$. This proves the required inequality. \\

    To show tightness it is sufficient to show the tightness of equation \ref{eq:sigma_tau_1}. Without loss of generality assume that $t^* = t_1$ and $t_* = t_l$. Consider a matrix $A$ with $\sigma_1(A) = 1$ which maps the canonical basis vector $e_1$ to $e_l$. Then 
        \begin{equation*}
            \frac{\norm{Ae_1}^2_\tau}{\norm{e_1}^2_\tau}
            = \frac{t_1}{t_l} \implies \tau_1^2(A) \geq  \frac{t_1}{t_l} =  \frac{t_1}{t_l} \sigma_1^2(A).
        \end{equation*}

    \noindent Similarly considering a matrix $A$ with $\tau_1(A) = 1$ which maps the vector $\sqrt{t_l} e_l$ to $\sqrt{t_1} e_1$ we get
        \begin{equation*}
             \frac{\norm{Ae_l}^2}{\norm{e_l}^2}
            = \frac{t_1}{t_l} \implies \sigma_1^2(A) \geq  \frac{t_1}{t_l} =  \frac{t_1}{t_l} \tau_1^2(A).
        \end{equation*}

    \noindent This completes the proof of the lemma.
\end{proof}
    


Now, we will adopt the special inner product to the context of our polynomial spaces. Let $W = [\vecw_{1} \dots \vecw_{m}]$ and $\Omega = [\veco_{1} \dots \veco_{t}]$ be two matrices with orthonormal columns with $\vecw_{i}, \veco_{i} \in \R^n$. The set of variables $\vecy$ and $\vecz$ are defined as as $\vecy = W^T \cdot \vecx$ and $\vecz = \Omega^T \cdot \vecx $. Further,
    \begin{align}
        \nonumber
        U_1 =  \R[\vecy]^{ = d + 1}, \qquad V_1 &=  \R[\vecy]^{ = d}, \\
        U_2 =  \R[\vecz]^{ = d + 1}, \qquad V_2 &=  \R[\vecz]^{ = d}.\label{def:sc_spaces}
    \end{align}
    
\noindent And, for $i\in [n]$, the collections $\{L_{i1}\}$ and $\{L_{i2}\}$ are scaled partial derivatives (with respect to $\vecx$) on the spaces $U_1$ and $U_2$ respectively. That is
    \begin{equation*}
        L_{i1} =  \restr{\frac{\D_i}{d + 1}}{U_1} \text{ and } L_{i2} = \restr{\frac{\D_i}{d + 1}}{U_2}.
    \end{equation*}

\noindent Imitating the definition \ref{def:inp_tau_basic} we define the special inner product $\IPT{\cdot}{\cdot}$ on a space of linear operators by scaling the standard inner product by the dimension of the domain. For example, for $E_1, E_2$ in $\Lin(V_2, V_1)$, the inner product is given by
    \begin{equation}\label{def:inp_tau}
            \IPT{E_1}{E_2} = \frac{\IPP{E_1}{E_2}}{\dim V_2} =  \frac{\IPP{E_1}{E_2}}{t^{(d)}}.
    \end{equation}

\noindent For this inner product on a sum of spaces of linear operators, we scale each component as above. For example, for the space $\Lin(U_2, U_1) \times \Lin(V_2, V_1)$, we have
    \begin{equation}\label{eq:inp_on_sum_spaces}
        \IPT{(D_1, E_1)}{(D_2, E_2)} = \IPT{D_1}{D_2} + \IPT{E_1}{E_2} =   \frac{\IPP{D_1}{D_2}}{t^{(d + 1)}} + \frac{\IPP{E_1}{E_2}}{t^{(d)}}.
    \end{equation}

With the above inner product in place, we can define a new operator $\Phi$ which we will see is closely connected to the adjoint algebra operator. Let 
    \begin{equation}\label{def:phi_map}
        \Phi : \Lin(V_2, V_1) \rightarrow \Lin(U_2, U_1) \text{ with } \Phi(E) \eqdef \sum_{i = 1}^n L_{i1}^T E L_{i2} .
    \end{equation}
\noindent We can calculate the adjoint of this map with respect to the inner product defined above as done below:
    \begin{align*}
          \Tr \Phi^*(D)^T E   &= \dim V_2 \cdot \IPT{\Phi^*(D)}{E} \\
                              &= \dim V_2 \cdot \IPT{D}{\Phi(E)}   \\
                              &= \frac{\dim V_2}{\dim U_2} \cdot \Tr \sum_{i = 1}^n \left( L_{i1} D L_{i2}^T \right)^T E.
    \end{align*}
\noindent This yields that 
    \begin{equation} \label{def:adj_of_phi}
        \Phi^*(D) = \displaystyle \frac{t^{(d)}}{t^{(d + 1)}} \sum_{i = 1}^n L_{i1} D L_{i2}^T.
    \end{equation}

Now recall that the adjoint algebra map $\adjmap: \Lin(U_2, U_1) \times \Lin(V_2, V_1) \rightarrow \Lin(U_2, V_1)^n$ is given by
    $$\adjmap(D, E) = \left(L_{11} D - E L_{12}, \dots, L_{n1} D - E L_{n2} \right).$$

\noindent Quickly note here that when $W = \Omega$, that is, $L_{i1} = L_{i2}$ for all $i\in [n]$, we have $\adjmap(I, I) = 0$. Thus, the bottom singular value of $\adjmap$ is $0$ in this case. Going forward we analyze the other singular values by looking at the eigenvalues of $\adjmap^T \adjmap$. We note that
\begin{align*}
            \fnorm{\adjmap(D, E)}^2 &= \sum_{i = 1}^{n} \fnorm{L_{i1} D - E L_{i2}}^2                \\
                                     &= \Tr D^T \left[\sum_{i = 1}^n L_{i1}^T L_{i1}\right] D  +
                                        \Tr E^T E \left[\sum_{i = 1}^n L_{i2} L_{i2}^T \right] \\& \enspace -
                                        \Tr D^T \left[\sum_{i = 1}^n L_{i1}^T E L_{i2} \right] -
                                        \Tr E^T \left[\sum_{i = 1}^n L_{i1} D L_{i2}^T \right] \\
                                     &= t^{(d + 1)}
                                        \IPP{\Vector{D, E}}{\begin{bmatrix}
                                                                \frac{I}{t^{(d + 1)}} & 0 \\
                                                                0                     & \frac{I}{t^{(d)}}
                                                            \end{bmatrix}
                                                            \begin{bmatrix} 
                                                                 I & -\Phi \\
                                                                - \Phi^*  &  I
                                                            \end{bmatrix}\Vector{D, E}},
    \end{align*}
\noindent where the map $\Phi$ is as defined in definition \ref{def:phi_map} and its adjoint $\Phi^*$ is taken according to the special inner product as defined in \ref{def:adj_of_phi}. Renaming $(D, E)$ to $X$ and giving appropriate names to the matrices appearing in the above equality, we obtain
    \begin{align*}
        \IPP{X}{\adjmap^T \adjmap X} = t^{(d + 1)} \IPP{X}{\Upsilon \Psi X},
    \end{align*}

\noindent where $\Upsilon$ is the diagonal matrix. Note here that the matrix $\Upsilon \Psi$ is Hermitian in the standard inner product. This gives us
    \begin{align*}
        \adjmap^T \adjmap = t^{(d + 1)}\Upsilon \Psi.
    \end{align*}

\noindent Let $\tau_i(\Phi)$ and $\tau_{-i}(\Phi)$ be used to index the singular values of $\Phi$ in decreasing and increasing orders respectively. From here we obtain the crucial relation between $\sigma_{-i}(\adjmap)$ and $\tau_i(\Phi)$ as follows:
    \begin{align} \nonumber
        \sigma^2_{-i}(\adjmap) &= \sigma_{-i}(\adjmap^T \adjmap)  \\ \nonumber
                               &\geq t^{(d + 1)} \en \sigma_{-1}(\Upsilon) \en \sigma_{-i}(\Psi) \\ \nonumber
                               &\geq \inparen{\tfrac{t^{(d)}}{t^{(d + 1)}}}^{\sfrac{1}{2}} \en \tau_{-i}(\Psi) \\ 
      \implies  \sigma^2_{-i}(\adjmap) &\geq \sqrt{\tfrac{d + 1}{t + d}} \inparen{1 - \tau_i(\Phi)}. \label{eq:adjmap_to_phi_map}
    \end{align}

\noindent Here we have used lemma \ref{lemma:inequality_singular_values} and the following facts. For any two matrices $A, B$ we have $\sigma_i(AB) \geq \sigma_{-1}(A) \sigma_i(B)$. And the eigenvalues of the Hermitian matrix $\Psi$ (with respect to the special inner product)is the set $\{1 \pm \tau_i(\Phi)\}$.
\subsubsection{Singular Values of the \texorpdfstring{$\Phi$}{Phi} Operator}
We immediately delve into finding $\tau_i(\Phi)$ following the definition of $\IPT{\cdot}{\cdot}$ in \ref{def:inp_tau} and $\Phi$ in \ref{def:phi_map}. Recall that $U_1, V_1$ and $U_2, V_2$ are homogeneous polynomial spaces with respect to the variables $\vecy$ and $\vecz$ defined by the matrices $W$ and $\Omega$ respectively, as given in \ref{def:sc_spaces}. \\

The map $\Phi$ is defined on $\Lin(V_2, V_1)$ which are linear maps on spaces of homogeneous polynomials of degree $d$. To emphasize this dependence on the degree we call $\Phi$ as $\Phi_d$, and subsequently, for different $d$s the domain and the co-domain of the map $\Phi_d$ changes. In this spirit, to distinguish the scaled derivatives of larger and smaller spaces we imitate the notation in \ref{eq:chainrule:derivatives_of_shifts} as follows:
    \begin{align*}
        \upL_{i1} = \restr{\frac{\partial_i}{d + 1}}{U_1},\enspace \upL_{i2} = \restr{\frac{\partial_i}{d + 1}}{U_2}, \enspace 
        \downL_{i1} = \restr{\frac{\partial_i}{d}}{V_1}, \enspace \downL_{i2} = \restr{\frac{\partial_i}{d}}{V_2}.
    \end{align*}

\noindent Therefore the above scaled derivatives define different $\Phi$ maps as follows:
    \begin{equation*}
        \Phi_d(E) = \sum_{i = 1}^n \upL_{i1}^T E \upL_{i2} \text{  and  } \Phi_{d - 1}(E) = \sum_{i = 1}^n \downL_{i1}^T E \downL_{i2}.
    \end{equation*}
\noindent  These identifications allow an inductive approach to calculate the relevant singular values, inspired by the inductive argument in recent works on analyzing eigenvalues for random walks on simplicial complexes (e.g. \cite{anari2019log}). 

\begin{lemma}\label{lemma:singvals_phi_d}
Let $W^T \Omega =\diag({\cos \theta_1, \dots, \cos \theta_t})$ (appended with $0$s if necessary) where $\theta_1, \dots, \theta_t$ are the canonical angles between $\inangle{W}$ and $\inangle{\Omega}$ according to remark \ref{remark:cs-decomposition}. Let $f_d\inparen{\inangle{W}, \inangle{\Omega}}$ $= \frac{d + 1}{t} \left[\sum_{k = 1}^t \sin^2 \theta_k + d \sin^2 \theta_{\min}\right]$ and $g(W, \Omega) = \frac{1}{t} \sum_{k = 1}^t \cos^2 \theta_k$. Then the top $m\cdot t$ singular values of $\Phi_d$ (with respect to the inner product defined in \ref{def:inp_tau}) are given by
    $$\tau_i^2(\Phi_d) \leq 1 - \frac{1}{d + 1} \frac{t}{t + d}\left[g + f_d - \tau_i^2(\Phi_0)\right].$$

\noindent Further, if $W = \Omega$, then the above inequality is an equality with $g = 1$ and $f_d = 0$. 
\end{lemma}

\begin{proof}
    We proceed with induction on $d$. Define operators $\phiupdown_d$ and $\phidownup_d$ on $\Lin\inparen{V_2,V_1}$ as follows: 
        $$ \phidownup_d = \Phi_{d-1}\Phi^*_{d-1} \text{ and }\phiupdown_d = \Phi^*_{d}\Phi_{d}.$$

    \noindent Then, using the definition of $\Phi^*$ from \ref{def:adj_of_phi} we can write $\phidownup_d$ and $\phiupdown_d$ as follows:
    \begin{align*}
        \phidownup_d(D) = \frac{d}{t+d-1} \sum_{i,j = 1}^n \downL_{i1}^T \downL_{j1} D \downL_{j2}^T  \downL_{i2} \enspace\text{and} \enspace \phiupdown_d(D) = \frac{d+1}{t+d} \sum_{i,j = 1}^n \upL_{i1} \upL_{j1}^T D \upL_{j2} \upL_{i2}^T.
    \end{align*}
    
    \noindent Recalling the relation for converting derivatives of shifts to shifts of derivatives from \ref{eq:chainrule:derivatives_of_shifts}, we obtain
    \begin{align}
        \nonumber
        (d+1)(t+d) \phiupdown_d(D) 
        &= \sum_{i,j = 1}^n \inparen{r_{ij} I + d \downL_{j1}^T \downL_{i1}} D 
                                                 \inparen{q_{ij} I + d \downL_{i2}^T \downL_{j2}} \\
        \nonumber
        &= \sum_{i,j=1}^n \left[ \inparen{ r_{ij}\cdot q_{ij}} D 
                                       +  d \cdot q_{ij}  \inparen{\downL_{j1}^T \downL_{i1}} D 
                                       + d \cdot r_{ij} \cdot D \inparen{ \downL_{i2}^T \downL_{j2}} \right.\\
                                       \nonumber
                                       & \quad \left. + \enspace d^2 \inparen{\downL_{j1}^T \downL_{i1} D \downL_{i2}^T  \downL_{j2} }  \right]  \\ 
        &= \inparen{t g \cdot D +  G D  +  D H}
          + {d(t+d-1)}\phidownup_d(D), \label{eq:phi_d_to_phi_d-1}
    \end{align} 
    where the involved quantities are defined as follows:
    \begin{align*}
        r_{ij} = \sum_{k\in[m]} w_{ki} w_{kj}, \enspace q_{ij} = \sum_{k\in[t]} \omega_{ki} \omega_{kj}, \enspace
        G =  d \sum_{i,j = 1}^n q_{ij} \cdot\downL_{j1}^T \downL_{i1} \text{ and }
          H = d \sum_{i,j = 1}^n r_{ij} \cdot  \downL_{i2}^T \downL_{j2},
    \end{align*}
    which further yields the following relation used above:
        $$\sum_{i,j = 1}^n r_{ij} q_{ij} = \sum_{k = 1}^m \sum_{i, j = 1}^n (w_{ki} \omega_{ki}) (w_{kj} \omega_{kj})
                                             = \sum_{k = 1}^m \IPP{\vecw_k}{\veco_k}^2
                                             = \sum_{k = 1}^t \cos^2(\theta_k)
                                             = tg.$$   

    Now, we express the matrices $G$ and $H$ in terms of the canonical angles $\theta_k$s. Use the matrix representation of $\downL_{j1}$ from \ref{def:derivative_matrix} to obtain
     \begin{align*}
            \left(\sum_{j = 1}^n \omega_{kj} \downL_{j1}\right)_{\vecbeta \vecalpha} 
            &= \sum_{j = 1}^n  \omega_{kj} w_{lj} \sqrt{\frac{\vecalpha_l}{d}}  \\&= 
            \begin{cases}
            \IPP{\veco_k}{\vecw_l} \sqrt{\frac{\vecalpha_l}{d}} 
                \text{ only when $\exists l$ such that $\vecalpha = \vecbeta \cup l,$ and $0$ otherwise, }  \\
            \cos \theta_k \sqrt{\frac{\vecalpha_k}{d}} 
             \text{ only when $\vecalpha = \vecbeta \cup k,$ and $0$ otherwise.}
            \end{cases}
        \end{align*}

    \noindent From here, observe that $G$ can also be written as
        \begin{align*}
         G = d \sum_{i,j = 1}^n  q_{ij} \cdot \downL_{j1}^T \downL_{i1} 
              = d \sum_{k=1}^m \left(\sum_{j = 1}^n \omega_{kj} \downL_{j1}\right)^T \left(\sum_{j = 1}^n \omega_{kj} \downL_{j1}\right).
        \end{align*}

    \noindent Plugging the appropriate quantities calculated above, and calculating similarly for $H$, we get
        \begin{equation*}
             G = \diag\left\{\sum_{k = 1}^m \vecalpha_k \cos^2 \theta_k\right\}_{\vecalpha} \text{ and }
             H = \diag\left\{\sum_{k = 1}^t \vecgamma_k \cos^2 \theta_k\right\}_{\vecgamma}
        \end{equation*}

    \noindent where $\vecalpha$s and $\vecbeta$s are multi-indices for the basis of $d$-degree homogeneous polynomials on $m$ and $t$ variables respectively. \\

    Now, computing the $i$-th singular value from the equation \ref{eq:phi_d_to_phi_d-1}, we get the following recurrence relation:
        \begin{align*}
            (d + 1)(t + d)\cdot \tau_i(\phiupdown_d) &\leq tg + \norm{G} + \norm{H} + d(t + d - 1) \cdot \tau_i(\phidownup_d) \\
            \implies  (d + 1)(t + d)\cdot \tau_i^2(\Phi_d) &\leq tg + 2d \cos^2 \theta_{\min} + d(t + d - 1) \cdot \tau^2_i(\Phi_{d - 1}).
        \end{align*}

    \noindent Solving the above recurrence yields the required result. In the case when $W = \Omega$, we have $\theta_k = 0$ for all $k$, implying $G = H = dI$ and $g = 1$. Now, as $\phiupdown_d$ and $\phiupdown_{d - 1}$ differ by a scaled identity, they are both diagonalizable in the same basis, and hence we get the following exact relation:
    \begin{equation*}
         (t+d)(d+1) \cdot \phiupdown_d = (t + 2d) \cdot I + (t + d - 1) \cdot \phidownup_d.
    \end{equation*}
    
    \noindent Computing the singular values from the above recurrence gives the required result for this case. This completes the proof of the lemma. 
\end{proof}

\noindent {\bf Base case.} Note that $\Phi_0$ is a map defined as $\Phi_0 : \R \rightarrow \R^{t \times m}$ such that
    \begin{align*}
        \Phi_0(c)\cdot p &= c \sum_{i = 1}^n L_{i1}^T L_{i2} \cdot p \\
                         &= c \sum_{i = 1}^n \sum_{k = 1}^m \sum_{j = 1}^t w_{ji} \omega_{ki} p_k y_j \\
                         &= c \sum_{j = 1}^m \left[\sum_{k = 1}^t \IPP{\vecw_j}{\veco_k} p_k \right] y_j \\
                         &= c\cdot W^T \Omega \cdot p
    \end{align*}

\noindent Therefore, with respect to our scaled inner product, we obtain
    \begin{equation}
         \tau_1^2(\Phi_0) = \frac{1}{t} \fnorm{W^T \Omega}^2 = \frac{1}{t} \sum_{k \in [t]} \cos^2 \theta_k = g.
    \end{equation}

\noindent Now, as $\Phi_0$ is defined on $\R$ it admits at most $1$ non-zero singular value. This shows that $\tau_i(\Phi_0) = 0$ for all $i \geq 2$. Clubbing the base case with the above lemma gives us the following remarks.

\begin{remark}\label{remark:sing_vals_of_phi}
    We get the following singular values of the $\Phi$ operator in the particular cases stated below:
    \begin{enumerate}
        \item When $W = \Omega$, we have $\tau_1(\Phi_0) = 1$. This gives $\tau_1(\Phi) = 1$.
        \item Also, when $W = \Omega$, plugging $\tau_2(\Phi_0) = 0$ gives $\tau_2^2(\Phi) =  1 - \frac{1}{d+1}\frac{t}{t+d}$.
        \item When $W \neq \Omega$, we plug $\tau^2_1(\Phi_0) = g$. This gives us
        $\tau_1^2(\Phi) = 1 - \frac{1}{d+1}\frac{t}{t+d}\cdot f_d$.
    \end{enumerate}
\end{remark}


\subsection{Adjoint Algebra Operator for Subspace Clustering}

Recall the setting of subspace clustering problem in section \ref{sec:scrProofs} where we have a set of $N$ points $A =\{ \veca_1, \dots, \veca_N \}\in \R^n$ clusterable to the subspaces $\inangle{A_1}, \dots, \inangle{A_s}$. These subspaces $\inangle{A_i}$ of dimension $t_i$ are also expressed as $\inangle{W_i}$, where the columns of $W_i$ are orthonormal. Let $\vecU$ and $\vecV$ be such that
    \begin{align*}
        \vecU = (U_1, \dots, U_s) &= \left( \inangle{\tensored{A_1}{d + 1}}, \dots,  \inangle{\tensored{A_1}{d + 1}} \right) \\
        \text{ and } \vecV = (V_1, \dots, V_s) &= \left( \inangle{\tensored{A_1}{d}}, \dots,  \inangle{\tensored{A_1}{d}} \right).
    \end{align*}

\noindent The section \ref{sec:scOverview} tells us that the tuples $\vecU$ and $\vecV$ along with the differential operators $\opB = \{\partial_{i}, \dots, \partial_{n}\}$  form an instance of the vector space decomposition problem. Following  the conventions set in the equations in \ref{def:m^(d)}, the spaces $U_i$ and $V_i$ have dimensions $t_i^{(d + 1)}$ and $t_i^{(d)}$ respectively. Further, we know that there exist bases $P$ and $Q$, which are $\vecU$-associated and $\vecV$-associated matrices (as defined in section \ref{sec:prelims}), such that in these bases the scaled derivatives, which are the operators we have studied in the above section, are block diagonal operators. That is  
            $$\opL = \{L_1, \dots, L_n\} \text{ with }  L_i = \diag\left(L_{i1}, \dots, L_{is} \right) \text{ where } L_{ij} = \restr{\frac{\partial_i}{d + 1}}{\vecU_j}.$$

\noindent Thus, the adjoint algebra operator $\adjmap = \adjmap(\opB)$ in \ref{thm:sc} can be expressed as follows:
    \begin{align*}
        \adjmap(D, E) &= (d + 1) \cdot \left(Q L_i P^{-1} D - E Q L_i P^{-1} \right)_{i \in [n]} \\
                    &= (d + 1) \cdot Q \cdot\left(L_i (P^{-1} D P) - (Q^{-1} E Q) L_i \right)_{i \in [n]} \cdot P^{-1}.
    \end{align*}

\noindent The above establishes the following relationship between $\adjmap$ and the adjoint algebra operator $\adjmap_{\opL}$ with respect to the operator collection $\opL$:
    $$\adjmap= (d + 1) \cdot \Gamma_2 \circ \adjmap_{\opL} \circ \Gamma_1 $$ where 
       $$\Gamma_1(D, E) = (P^{-1} D P, Q^{-1} E Q) \text{  and  } \Gamma_2(X_1, \dots, X_n) = Q \cdot (X_1, \dots, X_n) \cdot P^{-1}.$$

\noindent This immediately yields that
    \begin{equation}\label{eq:}
        \sigma_{-(s + 1)}(\adjmap) \geq (d + 1) \cdot \sigma_{-(s + 1)}\left(\adjmap_{\opL} \right) \cdot \sigma_{-1}(\Gamma_1) \cdot \sigma_{-1}(\Gamma_2).
        \end{equation}


\noindent One can compute $\sigma_{-1}(\Gamma_1)$ and $\sigma_{-1}(\Gamma_2)$ from the definition of $\Gamma_1$ and $\Gamma_2$ to obtain
    \begin{equation}\label{def:sigma_perp}
        \sigma_{-1}(\Gamma_1) \cdot \sigma_{-1}(\Gamma_2) \geq \frac{1}{\kappa^2(\vecU, \vecV)} \enspace \text{ where } \enspace \kappa(\vecU, \vecV) \eqdef \frac{\max\{\sigma_1(\vecU), \sigma_1(\vecV)\}}{\min\{\sigma_{-1}(\vecU), \sigma_{-1}(\vecV)\}}.
    \end{equation}

We now focus on computing the quantity $\sigma_{-(s + 1)}(\adjmap_\opL)$. We exploit the block diagonal structure of the matrices in the collection $\opL$. Treating $D$ and $E$ also as block matrices $[D_{jk}]$ and $[E_{jk}]$, we note that the Frobenius norm of $\adjmap_\opL$ can be separated across the blocks as follows:
    \begin{align*}
        \fnorm{\adjmap_\opL(D, E)}^2  
        = \sum_{i = 1}^{n} \fnorm{L_i D - E L_i}^2
        = \sum_{j, k = 1}^{s} \sum_{i = 1}^n \fnorm{L_{ij} D_{jk} - E_{jk} L_{ik}}^2 
        \eqdef \sum_{j, k = 1}^{s} \fnorm{\adjmap_{jk}(D_{jk}, E_{jk})}^2,
    \end{align*}

\noindent where each $\adjmap_{jk}:\Lin\left(U_k, U_j\right) \times \Lin\left(V_k, V_j\right) \rightarrow \Lin\left(U_k, V_j\right)^t$ is a {\bf block adjoint algebra operator} defined as follows:
    $$\adjmap_{jk}(D_{jk}, E_{jk}) = \left(L_{1j} D_{jk} - E_{jk} L_{1k}, \dots, L_{nj} D_{jk} - E_{jk} L_{nk} \right).$$

\noindent Therefore we note that the adjoint algebra map $\adjmap_\opL$ has the following block structure:
    \begin{equation}\label{eq:adjalgmap_block_structure}
        \adjmap_\opL^T \adjmap_\opL = \bigoplus_{j, k = 1}^s \adjmap_{jk}^T \adjmap_{jk},
    \end{equation}

\noindent  where the notation above implies that the map $\adjmap^T\adjmap$ acts separately as $\adjmap_{jk}^T \adjmap_{jk}$ on different independent subspaces of the domain. This immediately implies that the singular values of the adjoint algebra operator $\adjmap$ are the singular values of the block adjoint algebra operators $\adjmap_{jk}$ collected together.\\

Now recall that equation \ref{eq:adjmap_to_phi_map} gives
    $$\sigma_{-i}^2(\adjmap_{jk}) \geq \sqrt{\tfrac{d + 1}{t + d}} (1 - \tau_i(\Phi_{jk})), $$

\noindent where the $\Phi_{jk}$ maps are related to scaled derivatives on polynomial spaces on variables $W_j^T \vecx$ and $W_k^T \vecx$ respectively. Now, as $\sigma_{-1}(\adjmap_{jj}) = 0$ for all $j \in [s]$ with $(I, I)$ in its null space, we conclude that $\sigma_{-(s + 1)}(\adjmap_{\opL})$ must be $\sigma_{-2}(\adjmap_{jj})$ for some $j$, or $\sigma_{-1}(\adjmap_{jk})$ for some $j \neq k$. Thus, combining all the components above with remark \ref{remark:sing_vals_of_phi} we establish theorem \ref{thm:adjoint_algebra_sc_singular_values}. \\

Finally, we end with a remark on how the function $f_d$ helps relate $\sigma_{-(s + 1)}(\adjmap)$ to the geometry of the underlying input subspaces. 
\begin{remark} \label{remark:f_d_geometry}
Recall for given subspaces $\inangle{A_i}$ and $\inangle{A_j}$ we have $$f_d(\inangle{A_i}, \inangle{A_j} ) = \tfrac{d + 1}{t} \left[\sum_{j = 1}^t \sin^2 \theta_k + d \sin^2 \theta_{\min}\right]$$ where $\theta_k$s are the canonical angles between $\inangle{A_i}$ and $\inangle{A_j}$.
    \begin{enumerate} 
        \item For any pair of subspaces note that
            $$0 \leq t \cdot f_d  \leq (t + d) (d + 1).$$ 
        Further, $f_d = 0$ if and only if $\theta_k = 0$ for all $k$, implying $\inangle{A_i} \subseteq \inangle{A_j}$. Thus, excluding the degenerate cases when one of the subspaces is contained in another subspace, we always have $\sigma_{s + 1}(\adjmap) > 0$, thereby ensuring that the Adjoint algebra has dimension exactly equal to $s$. 

        \item Observe that we can write
            $f_d\geq \frac{(d + 1)(t + d)}{t} \sin^2 \theta_{\min}$. Substituting this in $\sigma_\offdiag$, we get $$\sigma_\offdiag \geq 1 - \cos \theta_{\min}, $$ where $\theta_{\min}$ is the smallest canonical angle between any pair of distinct subspaces. This is a good heuristic for the contribution of the off-diagonal blocks of $\adjmap$ to its $(s + 1)$-st smallest singular value in the case when the subspaces have trivial pairwise intersection.

        \item When $f_d > 1$ for all pairs of distinct subspaces $\inangle{A_i}$ and $\inangle{A_j}$, then $\sigma_\diag$ is smaller than each term involved in the minimum in $\sigma_\offdiag$. This gives 
            $$\sigma_{-(s + 1)}(\adjmap)^2 \geq \tfrac{(d + 1)^2}{\kappa^4(\vecU, \vecV)} \cdot \sigma_\diag.$$ 
        \noindent This is a bound that doesn't depend on geometry of the subspaces.
        This also shows that for any non-degenerate instance, one can choose $d$ large enough so that $\sigma_{-(s + 1)}^2(\adjmap)$ has the form as above.
    \end{enumerate}    
\end{remark} 



\end{document}